\documentclass[11pt,a4paper]{article}
\usepackage{amssymb}
\usepackage{rotating}
\usepackage{amsfonts}
\usepackage{amsmath}
\DeclareMathOperator*{\argmax}{arg\,max}

\DeclareMathOperator*{\esssup}{ess\,sup}
\DeclareMathOperator*{\essinf}{ess\,inf}
\usepackage{amsthm}
\usepackage{caption}
\usepackage{float}
\usepackage{graphicx}
\usepackage[colorlinks=true, urlcolor=blue, pdfborder={0 0 0}]{hyperref}
\usepackage{enumerate}
\usepackage{multicol}
\usepackage{array}
\usepackage{bbm}
\usepackage{mathtools}
\usepackage{leftidx}
\usepackage{floatflt,epsfig}

\newtheorem{theorem}{Theorem}
\newtheorem{proposition}[theorem]{Proposition}
\newtheorem{lemma}[theorem]{Lemma}
\newtheorem*{lemma*}{Lemma}
\newtheorem*{assumption*}{Assumption}
\newcounter{assumptionc}
\newtheorem{assumption}[assumptionc]{Assumption}

\newtheorem{definition}[theorem]{Definition}

\newtheorem{remark}[theorem]{Remark}
\newtheorem{res}[theorem]{Result}
\numberwithin{equation}{section}
\numberwithin{theorem}{section}
\newtheorem{ex}[theorem]{Example}

\newtheorem{subassumption}{Assumption}[assumptionc]

\newenvironment{assumption+}
 {\ifnum\value{subassumption}=0 \stepcounter{assumptionc}\fi\subassumption}
 {\endsubassumption}

\def\sign{\mbox{sign}}

\newcommand\cL{{\cal L}}

\newcommand\cD{{\cal D}}

\def\text#1{\hbox{#1}}

\def\v{{\bf v}}

\def\build #1_#2{\mathrel{\mathop{\kern 0pt #1}\limits_\zs{#2}}}
\newcommand{\zs}[1]{{\mathchoice{#1}{#1}{\lower.25ex\hbox{$\scriptstyle#1$}}
{\lower0.25ex\hbox{$\scriptscriptstyle#1$}}}}

\numberwithin{equation}{section}

\usepackage[top=1in, bottom=1in, left=1.in, right=0.9in]{geometry}

\newtheorem*{example*}{Example}

\newcommand{\indep}{\perp \!\!\! \perp}
\newcommand{\PP}{\mathbb{P}}
\newcommand{\E}{\mathbb{E}}
\newcommand{\R}{\mathbb{R}}
\newcommand{\N}{\mathbb{N}}
\newcommand{\T}{\mathcal{T}}
\newcommand{\normal}{\mathcal{N}}
\renewcommand{\S}{\mathbb{S}}
\newcommand{\bS}{\bf{S}}
\renewcommand{\L}{\mathcal{L}}
\renewcommand{\cL}{\mathbb{L}}

\renewcommand{\cD}{\mathbb{D}}

\newcommand{\Var}{\text{Var}}
\newcommand{\corr}{\text{Corr}}
\newcommand{\1}{\mathbbm{1}}
\newcommand{\sgn}{\text{sign}}
\newcommand{\diff}{\mathrm{d}}
\newcommand\norm[1]{\lVert#1\rVert}
\newcommand\scprod[1]{\left\langle#1\right\rangle}

\makeatletter
\newenvironment{myproof}[1][\proofname]{%
  \par\pushQED{\qed}\normalfont%
  \topsep6\p@\@plus6\p@\relax
  \trivlist\item[\hskip\labelsep\bfseries#1\@addpunct{.}]%
  \ignorespaces
}{%
  \popQED\endtrivlist\@endpefalse
}
\makeatother

\def\R{\mathbb{R}}

\usepackage{accents}

\allowdisplaybreaks

\newcommand{\Sum}{\textstyle \sum}
\newcommand{\Int}{\textstyle \int}
\newcommand{\textfrac}{\dfrac}

\begin{document}

\title{Time-Consistent Asset Allocation for Risk Measures in a L{\'e}vy Market\footnote{We thank Roger Laeven and the participants of the IME conference 2023 and the IFA-IVW workshop at Ulm University for their comments and feedback where the usual caveat applies.}}
\author{Felix Fie{\ss}inger\footnote{University of Ulm, Institute of Insurance Science and Institute of Mathematical Finance, Ulm, Germany. Email: felix.fiessinger@uni-ulm.de}  \, and Mitja Stadje\footnote{University of Ulm, Institute of Insurance Science and Institute of Mathematical Finance, Ulm, Germany. Email: mitja.stadje@uni-ulm.de}}
\date{\today}
\maketitle
\begin{abstract}
	Focusing on gains \& losses relative to a risk-free benchmark instead of terminal wealth, we consider an asset allocation problem to maximize time-consistently a mean-risk reward function with a general risk measure which is i) law-invariant, ii) cash- or shift-invariant, and iii) positively homogeneous, and possibly plugged into a general function. Examples include (relative) Value at Risk, coherent risk measures, variance, and generalized deviation risk measures. We model the market via a generalized version of the multi-dimensional Black-Scholes model using $\alpha$-stable L{\'e}vy processes and give supplementary results for the classical Black-Scholes model. The optimal solution to this problem is a Nash subgame equilibrium given by the solution of an extended Hamilton-Jacobi-Bellman equation. Moreover, we show that the optimal solution is deterministic under appropriate assumptions.
\end{abstract}



%C61:Optimizaiotn Techniques, C72:Noncooperative Games, C73:Stochastic and Dynamic Games, D52:Incomplete Markets, G11:Portfolio Choices/Investment Decisions
%91B51: Dynamic stochastic general equilibrium theory, 93E20: optimal stochastic control, 49J20: Existence theories for optimal control problems involving partial differential equations, 49L99: iwas mit Hamilton-Jacobi, 35Q91: PDEs in connection with game theory, economics, social and behavioral sciences, 46B09: Probabilistic methods in Banach space theory, 60G52: Stable stochastic processes

%\noindent\textbf{Published in:} European Journal of Operational Reasearch (DOI: \href{https://www.sciencedirect.com/science/article/pii/S0377221724007525?via%3Dihub}{10.1016/j.ejor.2024.09.049})\\
%
%\noindent\textbf{Copyright:} \textcopyright 2024. This manuscript version is made available under the CC-BY 4.0 license \url{https://creativecommons.org/licenses/by/4.0/}.

\section{Introduction}

In this paper, we study an optimal asset allocation problem of the following form
\begin{align} \label{J def intro}
	V(t,x) = \sup\nolimits_{u\in \mathcal{A}_{t,x}} \big\{ \mathbb{E}_{t,x} [\T(X_T^u-x e^{r(T-t)})]  - Risk_{t,x} (X_T^u - x e^{r(T-t)}) \big\},
\end{align}
where $x$ is the wealth at time $t$, $\mathcal{A}_{t,x}$ is the set of all time-consistent, Markovian strategies at time $t$, 
$X^u$ is the wealth process driven by an $\alpha$-stable L{\'e}vy process when the strategy $u$ is followed, $r$ is the discounting rate, $\T\in C^1$ is a continuous, non-decreasing, and concave reward function for gains \& losses. The risk function is given by $Risk_{t,x}(Y) = \lambda_t F(\rho_{t,x} (Y))$, where $\lambda: [0,T] \to \R_{\geq 0}$ is the time-dependent and continuous risk aversion function with values in $(0, \infty)$, where $\rho$ is a law-invariant, positively homogeneous, and translation- or shift-invariant dynamic risk measure and $F\in C^1$ is a non-decreasing convex function. For tractability reasons, we do not allow $\lambda$ to be state-dependent, which means that an investor can change his/her risk aversion only deterministically. Since we only assume that $\rho$ is law-invariant, translation or shift-invariant, and positively homogeneous, our results encompass most standard or alternative risk measures like (relative) Value at Risk, Average Value at Risk, coherent risk measures, standard deviation, or generalized deviation measures. On all these (classes of) risk measures, there is an extensive literature, but for space reasons, we only state Jorion \cite{jorion2007value}, F{\"o}llmer and Schied \cite{Follmer}, or Rockafeller et al. \cite{rockafellar2002deviation}. We emphasize that $\rho$ does not need to be additive, sub-additive, convex, or continuous. By allowing $\rho$ to be plugged in the convex function $F$, our risk part in \eqref{J def intro} also includes examples such as variance or semi-variance. Time-consistency means that the strategy which is optimal at time $t$, say $(u_s)_{t\leq s}$ will also be optimal from later points on. On the other hand, a time-inconsistent investor does not satisfy such a dynamic programming principle and may initiate a strategy at time $t$ (because at time $t$ its whole path attains the maximum in \eqref{J def intro}) knowing fully well that (s)he will deviate from this strategy from a later point on.
Following psychological theories like prospect theory (Kahneman and Tversky \cite{Kahneman}) and benchmark investment theories (see, for instance, Zhao \cite{zhao2007dynamic} or Pirvu and Schulze \cite{pirvu2012multi}), the investor at time $t$ in \eqref{J def intro} is assumed to consider future \emph{changes} in wealth $X_T^u-x e^{r(T-t)}$ (in \$-values at time $T$) instead of the terminal wealth $X_T^u$. We could generalize the benchmark by multiplying and adding deterministic, time-dependent $C^1$-functions (i.e., consider $a_t (X_t-xe^{r(T-t)})+b_t$ instead of $X_t-xe^{r(T-t)}$ in \eqref{J def intro}), which would structurally not change our results in the case of uniformly bounded strategies. Conceptually, this approach is in line with the so-called ``growth-optimal portfolio'', also known as the Kelly criterion in the context of portfolio theory, optimal decision under risk, game theory, information theory, and insurance, see Wei and Xu \cite{wei2021dynamic} and the references therein. When $\mathcal{T}$ is the identity, considering the terminal expected wealth instead of gain \& losses in the reward function in \eqref{J def intro} leads to identical solutions. We further remark that defining risk as \emph{changes} in values between two dates is a rather classical approach of assessing the riskiness of a position, see, for instance, Jorion \cite{jorion2007value}, the way risk is measured in the Basel III \cite{BaselIII} or Solvency II \cite{SolvencyII} regulations for banks and insurance companies, respectively, or the literature on return risk measures; see, for instance, Bellini et al. \cite{bellini2018robust}.

We solve problem \eqref{J def intro} recursively in discrete time and show that as the time grid gets finer, the discrete-time maximizers converge weakly along a subsequence to continuous time. For shift-invariant risk measures, the limit is a Nash subgame equilibrium under certain assumptions in continuous time. This kind of optimization is a non-cooperative interpersonal game in which the investor plays against reincarnations of himself/herself at other times. There is no direct generalization of Nash equilibria in continuous time since a single time point is a null set. Instead, we follow here the generalization of Ekeland and Lazrak \cite{ekeland2006being}, who allowed the investor at time $t$ to form a coalition with all players $s \in [t,t+\varepsilon]$ and derive the equilibrium strategy by taking $\varepsilon \to 0$. We also prove that, in general, a Nash subgame equilibrium in continuous time of \eqref{J def intro} satisfies an extended Hamilton-Jacobi-Bellman (HJB) equation. Moreover, we derive explicit solutions for some special cases. We illustrate the applicability of our approach in a few numerical examples showing that Value at Risk leads to less risky investments for shorter maturities than variance. Furthermore, compared to a pre-commitment investor, we show that a time-consistent investor invests more conservatively, i.e., less in the risky asset. This effect is more pronounced for shorter than for longer maturities and more for Value at Risk based risk measures than for variance.

We model the stock returns by an $\alpha$-stable L{\'e}vy processes under some limitations. In physics, $\alpha$-stable L{\'e}vy processes are also called L{\'e}vy flights and are used to model physical phenomena, see Cont and Tankov \cite[p.99]{Cont}. %In this context, some authors use the terms L{\'e}vy process and L{\'e}vy flight interchangeably. 
The marginal distribution of an $\alpha$-stable L{\'e}vy process are sometimes also called L-stable distributions. Using L-stable distributions to model prices goes back at least to Mandelbrot \cite{mandelbrot1963variation,mandelbrot1967variation}, and Fama \cite{fama1965behavior}. L{\'e}vy processes also play an important role in operations research for modeling asset returns; see, for instance, Kallsen \cite{kallsen2000optimal}, Feng and Linetsky \cite{feng2008pricing}, Kaishev and Dimitrova \cite{kaishev2009dirichlet}, Kardaras \cite{kardaras2009no}, Nowak and Romaniuk \cite{nowak2010computing}, Fu and Yang \cite{fu2012equilibruim}, Laeven and Stadje \cite{laeven2014robust}, Fu et al. \cite{fu2017option}, Ma et al. \cite{ma2021time}, or {\v{C}}ern{\'y} et al. \cite{vcerny2023numeraire}. The induced model is more general than the classical Black-Scholes model and is favorable for several reasons, see Cont and Tankov \cite[pp. 1-16]{Cont}. For example, contrary to L{\'e}vy processes (with jumps), the Black-Scholes model does not admit heavy tails, which are empirically observed in financial markets, does not allow for large, sudden movements of the asset prices, and, furthermore, induces a complete financial market in the sense that all risks can be perfectly hedged, which might not be realistic.

In 1952, Markowitz \cite{markowitz1952utility} was the first to study a (static) asset allocation problem of the form \eqref{J def intro} in \emph{one period} using portfolio returns as basic objects and introducing a trade-off between the expected portfolio return and its risk. In this case, $\rho$ in \eqref{J def intro} becomes the standard deviation and $F(x)=\max\{0,x\}^2$ so that risk is identified with the variance of a portfolio. Up to date, variance is the most prevalent risk measure and the cornerstone of most finance literature. Many alternative risk measures have been proposed to overcome potential shortfalls of variance. For instance, for regulatory purposes in the financial industry, Value at Risk (VaR) has been the most prominent risk measure. VaR is the amount of money to hold in reserve such that the financial institution can cover its losses in say, 99.5 \%, of the possible future developments of the market over a prescribed time horizon. VaR is the standard risk measure for insurance companies and banks, with its use being prescribed by Basel III \cite{BaselIII} or Solvency II \cite{SolvencyII} regulations. Jorion \cite{jorion2007value} also defines the shift-invariant version of the VaR, the so-called relative Value at Risk (RVaR). He states that RVaR for money situations is conceptually more appropriate than VaR since it accounts more suitable for the time value of money. Moreover, he states that RVaR is more conservative and more consistent with definitions of unexpected loss. In the context of portfolio choice, static mean-risk models with various alternative risk measures have been analyzed, see for instance, Campbell et al. \cite{campbell2001optimal}, Rockafeller and Ursayev \cite{rockafellar2002conditional}, Ruszczy{\'n}sky and Vanderbei \cite{ruszczynski2003frontiers}, Alexander and Baptista \cite{alexander2004comparison}, Jin et al. \cite{jin2006note}, G{\"u}lp{\i}nar and Rustem \cite{gulpinar2007worst}, Adam et al. \cite{adam2008spectral}, or He et al. \cite{he2015dynamic}. Additionally, Markowitz's original single-period static framework has been extended dynamically to discrete and to continuous time; see, for instance, Basak and Chabakauri \cite{Basak}, Brandt \cite{brandt2010portfolio}, Fu et al. \cite{fu2010dynamic}, Zeng and Li \cite{zeng2011optimal}, Hu et al. \cite{hu2012time}, Czichowsky \cite{czichowsky2013time}, or Bj{\"o}rk et al. \cite{bjork2017time}.

In a dynamic setting, the issue of time-consistency arises. Most static risk measures like variance or VaR are time-inconsistent when used in a dynamic context; see, for instance, Basak and Chabakauri \cite{Basak}, Cheridito and Stadje \cite{cheridito2009time}, or F{\"o}llmer and Schied \cite{Follmer}. Mathematically, this entails that a dynamic programming principle does not hold, making the problem of optimal portfolio choice difficult and often intractable. Furthermore, conceptually time-consistency seems an essential element of rational decision-making, and lacking it seems unreasonable; see, for instance, Björk and Murgoci \cite{Bjork}, Balter and Pelsser \cite{balter2021quantifying}, or Balter et al. \cite{balter2021time}.\footnote{In the case of using mean-variance in (\ref{J def intro}), recent literature has actually shown paradoxes arising from time-consistency, see van Staden, Dang, and Forsyth \cite{van2018time} and Bensoussan, Wong, and Yam \cite{bensoussan2019paradox}. See also the discussion in Bosserhoff and Stadje \cite{bosserhoff2021time}. Of course, since mean-variance does not respect first-order stochastic dominance, it also lends itself to various other paradoxes.}

If a given preference is time-inconsistent or if the regulator prescribes a time-inconsistent risk measure, there are three different approaches discussed in the literature: the pre-committed investor only maximizes once in the beginning, see, for instance, Li and Ng \cite{li2000optimal}, Zhou and Li \cite{zhou2000continuous}, Lim and Zhou \cite{lim2002mean}, or Xie et al. \cite{xie2008continuous}, the naive investor continuously optimizes, always using his actual preferences, see Basak and Chabakauri \cite{Basak} for references, and the sophisticated investor optimizes, taking into account that (s)he acts time-inconsistent. We use the third approach and ensure time consistency by analyzing an investor who, in the words of Strotz \cite{Strotz}, aims to find the best strategy ``among those that he will actually follow''. This means that we consider an interpersonal, non-cooperative game of the investor with himself/herself. Building on Strotz \cite{Strotz}, Basak and Chabakauri \cite{Basak} in a seminal work gave an equivalent time-consistent solution for optimizing a mean-variance problem using the conditional variance formula. This and the more technical work by Björk and Murgoci \cite{Bjork} which generalized this approach to problems involving general functions with outer and inner expectations, initiated a large stream of publications, see, for instance, Wang and Forsyth \cite{Wang}, Cui et al. \cite{cui2012better}, Czichowsky \cite{czichowsky2013time}, Bj{\"o}rk et al. \cite{bjork2014mean}, Bensoussan et al. \cite{bensoussan2014time}, Chiu and Wong \cite{chiu2015dynamic}, Wu and Chen \cite{wu2015nash}, Vigna \cite{vigna2016time}, Wei and Wang \cite{wei2017time}, Lindensjö \cite{Lindensjo}, Cui et al. \cite{cui2019time}, Bosserhoff and Stadje \cite{bosserhoff2021time}, or Hern{\'a}ndez and Possama{\"\i} \cite{hernandez2023time}. For a discussion and comparison between a pre-committed and a time-consistent investor, see Balter et al. \cite{balter2021time}. 
While for a mean-variance problem, the reward function can naturally be expressed as a functional of expectations, this is generally not the case for a risk measure as ours, which is merely assumed to be law-invariant, positively homogeneous, and translation or shift invariant. Consequently, to the best of our knowledge, results on time-consistent optimization in continuous time have not been obtained in the context of general and possibly quantile-based risk measures. Using stability results for $\alpha$-stable L\'evy processes, we tackle the optimization problem by showing that the optimal solution necessarily is deterministic and can be computed recursively in discrete time. From there, an extension to continuous-time can be shown using the structure of $\alpha$-stable processes. Using time-consistency (leading to the boundedness of the optimal strategies), Alaoglu's theorem, convex structures, Slutky's theorem, and Fatou's generalized lemma, we show that the discrete-time solutions weakly converge to a solution which dominates the discrete-time solutions in mean-risk (given by \eqref{J def intro}). Under specific assumptions, we show that this limit is also a Nash subgame equilibrium using Arzel{\`a}-Ascoli's theorem, with uniform convergence of the control functions. For mean-variance, convergence results have been shown by Czichowsky \cite{czichowsky2013time}.
We emphasize that even for VaR, arguably the most important risk measure used in industry, apart from Forsyth \cite{Forsyth} and Cui et al. \cite{cui2019time}, who do a numerical analysis in discrete time, the time-consistent asset allocation problem has not been studied much.

The paper is structured as follows. In Section \ref{sec: levy}, we recall the definition of $\alpha$-stable L{\'e}vy processes and briefly recall their properties. In Section \ref{sec: def}, we introduce the financial market. Section \ref{sec: disc} treats the discrete-time optimization case, while Section \ref{sec: cont} analyzes the continuous-time problem. In Section \ref{sec: Black-Scholes}, some specific examples are considered, and finally, in Section \ref{sec: numeric}, we present some numerical results.

\section{Basics about $\alpha$-stable L{\'e}vy processes} \label{sec: levy}

Let us describe some important properties of $\alpha$-stable L{\'e}vy processes used in the sequel. The following description is based on \c{C}{\i}nlar \cite[pp.313,314,329-339]{Cinlar}. 

A process $L$ is a L{\'e}vy process if $L$ is an adapted c{\`a}dl{\`a}g process with independent and stationary increments. Let $(\Omega,\mathcal{H},\mathcal{F}=(\mathcal{F}_t)_{t \in \R_{\geq 0}},\PP)$ be a $d$-dimensional stochastic basis, where $\mathcal{F}$ is the filtration generated by a L{\'e}vy process $L$ and let $\alpha \in (0,2]\backslash\{1\}$. We assume that $L$ is $\alpha$-stable meaning that $t^{-\frac{1}{\alpha}} L_t \overset{d}{=} L_1$ for all $t \in (0,\infty)$, whereby we exclude the case $L = 0$ a.s.. If $\alpha=2$, then $L$ is a Brownian Motion, i.e., for any $z>0$ we have $ z L_t \sim \normal (0,z^2 t)$. To highlight this special case, we then denote the process by $W$ instead of $L$. If $\alpha \in (0,2)\backslash\{1\}$, the associated L{\'e}vy measure $\nu$ has infinite measure and is constructed as follows: Fix a $c \in (0,\infty)$ and a probability measure $\tilde{\sigma}$ on the unit \emph{sphere} in $\R^d$, denoted by $\mathbb{S}^d$ (we write $\S$ if $d=1$). Then, for every Borel function $f \geq 0$, the L{\'e}vy measure $\nu$ is given by (see also Proposition 3.15 in Cont and Tankov \cite{Cont}):
\begin{align} \label{eq: levy measure}
	\nu f = \Int_{\R_{\geq 0}} \Int_{\S^d} \tfrac{c}{w^{\alpha+1}} f(wv) \tilde{\sigma} (\diff v) \diff w.
\end{align}
The associated L{\'e}vy process $L$ is called symmetric, if $L \overset{d}{=} -L$ which is equivalent to the symmetry of $\nu$, i.e., $\nu(B)=\nu(-B)$ for all $B \subset \R^d$ Borel. 
Recall that for an $\alpha$-stable L{\'e}vy process with $0<\alpha<2$ only $\theta$-moments with $\theta<\alpha$ exist, and marginal distributions are still $\alpha$-stable. In the case of $\alpha=2$, we have a Brownian Motion, and, of course, all real and exponential moments exist. Multiplying the L{\'e}vy process with a constant preserves $\alpha$-stability. 

\ref{app: Levy} gives a more detailed introduction to L{\'e}vy processes and their stochastic integral.

\section{Model setup and basic definitions} \label{sec: def}

We use the following notations for this paper: We denote the time dependence in the discrete case as an index and in the continuous case as a variable. Moreover, we write $\scprod{\cdot,\cdot}_A$ for the scalar product concerning the symmetric and positive definite matrix $A \in \R^{d\times d}$, i.e., $\scprod{x,y}_A := x^\intercal A y$ for $x,y \in \R^d$, where $\cdot^\intercal$ denotes the transposed vector. If $A$ is the identity matrix, we write $\scprod{\cdot,\cdot}$. In addition, we denote by $a^{i\cdot}$ the $i$'th row of the matrix $A$ and by $a^{\cdot j}$ the $j$'th column of $A$. Furthermore, we denote by $\L^{\alpha} (\Omega,\mathcal{H},\mu)$ the set of all $\mathcal{H}$-measurable and $\alpha$-times integrable functions, i.e., $\L^{\alpha}(\Omega,\mathcal{H},\mu)=\{ f: \Omega \to \R : f \text{ is measurable}, \Int_{\Omega} |f(x)|^{\alpha} \mu(\diff x) < \infty \}$ for the measure space $(\Omega,\mathcal{H},\tilde{\nu} )$. In particular, $\L^0 (\Omega,\mathcal{H},\mu)$ denotes the set of all measurable functions. Often, ``$\mathcal{H}$'' will be omitted in the notation if the context is clear.

\subsection{Model setup} \label{model setup}

We consider a finite investment period, which starts at time $0$ and ends at time $T$. Let $L_t$ be a $d$-dimensional $\alpha$-stable L{\'e}vy process with $\alpha \in (0,2]\backslash\{1\}$ defined with the constant $c \in (0,\infty)$ and the probability measure $\tilde{\sigma}$ (see Equation (\ref{eq: levy measure})). Then, we consider a financial market with $d$ risky assets $S^i$, $i \in \{1,\ldots,d\}$, and one risk free asset $B$. The price dynamics are given by
\begin{align} \label{def: financial market 1}
	\diff S^i_t = \mu^i_t S^i_t \diff t + \Sum_{j=1}^d \sigma^{ij}_t S^i_t \diff L^j_t \text{ and } 
	\diff B_t = r B_t \diff t, \quad i=1,\ldots,d,
\end{align}
where $\mu^i_t$ and $\sigma^{ij}_t$ are deterministic and continuous functions describing an additional drift, the volatility, and the dependency of the risky assets, whereas the constant $r$ denotes the short rate of the risk-free asset. %For a discussion about asset prices driven by L{\'e}vy processes, we refer again to Cont and Tankov \cite{Cont} and the many references therein. 
We assume for all $i,j \in \{1,\ldots,d\}$ and for all $t\in[0,T]$ that $\sigma^{ij}_t >0$, $\sigma^{ij}_t=\sigma^{ji}_t$, and $\sigma_t$ is positive definite and uniformly bounded away from zero. Let $u^i_t$ be the amount of money invested in asset $S^i$ until time $t$. Then, $u_t=(u_t^1, \ldots, u_t^d)^\intercal$ is a $d$-dimensional vector. We assume that $u$ may only take values in a given closed and convex set $\mathcal{V} \subset \R^d$. Define then the $d$-dimensional, measurable, closed, and convex control space $\{\mathcal{U},\mathcal{G}_U\}$ such that the control functions in $\mathcal{U}$ are $\mathcal{F}$-adapted, predictable, Markovian, and $\L^{\alpha}((0,T],\diff \PP \times \diff s)$-integrable. For instance, if short-selling is prohibited, then $\mathcal{V} \subset \R_{\geq 0}^d$. The wealth process $X^u$ is given by
\begin{align} \label{def: financial market 2}
	\diff X^u_t = [rX^u_t + \Sum_{i=1}^d u^i_t(\mu^i_t - r)]\diff t + \Sum_{i,j=1}^d u^i_{t} \sigma^{ij}_t \diff L^j_t.
\end{align}
For $\alpha=2$, \eqref{def: financial market 2} is also called the Black-Scholes model. In this case, we assume that the correlation matrix process $R_t := (R_t^{jl})_{jl}$ with $R_t^{jl} := \corr( W_t^j,W_t^l )$ is positive definite and uniformly bounded away from zero (see Skintzi and Refenes \cite{Skintzi}). %The exclusion of $\alpha=1$ is made since there exists no $1$-stable process which is not a pure drift process such that for a deterministic strategy $u \neq 0$, $\mathbb{E} [\T(X_T^u-x e^{rT})]$ in \eqref{J def intro} is integrable (since such processes cannot have only upward (or only downward) jumps, see \c{C}{\i}nlar \cite[pp.332-336]{Cinlar}). Thus, the negative part of $\mathbb{E} [\T(X_T^u-x e^{rT})]$ is not integrable for any concave function for $\alpha=1$.

Moreover, we define the measurable control space $\{\mathcal{\tilde{U}},\mathcal{G}_{\tilde{U}}\} \subset \{\mathcal{U},\mathcal{G}_U\}$ with piecewise constant Markovian control functions, add a ``$det$'' in the index of $\mathcal{U}$ resp. $\mathcal{\tilde{U}}$ if we restrict to deterministic control functions and we use the short notation $\E_{{t},x}^u [Y] := \E [ Y | X^u_{t} = x]$ (similarly also for variances and risk measures). As Bj{\"o}rk and Murgoci \cite{Bjork}, we will typically omit the $u$ in $\E_{{t},x}^u$.

\subsection{Risk measures}

In this subsection, we specify the properties of the risk measure from problem \eqref{J def intro}. 

\begin{definition} \label{risk measure definition}
	Let $\mathcal{R}$ be a closed vector space with $\mathcal{R} \supset \L^{\alpha} (\mathcal{F}_T)$. We can interpret $\mathcal{R}$ as all possible portfolio gains (and losses). Then, we call the mapping $\rho: \mathcal{R} \to \R \cup \{ + \infty \}$ a risk measure if it satisfies the following properties:
	\begin{enumerate}[(a)]
		\itemsep0em
		\item law-invariance: $\rho_0 (Y_1) = \rho_0 (Y_2)$ for $Y_1 \overset{d}{=} Y_2$ under $\PP$, i.e., the distribution of the changes of a portfolio determines the value of the risk measure,
		\item positive homogeneity: $\rho_0 (\beta Y) = \beta \rho_0 (Y)$ for all $\beta \geq 0$, i.e., the risk of a portfolio is proportional to its size,
		\item[(c1)] cash-invariance: $\rho_0 (Y+m) = \rho_0(Y) - m$ for all $m \in \R$, i.e., the addition of a sure capital amount to the portfolio reduces the risk by exactly this amount, or
		\item[(c2)] shift-invariance: $\rho_0 (Y+m) = \rho_0 (Y)$ for all $m \in \R$, i.e., changing the portfolio by a fixed capital amount does not influence the risk.
	\end{enumerate}
	Let $\rho_t$ be the dynamic version of $\rho$. Specially, we define $\rho'(H_Y) := \rho(Y)$, where $H_Y$ denotes the cdf of $Y$, and notice that by (a), $\rho'$ is well defined. We then set $\rho_t (Y) = \rho'(H_Y^t)$ with $H_Y^t (s) := \PP(Y \leq s | \mathcal{F}_t)$ and $\rho_{t,x}^u (Y) = \rho'(H_Y^{t,x,u})$ with $H_Y^{t,x,u} (s) := \PP(Y \leq s | X_t^u = x)$. We will typically omit the $u$ also in the notation of $\rho_{t,x}^u$. Moreover, we denote: $\varrho^{\bar{L}_1} := \rho_0 (\bar{L}_1)$ for a L{\'e}vy process $\bar{L}$.
\end{definition}

The case with (c1) refers to measuring the risk by the capital amount needed so that the financial institution is ``safe'', whereas (c2) refers to measuring the deviation of the portfolio. For background on risk measures satisfying (c1), we simply refer to the textbook by F{\"o}llmer and Schied \cite{Follmer} and the many references therein. For background of risk measures satisfying (c2), we refer to Rockafellar et al. \cite{Rockafellar}. Note that $\rho$ is not necessarily continuous; see, e.g., Example \ref{ex: risk measures}(\ref{ex: VaR}) below.

\begin{ex} \label{ex: risk measures}
	The following examples satisfy the previous definition's properties (a)-(c). Examples (i) and (ii) satisfy (c1), whereas examples (iii) - (v) satisfy (c2).
	\begin{enumerate}[(i)]
		\itemsep0em
		\item \label{ex: VaR}Value at Risk (VaR): $\mathrm{VaR}_t^{p}(Y) := \essinf \lbrace m\in \L^{0}(\mathcal{F}_t) \ | \ \PP (Y+m<0|\mathcal{F}_t) \leq p \rbrace$. 
		\item All coherent risk measures. A risk measure $\rho$ is coherent if it is additionally monotone, i.e., $\rho_0(Z) \geq \rho_0(Y)$ if $Z \leq Y$ a.s., and sub-additive, i.e., $\rho_0(Z+Y) \leq \rho_0(Z)+\rho_0(Y)$. One example is given by Average Value at Risk (AVaR): $\mathrm{AVaR}_t^{p}(Y) := \frac{1}{p} \Int_0^p \mathrm{VaR}_t^{q}(Y) \diff q$ which is the basis, for instance, of the Swiss solvency test.
		\item \label{ex: VaR ii} Relative Value at Risk (RVaR): $\mathrm{RVaR}_t^{p}(Y) := \mathrm{VaR}_t^{p}(Y - \E[Y])$ (see Jorion \cite{jorion2007value}) which is the shift-invariant version of VaR considering the deviation from the mean.
		\item \label{ex: sd}Standard deviation: $sd_t (Y) := \sqrt{\Var_t (Y)}$ with $\Var_t$ denoting the conditional variance.
		\item All generalized deviation risk measures (see Rockafeller et al. \cite{rockafellar2002deviation}, Pistorius and Stadje \cite{pistorius2017dynamic} and Stadje \cite{stadje2020two}). A generalized deviation risk measure additionally to (a), (b), and (c2) is assumed to be positive, i.e., $\rho_0 (Y) > 0$ for all non-constant $Y$, and $\rho_0 (Y)=0$ for all constant risks $Y$, and sub-additive.
	\end{enumerate}
\end{ex}

To also include risk measures like variance, in the sequel, we allow the risk measure to be plugged into a convex $C^1$ function $F$.

\section{Analysis of (\ref{J def intro}) in discrete time} \label{sec: disc}

\subsection{Definition of the model}

At every point, we maximize the expected target value of the future net gains penalized through the composition of the function $F$ and the risk measure $\rho$.\\
We divide $[0,T]$ in $N$ intervals with $N+1$ grid points: $0=t_0 < t_1 < \ldots < t_N = T$. We denote by $\delta$ the mesh size, i.e., $\delta = \max_{i \in \{ 1,\ldots,N\}} (t_i-t_{i-1})$. Without loss of generality, we assume that the grid points are equidistant. We consider constant control functions in each interval and define the control function $u^i_{t_n}$ at all time points setting $u^i_{t_n} = u^i_{t_{n}-h}$ for all $h \in [0,\delta)$, for all $n \in \{1,\ldots,N\}$ and for all $i \in \{1,\ldots,d\}$. \\
In the first step, let us set the step size $\delta=1$ and write $n$ instead of $t_n$. Thus, $u_n^i$ describes the investment in the risky asset $i$ in the interval $(n-1,n]$. 

To get an equivalent structure as the introduced problem \eqref{J def intro} (see Proposition \ref{equivalent problem}), we define the value functional $J_n$ as follows: Let $(n,x) \in \{0,\ldots,T-1\} \times \R$ and a control law $u$ be fixed. Then, we define $J_n$ as:
\begin{align} \label{J def}
	J_n (x,u) = \mathbb{E}_{n,x} [\T(X_T^u-x e^{r(T-n)})]  - \lambda_n F(\rho_{n,x} (X_T^u - x e^{r(T-n)})),
\end{align}
where $\T$, $F$, and $\rho$ are as previously described in the introduction and Section \ref{sec: def}. Without loss of generality, see Remark \ref{remark: non Markovian}, we ex-ante restrict ourselves to Markovian strategies.

From this definition, we see that this is a generalization of the mean-variance problem, which we get using $\T(x)=x$, $F(x)=\max \{0,x\}^2$ and the standard deviation for $\rho$.

To ensure that $J_n$ is well-defined, we have to make the following assumption:

\begin{assumption} \label{assumption expected value}
	If $\mathcal{V} \subset \R^d_{\geq 0}$ (resp. $\mathcal{V} \subset \R^d_{\leq 0}$), then $\E [|\T(a+b L_t)|] < \infty$ for all $a \in \R$, $b>0$ (resp. $b<0$), and for all $t \in (0,T]$. Otherwise, $\E [|\T(a+b L_t)|] < \infty$ for all $a,b \in \R$ and for all $t \in (0,T]$.
\end{assumption}

For $\alpha=2$, Assumption \ref{assumption expected value} is in particular true if there exists $K_1,K_2,K_3 \in \R_{\geq 0}$ such that $|\T(x)| \leq K_1 + K_2 e^{K_3 |x|}$, since in this case $L_1$ is normally distributed. For $\alpha<2$, it is sufficient that there exists $K_1,K_2>0$ and an $\varepsilon>0$ such that $|\T(x)| \leq K_1+ K_2|x|^{\alpha-\varepsilon}$ since only the moments up to order $\theta<\alpha$ exist in this case. If $\mathcal{V} \subset \R^d_{\geq 0}$ (resp. $\mathcal{V} \subset \R^d_{\leq 0}$), the latter inequality only has to hold for positive (resp. negative) values of $x$.

\begin{assumption} \label{ass: rho endlich}
	It holds that $\varrho^{\bar{L}_1}=\rho_0 (\bar{L}_1) < \infty$ for all $c>0$ and the corresponding $\alpha$-stable L{\'e}vy process $\bar{L}$ with L{\'e}vy measures given by \eqref{eq: levy measure}.
\end{assumption}

Next, one of the following three alternatives is assumed to hold:

\begin{assumption+} \label{ass: symmetric multidemensional}
	$L$ is symmetric.
\end{assumption+}

\begin{assumption+} \label{ass: asymmetric}
	$L$ has only upward or downward jumps. In addition, $\mathcal{V} \subset \R_{\geq 0}^d$, i.e., we restrict the optimization to non-negative control functions, i.e., short selling is prohibited and $u \geq 0$.
\end{assumption+}

If $\alpha=0.5$ and $L$ has only upward jumps, the distribution is often called a L{\'e}vy distribution; see, for instance, Cont and Tankov \cite{Cont}.

\begin{assumption+} \label{ass: one dimensional}
	$d=1$, i.e., we only have one risky asset. In addition, $\mathcal{V} \subset \R_{\geq 0}^d$, i.e., we restrict the optimization to non-negative control functions, i.e., short selling is prohibited and $u \geq 0$.
\end{assumption+}

%In the Assumptions \ref{ass: asymmetric} and \ref{ass: one dimensional}, we could alternatively assume that $\mathcal{V} \subset \R_{\leq 0}^d$.
In the case $\alpha=2$, we have a Brownian Motion, which is always symmetric. Hence, distinguishing between Assumptions \ref{ass: symmetric multidemensional}, \ref{ass: asymmetric}, and \ref{ass: one dimensional} is only needed in the case $\alpha<2$. In the following, we denote by $\tilde{\lambda}_{min}$ the smallest eigenvalue of $\sigma_{s}e^{r(T-s)}$ for $s \in [0,T]$ if $\alpha<2$ (or denote by $\tilde{\lambda}^2_{min}$ the smallest eigenvalue of $\sigma_{s}R_s \sigma_s e^{2r(T-s)}$ for $s \in [0,T]$ if $\alpha=2$) and by $\tilde{\mu}_{max} := \max\{|\mu^i_s-r|e^{r(T-s)} | s \in [0,T], i \in \{1,\ldots,d\}\}$. It holds that $\tilde{\lambda}_{min}>0$. 

\setcounter{subassumption}{0}

Next, one of the following four alternatives is assumed to hold:

\begin{assumption+} \label{ass: bounded strategies}
	There exists an $M>0$ such that $\mathcal{V} \subseteq [-M,M]^d$, i.e., the strategies are uniformly bounded.
\end{assumption+}

We remark that under Assumption \ref{ass: bounded strategies}, the market may admit arbitrage (which, however, cannot be scaled up). For instance, for $d\geq 2$ suppose that $\tilde{\sigma}$ has a positive probability only at the two points on the first axis, i.e., $\tilde{\sigma} ((1,0,\ldots,0)) = \tilde{\sigma} ((-1,0,\ldots,0)) = \frac{1}{2}$. Then, the marginal L{\'e}vy process in all other directions is a deterministic drift process given by $\mu^it$. In particular, the market admits arbitrage if $\mu^i \not \equiv r$ for an $i \geq 2$. To exclude arbitrage and such examples while not relying on a boundedness assumption, we need $\tilde{\sigma}$ to be non-degenerated on the sphere. In particular, we need some measure in every axis direction. In the Black-Scholes market, non-degeneracy assumptions are standard, i.e., for instance, it is common there to assume that $\sigma^\intercal \sigma$ has a full rank and $\sigma$ is bounded away from zero where $\sigma$ is the volatility matrix.

\begin{assumption+} \label{ass: one dimensional 2}
	$d=1$, i.e., we only have one risky asset, $\alpha \in (1,2)$, $\lim_{x \to \infty} \T'(x) =0$, and $\lambda_n \equiv \lambda>0$. If $\rho$ is cash-invariant, we additionally assume that $\tilde{\mu}_{max} < \varrho^{\hat{L}_1} \min_{s \in [0,T]} \{ \sigma_s e^{r(T-s)}\}$.
\end{assumption+}

Note that the last inequality is usually satisfied for the Black-Scholes model and $\rho = \mathrm{VaR}^{\alpha}$ with typically $\alpha \geq 90 \%$, $\mu-r \leq 0.08$, and $\sigma \geq 0.1$. Recall for the following assumption that a hyperspherical cap with polar angle $\beta$ is a portion of the sphere cut off by a hyperplane, where $\beta$ is the angle between the vectors from the center of the hypersphere to the pole of the cap and the edge of the hyperdisk forming the base of the cap.

\begin{assumption+} \label{ass: shift-inv and cash-inv and condition}
	$d \geq 2$, $\alpha \in (1,2)$, $\lim_{x \to \infty} \T'(x) =0$, $\lambda_n \equiv \lambda>0$, and $\varrho^{\hat{L}_1} >0$. In addition, we assume that there exist $\hat{\varepsilon}>0$, $\tilde{\varepsilon} \in (0,\frac{\pi}{24})$ such that either (a) $\tilde{\sigma}$ has a density bigger than $\hat{\varepsilon}$ on the unit ball around a vector on the sphere with respect to the radial distance\footnote{The radial distance gives the distance between two points on $\S^d$ along the radial dimension, i.e., the length of the shortest path on $\S^d$.} with radius $12 \tilde{\varepsilon}$, or (b) $\tilde{\sigma}$ has a point mass bigger than $\hat{\varepsilon}$ at $(d+1)$ points such that each point has a distance of at least $\tilde{\varepsilon}$ to every $(d-1)$-dimensional hyper-plane spanned by the other $d$ points. If $\rho$ is cash-invariant, we additionally assume that $\varrho^{\hat{L}_1} \tilde{\lambda}_{min} \min\{1,\mathbb{A}_d\}^{1/\alpha} \cos (\tilde{\varepsilon}) \hat{\varepsilon}^{1/\alpha} > \tilde{\mu}_{max}$, where $\mathbb{A}_d$ denotes the surface area of a hyperspherical cap of $\S^d$ with polar angle $\tilde{\varepsilon}$.
\end{assumption+}

If $\tilde{\sigma}$ has full support on the whole sphere and a continuous density, it fulfills the assumption stipulated in Assumption \ref{ass: shift-inv and cash-inv and condition} for $\rho$ being shift-invariant.

\begin{assumption+} \label{ass: alpha=2}
	$\alpha=2$, $\lim_{x \to \infty} \T'(x) =0$, $\lambda_n \equiv \lambda>0$, and $\varrho^{W_1} >0$. If $\rho$ is cash-invariant, we additionally assume that $\varrho^{W_1} \tilde{\lambda}_{min} > \tilde{\mu}_{max}$.
\end{assumption+}

Note that the last inequality is satisfied for $\rho = \mathrm{VaR}^{\alpha}$ in a multidimensional Black-Scholes setting with typical parameter values.

\subsection{Infinitesimal generator and Nash equilibrium} \label{def: infini and Nash}

\begin{definition} \label{inf gen def}
	Let $f: \mathbb{N} \times \R \rightarrow \mathbb{R} $ be a function sequence. We define the infinitesimal generator $\mathrm{A}^u$ in discrete time as $(\mathrm{A}^u f)_n (x) = \mathbb{E}_{n,x} [ f_{n+1} ( X^u_{n+1}) - f_n (x) ].$
\end{definition}

\begin{definition} \label{def: Nash equilibrium disc}
	Consider a fixed control law $\hat{u}$. Fix an arbitrary point $(n,x) \in \mathbb{N} \times \R$ with $n<T$ and take an arbitrary control law $u \in \mathcal{U}$. We define then for the time points $n, n+1, \ldots, T-1$ the control law $\bar{u}$ as follows for any $y \in \R$: $\bar{u}_k (y) = \left\{ \begin{array}{cl}
		\hat{u}_k (y) & \text{for } k=n+1, \ldots, T-1,  \\
		u & \text{for } k=n
	\end{array} \right.  .$
	Then the control law $\hat{u}$ is a perfect subgame Nash equilibrium if for all fixed $(n,x)$ it holds that $\sup_{u \in \mathcal{U}} J_n (x, \bar{u}) = J_n (x, \hat{u})$. If such an equilibrium exists, then we also define the equilibrium value function $V$ by $V_n (x) = J_n (x, \hat{u})$.
\end{definition}
The following proposition is well-known and not difficult to prove. It shows that an optimal time-consistent strategy indeed corresponds to a Nash equilibrium.
\begin{proposition} \label{equivalent problem}
	Let $\mathcal{A}_{T,x} := \{ u_T = u(T,x) | u_T \in \mathcal{V} \}$. Next, we recursively define $u_j^*$ as $u_j^*:= \argmax_{u \in \mathcal{A}_{j,x}} \mathbb{E}_{n,x} [\T(X_T^u-x e^{r(T-n)})]  - \lambda_n F(\rho_{n,x} (X_T^u - x e^{r(T-n)}))$ for $j=n+1,\ldots,T$ and $\mathcal{A}_{n,x}$ by $\mathcal{A}_{n,x} := \{ (u_n,u_{n+1},\ldots, u_T) | u_n = u(n,x), u_n \in \mathcal{V}, u_j = u_j^* \textit{ for } j=n+1,\ldots,T \}$. Then, the control law $u^*$ is a perfect subgame Nash equilibrium if and only if $u^*$ is the $\argmax$ of problem \eqref{J def intro}.
\end{proposition}

\begin{proof}
	The $\argmax$ may be either restricted to $u$ bounded by $M$ (under Assumption \ref{ass: bounded strategies}) or by Proposition \ref{u unif bounded} to $\mathcal{U} \cap \{ u\ | \sum_{i=n}^T |u_i|^{\alpha} \leq M_n^{\alpha} \}$ (for appropriate $M_n>0$). Hence, the maximum exists, and the proof follows by induction. 
\end{proof}

\subsection{Properties of the optimal control function} \label{section optimal}

This section proves that the optimal control function is deterministic and exists. It is unique if $\mathcal{T}$ or $-F$ is strictly concave. 
Here, we use the definition of stochastic integrals for functions from $\cL$ (the set of adapted and c{\`a}gl{\`a}d processes, see \ref{def: stochastic integral}) since the discrete model restricts to piecewise constant control functions with left limits. 

\begin{theorem} \label{deterministic control function}
	The supremum in \eqref{J def intro} exists and may (and will in the sequel) be taken only over deterministic strategies. It is unique if $\T$ or $-F$ is strictly concave and $\alpha>1$.
\end{theorem}

Before proving this theorem, we need some additional definitions and results:

\begin{remark} \label{remark: p-q ind a}
	Under Assumption \ref{ass: symmetric multidemensional} and \ref{ass: asymmetric}, ${\Int_{\S^d} \norm{\langle a,v \rangle}^{\alpha} \sgn( \langle a,v \rangle )\tilde{\sigma} (\diff v)}({\Int_{\S^d} \norm{\langle a,v \rangle}^{\alpha} \tilde{\sigma} (\diff v)})^{-1}$ is independent of $a \in \R_{\geq 0}^d$, and we will in the sequel refer to this number as $p-q$ with $p+q=1$. It holds that $p-q=1$ (resp. $p-q=-1$) if $L$ has only upward (resp. downward) jumps and $p-q=0$ if $L$ is symmetric. This will help us to identify later certain stochastic integrals with one-dimensional L{\'e}vy processes using equation \eqref{eq: characteristic function 1dim}.
\end{remark}

\begin{definition} \label{remark: alpha=2}
	Under Assumption \ref{ass: symmetric multidemensional} and \ref{ass: asymmetric}, if $\alpha \in (0,2] \backslash \{1\}$, let $\tilde{L}$ be a one-dimensional $\alpha$-stable L{\'e}vy process with characteristic function given by \eqref{eq: characteristic function 1dim} with the same constant $c$ in the L{\'e}vy measure as $L$ and (as in Remark \ref{remark: p-q ind a}) with $p-q= 2p-1 = {\Int_{\S^d} \norm{\langle a,v \rangle}^{\alpha} \sgn( \langle a,v \rangle )\tilde{\sigma} (\diff v)}$ $({\Int_{\S^d} \norm{\langle a,v \rangle}^{\alpha} \tilde{\sigma} (\diff v)})^{-1}$, where $p-q$ is as in \eqref{eq: characteristic function 1dim}. \\%see also Subsection \ref{model setup} and Equation (\ref{eq: levy measure})). \\
	If $\alpha=2$, let $\tilde{L}$ be a one-dimensional Brownian Motion, i.e., $\tilde{L} \overset{d}{=} W^j$ for all $j \in \{1,\ldots,d\}$.
\end{definition}

\begin{definition} \label{Hilbert space H}
	We define the Hilbert space $H_t$ as $(\R^d,\langle \cdot,\cdot \rangle_{R_t})$ with $R_t$ as in Subsection \ref{model setup}. Then, we set $\norm{x}_{R_t} := \sqrt{\langle x,x \rangle_{R_t}} = \sqrt{x R_t x^\intercal}$ for a row vector $x$ and with a slight abuse of notation write $\norm{x}_{\L^2((a,b],\diff t; H_t)} = \sqrt{\Int_a^b \norm{x}_{R_t}^2 \diff t} $ with $a<b$.
\end{definition}

\begin{definition} \label{def m w}
	We define the two functions $m$ and $w$ as follows for all $s \in (0,T]$
	\begin{align*}
		m_s (u) &:= \left\{ \begin{matrix*}[l]
			u_s^\intercal (\mu_s - r) e^{r(T-s)} & \text{if the risk measure $\rho$ is cash-invariant,} \\
			0 & \text{if the risk measure $\rho$ is shift-invariant,}
		\end{matrix*} \right.  \\
		\leftidx{^{\alpha}}{w}_s (u)& := \left\{ \begin{matrix*}[l]
			\norm{u_s^\intercal \sigma_s v e^{r(T-s)}}_{\L^{\alpha}(\S^d,\tilde{\sigma}(\diff v))} & \text{if }\alpha<2. \\
			\norm{ u_{s}^\intercal \sigma_{s} e^{r(T-{s})}}_{R_s} & \text{if }\alpha=2.
		\end{matrix*}  \right. 
	\end{align*}
	Moreover, we define $\hat{L} := \left\{ \begin{matrix*}[l]
		\tilde{L} & \text{if }\alpha<2 \ \& \ \text{Ass. } \ref{ass: symmetric multidemensional} \text{ or } \ref{ass: asymmetric} \text{ (see Definition \ref{remark: alpha=2})}. \\
		L^1 & \text{if }\alpha<2 \ \& \ \text{Ass. } \ref{ass: one dimensional}. \\
		W^1 & \text{if }\alpha=2.
	\end{matrix*}  \right. $
\end{definition}

If $d=1$, then $\leftidx{^{\alpha}}w_s (u) = |u_s^1 \sigma_s^{11} e^{r(T-s)}|$ in all cases since then it holds that $R_s = (1)$ and
\begin{align} \label{eq: one dimensional simplification}
	\norm{u_s^\intercal \sigma_s v e^{r(T-s)}}_{\L^{\alpha}(\S^1,\tilde{\sigma}(\diff v))} &= |u_s^1 \sigma_s^{11} e^{r(T-s)}| \sqrt[\alpha]{\Int_{\S^1} |v|^{\alpha} \tilde{\sigma} (\diff v)} 
	= |u_s^1 \sigma_s^{11} e^{r(T-s)}|.
\end{align}

Note that for $\alpha<1$, $\norm{\cdot}_{\L^{\alpha}}$ is only a quasi-norm.

The following lemmas show that the stochastic integral of an $\alpha$-stable L{\'e}vy process is still $\alpha$-stable. We state the proofs in \ref{proofs}. 

\begin{lemma} \label{alpha additive}
	Let $a = (a^1,\ldots,a^d)^\intercal \in \R^d_{\geq 0}$, $\alpha<2$ and $0\leq s<t$. Then it holds in the case of the Assumptions \ref{ass: symmetric multidemensional} or \ref{ass: asymmetric} that $\scprod{a,L_t - L_s} \overset{d}{=} \sqrt[\alpha]{\Int_{\S^d} |\scprod{a,v}|^{\alpha} \tilde{\sigma}(\diff v)} (\tilde{L}_t-\tilde{L}_s )$.
\end{lemma}

\begin{lemma} \label{alpha dist}
	Let $f$ be a deterministic and continuous function. Moreover, in the case of the Assumptions \ref{ass: asymmetric} or \ref{ass: one dimensional}, let $f$ additionally be non-negative. Then it holds for all $t \in [0,T)$ that $\Int_{t+}^T f(s) \diff \hat{L}_s \overset{d}{=} \sqrt[\alpha]{\Int_t^T |f(s)|^{\alpha} \diff s} \hat{L}_1$.
\end{lemma}

\begin{proposition} \label{risk measure formula}
	It holds $\rho_{n,x} (X_T^u - x e^{r(T-n)}) 
	= -\Int_n^T m_s (u) \diff s + \varrho^{\hat{L}_1} \norm{\leftidx{^{\alpha}}w_s (u)}_{\L^{\alpha}((n,T],\diff s)}$ for all $u \in \mathcal{U}_{det}$.
\end{proposition}

\begin{proof}
	We start with the case of Assumption \ref{ass: symmetric multidemensional} or \ref{ass: asymmetric}. First of all, we can rewrite the controlled Markov process $X$ with It\^{o}'s lemma for semimartingales in integral form (see Protter \cite[pp.81-82] {protter2005stochastic}). Hence, we get for all $u \in \mathcal{U}$:
	\begin{align} \label{eq: umschreibung xt-xn}
		X_T^u - X_n^u e^{r(T-n)} =& \Int_{n+}^T -re^{r(T-s_-)} X_{s-}^u \diff s +  \Int_{n+}^T e^{r(T-s_-)} \diff X_s^u + 0 \notag \\
		&+ \Sum_{t < s \leq T} \{ e^{r(T-s)} X_s^u - e^{r(T-s_-)} X_{s_-}^u - 0 - e^{r(T-s_-)} \Delta X_s^u \} \notag \\
		=& \Int_{n+}^T re^{r(T-s)} \Delta X_{s}^u \diff s + \Sum_{i=1}^d \Int_{n+}^T u^i_{s-} (\mu^i_s-r) e^{r(T-s)} \diff s \notag \\
		&+ \Sum_{i,j=1}^d  \Int_{n+}^T u^i_{s-} \sigma^{ij}_s e^{r(T-s)} \diff L^j_s \notag \\
		=& \Sum_{i=1}^d \Int_{n}^T u^i_s (\mu^i_s-r) e^{r(T-s)} \diff s + \Sum_{i,j=1}^d  \Int_{n+}^T u^i_{s} \sigma^{ij}_s e^{r(T-s)} \diff L^j_s,
	\end{align} 
	since $X_s^u$ is a c{\`a}dl{\`a}g process (due to $L_s$ being c{\`a}dl{\`a}g) and has therefore only countably many points of discontinuity. Thus, the points of discontinuity form a Lebesgue null set.\\
	Then, we get from (\ref{eq: umschreibung xt-xn}) combined with the definition of $m$ (see Definition \ref{def m w}) and the cash- or shift-invariance of $\rho$ that
	\begin{align*} 
		\rho_{n,x} &( X_T^u - x e^{r(T-n)} ) = -\Int_n^T m_s (u) \diff s + \rho_{n,x} ( \Sum_{i,j=1}^d  \Int_{n+}^T u^i_{s} \sigma^{ij}_s e^{r(T-s)} \diff L^j_s ).
	\end{align*}
	Now, since $u \in L^\alpha$ and deterministic, there exist deterministic and continuous $\leftidx{^l}{u}$ such that $\leftidx{^l}{u} \xrightarrow{L^\alpha} u$. Then, (\ref{eq: formula integral construction}) implies that
	\begin{align} \label{prop: eq2}
		\Sum_{i,j=1}^d  \Int_{n+}^T \leftidx{^l}u^i_{s} \sigma^{ij}_s e^{r(T-s)} \diff L^j_s &\overset{\PP}{=} \Sum_{i,j=1}^d  \lim\nolimits_{K \to \infty} \Sum_{k=1}^{K} \leftidx{^l}u^i_{t^K_k} \sigma^{ij}_{t^K_k} e^{r(T-{t^K_k})} ( L_{t^K_k}^j-L_{t^K_{k-1}}^j ) \notag \\
		&= \lim\nolimits_{K \to \infty} \Sum_{k=1}^{K} \Sum_{j=1}^d ( \Sum_{i=1}^d  \leftidx{^l}u^i_{t^K_k} \sigma^{ij}_{t^K_k} e^{r(T-{t^K_k})} ) ( L_{t^K_k}^j-L_{t^K_{k-1}}^j )
	\end{align}
	with the equidistant grid points $t_k$, i.e., $n=t^K_0 < t^K_1 < \ldots < t^K_K=T$.\\
	We must distinguish the cases $\alpha<2$ and $\alpha=2$. We start with the case $\alpha<2$:\\
	We apply Lemma \ref{alpha additive} on the term inside of the sum over $k$, i.e.:
	\begin{align} \label{prop: eq3}
		&\Sum_{j=1}^d ( \Sum_{i=1}^d  \leftidx{^l}u^i_{t^K_k} \sigma^{ij}_{t^K_k} e^{r(T-{t^K_k})} ) ( L_{t^K_k}^j-L_{t^K_{k-1}}^j ) \notag \\
		&\hspace{3cm}\overset{d}{=} \sqrt[\alpha]{\Int_{\S^d} | \Sum_{j=1}^d ( \Sum_{i=1}^d  \leftidx{^l}u^i_{t^K_k} \sigma^{ij}_{t^K_k} e^{r(T-{t^K_k})} ) v^j |^{\alpha} \tilde{\sigma}(\diff v)} ( \tilde{L}_{t^K_k}-\tilde{L}_{t^K_{k-1}} ).
	\end{align}
	We know for random variables $Y_i,Z_i, i \in \{1,\ldots,d\}$: If $Y_i \overset{d}{=} Z_i$, $Y_i \indep Y_j$, and $Z_i \indep Z_j$ for all $i \neq j$, then also $\Sum_{i=1}^d Y_i \overset{d}{=} \Sum_{i=1}^d Z_i$. It is: 
	\begin{align*}
		\rho_{n,x} &( \Sum_{i,j=1}^d  \Int_{n+}^T u^i_{s} \sigma^{ij}_s e^{r(T-s)} \diff L^j_s ) \\
		&= \rho_{n,x} ( \lim\nolimits_{l \to \infty} \Sum_{i,j=1}^d  \Int_{n+}^T \leftidx{^l}u^i_{s} \sigma^{ij}_s e^{r(T-s)} \diff L^j_s ) \\
		&= \rho_{n,x} ( \lim\nolimits_{l \to \infty} \lim\nolimits_{K \to \infty} \Sum_{k=1}^{K} \Sum_{j=1}^d ( \Sum_{i=1}^d  \leftidx{^l}u^i_{t^K_k} \sigma^{ij}_{t^K_k} e^{r(T-{t^K_k})} ) ( L_{t^K_k}^j-L_{t^K_{k-1}}^j ) )\\
		&= \rho_{n,x} ( \lim\nolimits_{l \to \infty} \lim\nolimits_{K \to \infty} \Sum_{k=1}^{K} \sqrt[\alpha]{\Int_{\S^d} | \Sum_{j=1}^d ( \Sum_{i=1}^d  \leftidx{^l}u^i_{t^K_k} \sigma^{ij}_{t^K_k} e^{r(T-{t^K_k})} ) v^j |^{\alpha} \tilde{\sigma}(\diff v)} ( \tilde{L}_{t^K_k}-\tilde{L}_{t^K_{k-1}} ) )\\
		&= \rho_{n,x} ( \lim\nolimits_{l \to \infty} \Int_{n+}^T \sqrt[\alpha]{\Int_{\S^d} | \Sum_{j=1}^d ( \Sum_{i=1}^d  \leftidx{^l}u^i_{s} \sigma^{ij}_{s} e^{r(T-{s})} ) v^j |^{\alpha} \tilde{\sigma}(\diff v)} \diff \tilde{L}_s )\\
		&= \rho_{0} ( \lim\nolimits_{l \to \infty} \Int_{n+}^T \sqrt[\alpha]{\Int_{\S^d} | \Sum_{j=1}^d ( \Sum_{i=1}^d  \leftidx{^l}u^i_{s} \sigma^{ij}_{s} e^{r(T-{s})} ) v^j |^{\alpha} \tilde{\sigma}(\diff v)} \diff \tilde{L}_s )\\
		&= \rho_{0} ( \lim\nolimits_{l \to \infty} \sqrt[\alpha]{\Int_{n+}^T | \sqrt[\alpha]{\Int_{\S^d} | \Sum_{j=1}^d ( \Sum_{i=1}^d  \leftidx{^l}u^i_{s} \sigma^{ij}_{s} e^{r(T-{s})} ) v^j |^{\alpha} \tilde{\sigma}(\diff v)} |^{\alpha} \diff s} \tilde{L}_1 )\\
		&= \rho_{0} ( \sqrt[\alpha]{\Int_{n+}^T \Int_{\S^d} | \Sum_{j=1}^d ( \Sum_{i=1}^d u^i_{s} \sigma^{ij}_{s} e^{r(T-{s})} ) v^j |^{\alpha} \tilde{\sigma}(\diff v)} \diff s \tilde{L}_1 )\\
		&= \varrho^{\tilde{L}_1} \norm{\Sum_{j=1}^d ( \Sum_{i=1}^d  u^i_{s} \sigma^{ij}_{s} e^{r(T-{s})} ) v^j}_{\L^{\alpha}((n,T] \times \S^d, \diff s \times \tilde{\sigma}(\diff v))} = \varrho^{\tilde{L}_1} \norm{u_s^\intercal \sigma_s v e^{r(T-s)}}_{\L^{\alpha}((n,T] \times \S^d, \diff s \times \tilde{\sigma}(\diff v))}
	\end{align*}
	with Applebaum \cite[p.237]{applebaum2009levy} and $\leftidx{^l}{u} \xrightarrow{L^\alpha} u$ in the first equality, (\ref{prop: eq2}) in the second equality, (\ref{prop: eq3}) in the third equality, the independent increments of $L$ and the definition of $\rho$ in the fifth equality, Lemma \ref{alpha dist} in the sixth equality due to $\leftidx{^l}u$ being deterministic, bounded, and continuous, $\leftidx{^l}{u} \xrightarrow{L^\alpha} u$ in the seventh equality, and the positive homogeneity and the law-invariance of $\rho$ in the eighth equality.
	
	The proof works with similar ideas if $\alpha=2$ or if Assumption \ref{ass: one dimensional} holds.
\end{proof}

\begin{proposition} \label{Jn concave}
	Let $\alpha>1$. If the control function $u$ is deterministic, then the functional $J_n$ is concave in $u$. If $\T$ or $-F$ is strictly concave, $J_n$ is also strictly concave.
\end{proposition}

We state the proof of this and the following two propositions in \ref{proofs}.

\begin{proposition} \label{u unif bounded}
	The optimal deterministic strategies are uniformly bounded in $(\L^\alpha,\norm{\cdot}_{\L^\alpha})$ by an $M>0$ independent of $\delta > 0$.
\end{proposition}

\begin{proposition} \label{Jn cont}
	If the control function $u$ is deterministic, then the functional $J_n(x,u)$ is continuous in $u\, \in \tilde{\mathcal{U}}_{det}$ for every $x$. (Note that since $\tilde{\mathcal{U}}_{det}$ is the set of deterministic discrete-time strategies, $u$ may be seen as a vector in $\R^{d \times (T-n)}$).
\end{proposition}

Now, we are able to prove our main result of this section, i.e., Theorem \ref{deterministic control function}:

\begin{myproof}[Proof of Theorem \ref{deterministic control function}] 
	We prove this theorem by backward induction. Let $M>0$ be greater than the constant from Proposition \ref{u unif bounded}.
	
	\underline{$n=T-1$:}\\
	We know due to \eqref{eq: umschreibung xt-xn} and as the control function is piecewise constant:
	\begin{align} \label{xt - xt-1}
		X_T^u - X_{T-1}^u e^{r(T-n)} = u_{T}^\intercal \Int_{T-1}^T (\mu_s-r) e^{r(T-s)} \diff s + u_T^\intercal \Int_{(T-1)+}^T  \sigma_s e^{r(T-s)} \diff L_s.
	\end{align}
	For a vector $u=(u^1,\ldots,u^d)^\intercal \in \R^d$, denote $\norm{u}_\alpha = \sum_{i=1}^{d} |u^i|^\alpha$. We define $u^* \in \mathcal{V} \subset \R^d$ as:
	\begin{align} \label{eq: definition ustar}
		u^* \ \in \argmax_{\norm{u}_{\alpha} \leq M, u \in \mathcal{V}} &\{ \mathbb{E} [\T(u^\intercal \Int_{T-1}^T (\mu_s-r) e^{r(T-s)} \diff s + u^\intercal \Int_{(T-1)+}^T  \sigma_s e^{r(T-s)} \diff L_s)]   \notag \\
		&\hspace{8pt}-   \lambda_{T-1} F(\rho (u^\intercal \Int_{T-1}^T (\mu_s-r) e^{r(T-s)} \diff 	s + u^\intercal \Int_{(T-1)+}^T  \sigma_s e^{r(T-s)} \diff L_s)) \} \\
		= \argmax_{\norm{u}_{\alpha} \leq M, u \in \mathcal{V}} &\{ \mathbb{E} [\T(u^\intercal ( \Int_{T-1}^T (\mu_s-r) e^{r(T-s)} \diff s + \Int_{(T-1)+}^T  \sigma_s e^{r(T-s)} \diff L_s) ) ]  \notag  \\
		&\hspace{8pt}-   \lambda_{T-1} F(-\Int_{T-1}^T m_s (u) \diff s + \varrho^{\hat{L}_1} \norm{\leftidx{^{\alpha}}w_s (u)}_{\L^{\alpha}((T-1,T],\diff s)} ) \} \notag \\
		=: \argmax_{\norm{u}_{\alpha} \leq M, u \in \mathcal{V}} &\tilde{V}(u), \notag
	\end{align}
	where we used Proposition \ref{u unif bounded} (with the strategy $u_{T-1} \1_{(T-1,T]}$ and $u$ being then identified with $u_{T-1}$ in $\R^d$) and Proposition \ref{risk measure formula} since $u$ in \eqref{eq: definition ustar} is deterministic. Now, it follows from Proposition \ref{Jn cont} (and $\mathcal{V}$ being closed and convex) that $\tilde{V}$ is continuous in $u$ on a compact set and hence has a global maximum, i.e., $u^*$ exists, and by Proposition \ref{u unif bounded} is independent of $M$. If $\alpha>1$ and $\T$ or $-F$ are strictly concave, Proposition \ref{Jn concave} gives us the uniqueness. Next, let $u = u_{T-1} \1_{(T-1,T]}$ be a Markovian strategy with $u_{T-1} = u(T-1,X_{T-1}^u)$. Note that for all $x \in \R$, we have $\tilde{V}(u(T-1,x)) \uparrow_{M \to \infty} \1_{|u(T-1,x)|\leq M} \tilde{V}(u(T-1,x)) \leq \tilde{V}(u^*)$. Hence, it follows that $V(T-1,x) = \esssup_{u(T-1,x) \in \mathcal{V}} \big\{ \mathbb{E}_{T-1,x} [\T(X_T^u-x e^{r})]  - Risk_{T-1,x} (X_T^u - x e^{r}) \big\} = \esssup_{u(T-1,x) \in \mathcal{V}} \tilde{V} (u(T-1,x)) \leq \tilde{V} (u^*) \leq V(T-1,x)$ a.s.
	for $V$ defined as in \eqref{J def intro}. %(Strictly speaking, we would need to take the $\esssup$ in \eqref{J def intro}. However, since the optimal strategy is deterministic, the ``$\sup$'' and the ``$\esssup$'' agree.) 
	Hence, all the inequalities must be equalities, and $u^* \in \mathcal{V} \subset \R^d$ attains the supremum for $t=T-1$ in \eqref{J def intro}. Thus, the $\argmax$ of $V$ from \eqref{J def intro} can be taken only over deterministic control functions at time $T-1$. Furthermore, the equilibrium ${\hat{u}}_{T}=u^*$ attains this maximum. Since the first inequality is strict if $u(T-1,x) \neq u^*$ provided that either $\T$ or $-F$ is strictly concave and $\alpha>1$, uniqueness follows.
	
	\underline{arbitrary $n$:} This can be done analogously to the first case. Using that the equilibrium value attains the optimum for all later time points, $\hat{u}_{n+1}$ is also the (unique) equilibrium. 
\end{myproof}

\begin{remark} \label{remark: non Markovian}
	If one considers the time-consistent, ex-ante \emph{non-Markovian} problem
	\begin{align*}
		V(t) = \esssup\nolimits_{u\in \hat{\mathcal{A}}_t} \big\{ \mathbb{E}_{t} [\T(X_T^u-X_t^u e^{r(T-t)})]  - Risk_{t} (X_T^u - X_t^u e^{r(T-t)}) \big\},
	\end{align*}
	with $\hat{\mathcal{A}}_t$ being the set of time-consistent, (possibly non-Markovian) \emph{adapted} strategies, one can show as in the proof of Theorem \ref{deterministic control function} that the supremum is attained in the deterministic strategy $u^*$ from the proof of Theorem \ref{deterministic control function}. This justifies the restriction to Markovian strategies in \eqref{J def intro}.
\end{remark}

Next, we analyze if the optimal solution is non-negative. In the case of Assumptions \ref{ass: asymmetric} and \ref{ass: one dimensional}, we already assumed this property; hence, we only need to consider the case of Assumption \ref{ass: symmetric multidemensional}: 

\begin{proposition} \label{un non-negative}
	Under Assumption \ref{ass: symmetric multidemensional}: Let $d=1$, $\alpha>1$, $\mu_t \geq r$ for all $t \in [0,T]$, and $\T$ be such that $\E [|\T'(a+b L_t^1)|] < \infty$ for all $a,b \in \R$ and for all $t \in (0,T]$. Further assume that $F'(0) = 0$. Then, every deterministic optimal control function is non-negative.
\end{proposition}
We state the rather technical proof in \ref{proofs}. If $d>1$, the proposition is no longer valid since some assets might serve as a hedge against others.

\subsection{The HJB equation in discrete time}

Based on the results from above, we can now restrict the analysis to deterministic control functions.

\begin{theorem} \label{hjb discrete}
	If an equilibrium control law $\hat{u}$ exists, the equilibrium value function $V$ satisfies:
	\begin{align} \label{hjb equation disc}
		\sup\nolimits_{u \in \mathcal{V}} \{ (\mathrm{A}^u V)_n (x) - (\mathrm{A}^u \mathbb{E}_{\cdot} [\T(X_T^{\bar{u}} - X_{\cdot}^u e^{r(T-\cdot)})])_n (x)  \hspace{3cm}& \notag\\
		+ (\mathrm{A}^u (\lambda_{\cdot} F(\rho_{\cdot} (X_T^{\bar{u}} - X_{\cdot}^u e^{r(T-\cdot)}))))_n (x) \} &= 0 \\
		V_T (x) &= \T(0) - \lambda_T F(0), \notag
	\end{align}
	where $\bar{u}_s = \footnotesize\begin{cases} u & \text{,if } s = n \\ \hat{u}_s & \text{,if } s \geq n + 1 \end{cases} $ and the supremum is attained at $u=\hat{u}_n$, see Definition \ref{inf gen def} for the definition of $A$. \\
	Let $m$, $w$ and $\hat{L}$ be as in Definition \ref{def m w}. Then we have the following identities for $0 \leq n \leq T-1$:
	\begin{align*}
		&(\mathrm{A}^{\hat{u}} \mathbb{E}_{\cdot} [\T(X_T^{\hat{u}} - X_{\cdot}^{\hat{u}} e^{r(T-\cdot)})])_n (x) = \mathbb{E}_{n,x} [ \T(X_T^{\hat{u}} - X_{n+1}^{\hat{u}} e^{r(T-n-1)}) - \T(X_T^{\hat{u}} - x e^{r(T-n)})] \\
		&(\mathrm{A}^{\hat{u}} (\lambda_{\cdot} F(\rho_{\cdot} (X_T^{\hat{u}} - X_{\cdot}^{\hat{u}} e^{r(T-\cdot)}))))_n (x) \\
		&\hspace{0.6cm}= \lambda_{n+1} [ F (-\Int_{n+1}^T m_s (\hat{u}) \diff s + \varrho^{\hat{L}_1} \norm{\leftidx{^{\alpha}}w_s (\hat{u})}_{\L^{\alpha}((n+1,T], \diff s)} )   \\
		&\hspace{2.2cm}-  F (-\Int_{n}^T m_s (\hat{u}) \diff s + \varrho^{\hat{L}_1} \norm{\leftidx{^{\alpha}}w_s (\hat{u})}_{\L^{\alpha}((n,T], \diff s)} ) ] \\
		&\hspace{1cm}- (\lambda_{n}-\lambda_{n+1}) F(-\Int_{n}^T m_s (\hat{u}) \diff s + \varrho^{\hat{L}_1} \norm{\leftidx{^{\alpha}}w_s (\hat{u})}_{\L^{\alpha}((n,T], \diff s)} ) \\
		&\text{with } \lambda_T F(\rho_{T,x} (X_T^{\hat{u}} - x e^{r(T-T)})) = \lambda_T F(0).
	\end{align*}
\end{theorem}
The proof of this theorem is rather technical, and we give it in \ref{proofs}.

\section{Analysis in continuous time} \label{sec: cont}

\subsection{Definition of the model}

We now extend the discrete-time result into continuous time. Hence, we maximize the value functional $J$ defined as follows with fixed $(t,x) \in [0,T) \times \R$ and a $\L^\alpha (\diff \PP \times \diff s)$-integrable, fixed control law $u$:
\begin{align} \label{J def cont}
	J (t,x,u) = \mathbb{E}_{t,x} [\T(X_T^u-x e^{r(T-t)})]  - \lambda_t F(\rho_{t,x} (X_T^u - x e^{r(T-t)})).
\end{align}

We use a convergence approach to formally obtain the extended HJB equation. However, we first define the infinitesimal generator and Nash equilibria in continuous time and prove some properties of the optimal control function.

\subsection{Infinitesimal generator and Nash equilibrium}

%Following the same references as in Subsection \ref{def: infini and Nash}, we define:
%\begin{definition} \label{inf gen def cont}
%	Let $f: \mathbb{R} \times \R \rightarrow \mathbb{R} $ be a function sequence.
%	\begin{enumerate}
%		\itemsep0em 
%		\item We define the operator $\mathrm{P}_h^u$, $h \in \mathbb{R}$, as
%		$(\mathrm{P}_h^u f) (t,x) =  \mathbb{E}_{t,x} [ f (t+h , X^u_{t+h})].$
%		\item The infinitesimal generator $\mathrm{A}^u$ is defined as $\mathrm{A}^u =  \tfrac{\partial \mathrm{P}_h^u}{\partial h} \big|_{h=0}$.
%		\item Moreover, we define the operator $\mathrm{A}_h^u$, $h \in \mathbb{R}$, as $		(\mathrm{A}_h^u f) (t,x) =  \mathbb{E}_{t,x} [ f (t+h , X^u_{t+h}) - f(t,x)].$
%	\end{enumerate}
%\end{definition}

Building on the definitions from same references as in Subsection \ref{def: infini and Nash}, we define:
\begin{definition} \label{inf gen def cont}
	Let $f_{\hat{u}}: [0,T] \times \R \rightarrow \mathbb{R} $ be a function sequence for a strategy $\hat{u}$. Moreover, we define the strategy $u_h^{v,t}$ for $v \in \mathcal{V}$ by $u_h^{v,t} (s) = \begin{cases}
			\hat{u}_s & \text{for } t+h \leq s \leq T, \\
			v & \text{for } t+h \leq s \leq T.
		\end{cases}$
	\begin{enumerate}
		\itemsep0em 
		\item The infinitesimal generator $\mathrm{A}^{\hat{u}}$ is defined as $(\mathrm{A}^{\hat{u}} f_{\hat{u}}) (t,x) =  \lim_{h \to 0} \tfrac{1}{h} (\E_{t,x} [f_{u_h^{v,t}}(t+h, X^u_{t+h})-f_{h_h^{v,t}}(t,x)])$.
		\item Moreover, we define the operator $\mathrm{A}_h^u$, $h \in \mathbb{R}$, as $(\mathrm{A}_h^u f) (t,x) =  \mathbb{E}_{t,x} [ f (t+h , X^u_{t+h}) - f(t,x)].$
	\end{enumerate}
\end{definition}
By definition, it holds $(\mathrm{A}^{\hat{u}} f) (t,x) = \lim_{h \to 0} \tfrac{1}{h} (\mathrm{A}_h^{u_h^{v,t}} f) (t,x)$.

\begin{definition} \label{equi control law cont}
	We consider a fixed control law $\hat{u}$, fix an arbitrary point $(t,x) \in \R \times \R$ with $t<T$, an $h>0$ and an arbitrary control law $u \in \mathcal{U}$. Then, we define $u_h^{v,t}$ as in Definition \ref{inf gen def cont}.
	If it holds that $\liminf_{h \rightarrow 0} \tfrac{J(t,x,\hat{u})-J(t,x,u_h^{v,t})}{h} \geq 0$ for all $u \in \mathcal{V}$, then we call the control law $\hat{u}$ a Nash equilibrium control law. If such an equilibrium exists, then we also define the equilibrium value function $V$ by $V (t,x) = J (t,x, \hat{u})$ and call $\hat{u}$ an optimal control function.
\end{definition}

Ekeland and Lazrak \cite{ekeland2006being}, and Ekeland and Pirvu \cite{ekeland2008investment} introduced this definition for Nash equilibria in continuous time. It is also used, for instance, in Bj{\"o}rk and Murgoci \cite{Bjork}, Hu et al. \cite{hu2012time}, Bj{\"o}rk et al. \cite{bjork2014mean}, and in other references for time-consistency given in the introduction. As already stated in the introduction, this generalization of Nash equilibria allows the player $t$ to form a coalition with all players $s \in [t,t+\varepsilon]$ and the equilibrium is taken by $\varepsilon \to 0$. This construction is necessary since the decision of player $t$ alone influences only the strategy at the null set $\{t\}$ and hence has no influence. 

\begin{assumption} \label{alpha >1}
	$\alpha>1$ or an optimal solution $\hat{u}$ exists which is continuous and deterministic.
\end{assumption}

This assumption is needed in the following since for $\alpha<1$, the $L^{\alpha}$ norm is only a quasi norm, and the respective space is no longer locally convex.

\subsection{Derivation of the HJB in continuous time} \label{derivation HJB}

Unlike the discrete-time case, we state the HJB equation only as a definition and prove it afterwards in a Verification Theorem. In this subsection, we show the motivation for the HJB equation: 

As in Bj{\"o}rk and Murgoci \cite{Bjork}, we assume there exists an equilibrium law $\hat{u}$ and proceed as follows: First, we take a point $(t,x)$ and an $h>0$ small. Second, we define $u_h^{v,t}$ as in Definition \ref{inf gen def cont} and third, for $h$ small enough, we expect to have: $J(t,x,u_h^{v,t}) \leq J(t,x,\hat{u})$.

Theorem \ref{hjb discrete} for the interval from $t$ to $t+h$ gives us:
\begin{align*}
	(\mathrm{A}_h^u V)(t,x) - (\mathrm{A}_h^u \mathbb{E}_{\cdot} [\T(X_T^{u_h^{v,t}} - X_{\cdot}^u e^{r(T-\cdot)})]) (t,x) + (\mathrm{A}_h^u (\lambda_{\cdot} F(\rho_{\cdot} (X_T^{u_h^{v,t}} - X_{\cdot}^u e^{r(T-\cdot)})))) (t,x) \leq 0.
\end{align*}
Next, we divide by $h$ and take the limit of $h$ to $0$:
\begin{align*}
	(\mathrm{A}^u V)(t,x) - (\mathrm{A}^u \mathbb{E}_{\cdot} [\T(X_T^{\hat{u}} - X_{\cdot}^u e^{r(T-\cdot)})]) (t,x) + (\mathrm{A}^u (\lambda_{\cdot} F(\rho_{\cdot} (X_T^{\hat{u}} - X_{\cdot}^u e^{r(T-\cdot)})))) (t,x) \leq 0.
\end{align*}

Before defining the extended HJB, we want to specify the formulas for the infinitesimal generators, which we infer using a limiting approach from the constraints of the HJB in discrete time (Theorem \ref{hjb discrete}):

\begin{lemma} \label{au exp cont}
	Let $u$ be continuous and deterministic. Moreover, let Assumption \ref{assumption expected value derivative} and \ref{ass: lambda C1} below hold. Then, we have $(\mathrm{A}^{\hat{u}} \mathbb{E}_{\cdot} [\T(X_T^{\hat{u}} - X_{\cdot}^{\hat{u}} e^{r(T-\cdot)})]) (t,x) = \tfrac{\partial}{\partial h} \mathbb{E} [ \T(X_T^{\hat{u}} - X_{t+h}^{\hat{u}} e^{r(T-t-h)})] \big|_{h=0}.$
\end{lemma}

\begin{lemma} \label{risk measure formula in continuous time}
	Let $u$ be continuous and deterministic, and $m$, $w$, and $\hat{L}$ be as in Definition \ref{def m w}. Moreover, let Assumption \ref{assumption expected value derivative} and \ref{ass: lambda C1} below hold. Then, it holds that 
	\begin{align*}
		&(\mathrm{A}^{\hat{u}} (\lambda_{\cdot} F(\rho_{\cdot} (X_T^{\hat{u}} - X_{\cdot}^{\hat{u}} e^{r(T-\cdot)})))) (t,x)
		= \lambda'_t F(-\Int_{t}^T m_s (\hat{u})\diff s + \varrho^{\hat{L}_1} \norm{\leftidx{^{\alpha}}w_s (\hat{u})}_{\L^{\alpha}((t,T], \diff s)} ) \\
		&\hspace{0.6cm}+ \lambda_t F'(-\Int_{t}^T m_s (\hat{u}) \diff s + \varrho^{\hat{L}_1} \norm{\leftidx{^{\alpha}}w_s (\hat{u})}_{\L^{\alpha}((t,T], \diff s)} ) [ m_t (\hat{u}) - \varrho^{\hat{L}_1} \leftidx{^{\alpha}}w_t (\hat{u}) ( \Int_t^T | \leftidx{^{\alpha}}w_s (\hat{u}) |^{\alpha} \diff s )^{\frac{1}{\alpha}-1} ].
	\end{align*}
\end{lemma}

The proofs of these lemmas are given in \ref{proofs}.

\subsection{The extended HJB equation in continuous time} \label{sec: extended HJB}

We define the extended HJB equation by using the formal derivation of the previous Subsection \ref{derivation HJB} and the formulas from Lemmas \ref{au exp cont} and \ref{risk measure formula in continuous time} for the constraints.

\begin{definition}[extended HJB equation] \label{hjb}
	The extended HJB equation for the Nash equilibrium problem is for $0 \leq t \leq T$ defined as
	\begin{align} \label{hjb equation}
		\sup\nolimits_{u \in \mathcal{V}} \{ (\mathrm{A}^u V)(t,x) - (\mathrm{A}^u \mathbb{E}_{\cdot} [\T(X_T^{\hat{u}} - X_{\cdot}^u e^{r(T-\cdot)})])   (t,x)\hspace{3cm}& \notag \\
		+   (\mathrm{A}^u (\lambda_{\cdot} F(\rho_{\cdot} (X_T^{\hat{u}} - X_{\cdot}^u e^{r(T-\cdot)})))) (t,x) \} &= 0,  \\
		V (T,x) &= \T(0) - \lambda_T F(0), \notag
	\end{align}
	under the constraints:
	\begin{align*}
		&(\mathrm{A}^{\hat{u}} \mathbb{E}_{\cdot} [\T(X_T^{\hat{u}} - X_{\cdot}^{\hat{u}} e^{r(T-\cdot)})]) (t,x) = \textfrac{\partial}{\partial h} \mathbb{E} [ \T(X_T^{\hat{u}} - X_{t+h}^{\hat{u}} e^{r(T-t-h)})] |_{h=0} \text{ and} \\
		&(\mathrm{A}^{\hat{u}} (\lambda_{\cdot} F(\rho_{\cdot} (X_T^{\hat{u}} - X_{\cdot}^{\hat{u}} e^{r(T-\cdot)})))) (t,x)
		= \lambda'_t F(-\Int_{t}^T m_s (\hat{u}) \diff s + \varrho^{\hat{L}_1} \norm{\leftidx{^{\alpha}}w_s (\hat{u})}_{\L^{\alpha}((t,T], \diff s)} ) \\
		&\hspace{0.6cm}+ \lambda_t F'(-\Int_{t}^T m_s (\hat{u}) \diff s + \varrho^{\hat{L}_1} \norm{\leftidx{^{\alpha}}w_s (\hat{u})}_{\L^{\alpha}((t,T], \diff s)} ) [ m_t (\hat{u}) - \varrho^{\hat{L}_1} \leftidx{^{\alpha}}w_t (\hat{u}) ( \Int_t^T | \leftidx{^{\alpha}}w_s (\hat{u}) |^{\alpha} \diff s )^{\frac{1}{\alpha}-1} ],
	\end{align*}
	with $m$, $w$, $\hat{L}$ as in Definition \ref{def m w}, and $\hat{u}$ taking the supremum in \ref{hjb equation}.
\end{definition}

\begin{remark} \label{rem: hjb constraints}
	The constraints in Definition \ref{hjb} always hold if the strategy under question is deterministic and continuous.
\end{remark}

\subsection{Properties of the optimal control function} \label{section optimal cont}

This subsection proves that the continuous-time optimal control function is again deterministic under some assumption. Note that by Assumption \ref{alpha >1}, we assume $\alpha>1$.

\begin{proposition} \label{concave J cont}
	If $u$ is deterministic, then the functional $J$ from (\ref{J def cont}) is concave in $u$. If $\T$ or $-F$ is strictly concave, $J$ is strictly concave.
\end{proposition}

\begin{proof}
	This proof is identical to the proof of Proposition \ref{Jn concave}.
\end{proof}

\begin{theorem} \label{ut deterministic}
	There exists a sequence $\leftidx{^{\delta}}{\hat{u}}$ of optimal control functions with step size $\delta$ and a $\hat{u} \in \mathcal{U}_{det}$ such that $\leftidx{^{\delta}}{\hat{u}} \xrightharpoonup{\delta \to 0} {\hat{u}}$ in $L^{\alpha}((0,T),\diff t)$ and $\limsup_{\delta \rightarrow 0 } J_t(x, \leftidx{^{\delta}}{\hat{u}} ) \leq J(t,x,{\hat{u}})$.\\
	Furthermore, if $\hat{u}$ is continuous, $\T(x)=x$, $F(x)=\max\{0,x\}^{\alpha}$ and $\rho$ shift-invariant, then $\hat{u}$ is a deterministic continuous-time equilibrium control function.
\end{theorem}

\begin{proof}
	First of all, Theorem \ref{deterministic control function} applied to the interval $(t,t+\delta]$ gives us a deterministic optimal control function, denoted by $\leftidx{^{\delta}}{\hat{u}}$, for each $\delta$. Since all $\leftidx{^{\delta}}{\hat{u}}$ are uniformly bounded in the reflexive Banach space $\L^{\alpha}((0,T],\diff t)$ by $M$ due to Proposition \ref{u unif bounded}, a corollary of the theorem of Banach-Alaoglu gives us weak convergence of $\leftidx{^{\delta}}{\hat{u}}$ by possibly switching to a subsequence. Denote the corresponding weak limit by $\tilde{u} \in \L^{\alpha}((0,T],\diff t)$. Since $M$ bounds the norm of all $\leftidx{^{\delta}}{\hat{u}} \in \tilde{\mathcal{U}}_{det}$ for all $\delta>0$, weak limits of deterministic processes are deterministic, and $\mathcal{U}$ is weakly closed since it is convex and closed, we get $\tilde{u} \in \mathcal{U}_{det}$. \\
	Next, we prove the claim $\limsup_{\delta \rightarrow 0 } J_t(x,\leftidx{^{\delta}}{\hat{u}}) \leq J(t,x,\tilde{u})$:
	Due to $\leftidx{^{\delta}}{\hat{u}}$ and $\tilde{u}$ being deterministic, the functional $J$ is concave by Proposition \ref{concave J cont}. Hence, the claim reduces to the weak upper semi-continuity for deterministic function sequences. This is equivalent to the strong upper semi-continuity in $\L^{\alpha}((0,T],\diff t)$ since $J$ is concave in $\L^{\alpha}((0,T],\diff t)$. So suppose that $\leftidx{^{\delta}}{\hat{u}}$ converges strongly to $\tilde{u}$ in $\L^{\alpha}((0,T],\diff t)$. \\
	First, we notice that by possibly switching to a subsequence $\leftidx{^{\delta}}{\hat{u}}$ converges $\diff \PP \times \diff t$-a.s. to $\tilde{u}$. Then, note that $|\leftidx{^{\delta}}{\hat{u}}_s^\intercal \sigma_s v e^{r(T-s)}|^\alpha$ ($\leq C \norm{\leftidx{^{\delta}}{\hat{u}}_s}^{\alpha}$ for some $C>0$) is dominated by a converging sequence and thus uniformly integrable.
	Hence, it follows using the exact same steps as in the proof of Proposition \ref{risk measure formula} by replacing the risk measure $\rho(\cdot)$ with equality in distribution that under Assumption \ref{ass: symmetric multidemensional} or \ref{ass: asymmetric} with $\alpha<2$ for all $t \in (0,T]$
	\begin{align*} 
		\Int_t^T \Sum_{i,j=1}^d \leftidx{^{\delta}}{\hat{u}}_{s}^i \sigma^{ij}_s e^{r(T-s)} \diff L^j_s \overset{d}{=} \sqrt[\alpha]{\Int_{t+}^T | \sqrt[\alpha]{\Int_{\S^d} | \Sum_{j=1}^d ( \Sum_{i=1}^d  \leftidx{^{\delta}}{\hat{u}}_{s}^i \sigma^{ij}_{s} e^{r(T-{s})} ) v^j |^{\alpha} \tilde{\sigma}(\diff v)} |^{\alpha} \diff s} \tilde{L}_1.
	\end{align*}
	Analog to the proof of Proposition \ref{risk measure formula}, a similar result also holds if $\alpha=2$ and in the case of Assumption \ref{ass: one dimensional}. Then, Slutsky's theorem implies that for all $i,j \in \{1,\ldots,d\}$ and for all $t \in (0,T]$
	\begin{align} \label{eq: L1+eps convergence}
		\Int_t^T \leftidx{^{\delta}}{\hat{u}}_{s}^i \sigma^{ij}_s e^{r(T-s)} \diff L^j_s \xrightarrow[\delta \to 0]{d} \Int_t^T \tilde{u}_{s}^i \sigma^{ij}_s e^{r(T-s)} \diff L^j_s,
	\end{align}
	where $\xrightarrow{d}$ denotes convergence in distribution. Now, Skohorod's representation theorem implies the existence of a common probability space such that the convergence holds almost surely.
	Furthermore, note that by possibly switching to a subsequence also $\leftidx{^{\delta}}{\hat{u}}_s^{i} (\mu^i_s-r) e^{r(T-s)}$ is a.s. convergent to $\tilde{u}^i_s (\mu^i_s-r) e^{r(T-s)}$ and uniformly integrable (by the same arguments as above). Hence, the Vitali convergence theorem implies for all $i \in \{1,\ldots,d\}$:
	\begin{align} \label{eq: convergence in the proof}
		\Int_t^T \leftidx{^{\delta}}{\hat{u}}_s^{i} (\mu^i_s-r) e^{r(T-s)} \diff s \xrightarrow{\delta \to 0} \Int_t^T \tilde{u}^i_s (\mu^i_s-r) e^{r(T-s)} \diff s.
	\end{align}
	Hence, using (\ref{eq: umschreibung xt-xn}) we obtain: $X_T^{{^{\delta}}{\hat{u}}}-X_t^{{^{\delta}}{\hat{u}}}e^{r(T-t)} \xrightarrow[\delta \to 0]{a.s.} X_T^{\tilde{u}}-X_t^{\tilde{u}}e^{r(T-t)}.$ This implies that
	\begin{align} \label{4.3 equi 1}
		\T(X_T^{{^{\delta}}{\hat{u}}}-X_t^{{^{\delta}}{\hat{u}}}e^{r(T-t)}) \xrightarrow[\delta \to 0]{a.s.} \T(X_T^{\tilde{u}}-X_t^{\tilde{u}}e^{r(T-t)})
	\end{align}
	due to the continuity of $\T$. Now we know for the positive part $\T_+$ of the function $\T$ that $\T_+(x) \leq (\T(0) + \T'(0) \cdot x)_+ \leq |\T(0)| + |\T'(0) \cdot x|$ since $\T$ is concave. Let $\varepsilon>0$ such that $1+\varepsilon<\alpha$. Then, we get:
	\begin{align*}
		\sup\limits_{\delta >0} \mathbb{E}[| X_T^{{^{\delta}}{\hat{u}}}-X_t^{{^{\delta}}{\hat{u}}}e^{r(T-t)} |^{1+\varepsilon}]
		&= \sup\limits_{\delta >0} \mathbb{E}[| \Sum_{i=1}^d \Int_t^T \leftidx{^{\delta}}{\hat{u}}^i_s (\mu^i_s - r) e^{r(T-s)}\diff s + \Sum_{i,j=1}^d \Int_t^T \leftidx{^{\delta}}{\hat{u}}^i_{s} \sigma^{ij}_s e^{r(T-s)} \diff L^j_s |^{1+\varepsilon}] \\
		&\leq \sup\limits_{\delta >0} | \Int_t^T \leftidx{^{\delta}}{\hat{u}}_s^\intercal (\mu_s - r) e^{r(T-s)}\diff s |^{1+\varepsilon} \\
		&\hspace{0.5cm}+ \left\{ \begin{matrix*}[l]
			\sup\limits_{\delta >0} \norm{\leftidx{^{\delta}}{\hat{u}}_s^\intercal \sigma_s v e^{r(T-s)}}_{\L^{\alpha}((t,T] \times \S^d, \diff s \times \tilde{\sigma}(\diff v))}^{1+\varepsilon} \mathbb{E}|\hat{L}_1|^{1+\varepsilon} & \text{if } \alpha < 2. \\
			\sup\limits_{\delta >0} \norm{\leftidx{^{\delta}}{\hat{u}}_s^\intercal \sigma_{s} e^{r(T-{s})}}_{\L^2((t,T],\diff s;H_s)}^{1+\varepsilon} \mathbb{E}|\hat{L}_1|^{1+\varepsilon} & \text{if } \alpha = 2.
		\end{matrix*}  \right. 
	\end{align*}
	For getting the second part of this estimation, we use again the exact same steps as in the proof of Proposition \ref{risk measure formula} by replacing the risk measure $\rho(\cdot)$ with $\E [|\cdot|^{1+\varepsilon}]$. 
	
	Since $\leftidx{^{\delta}}{\hat{u}}$ is uniformly bounded in $\L^{\alpha}$, the first norm in the last term is uniformly bounded by the argument from above. Hence, $X_T^{{^{\delta}}{\hat{u}}}-X_t^{{^{\delta}}{\hat{u}}}e^{r(T-t)}$ is integrable for all $\varepsilon$ with $1+\varepsilon<\alpha$.
	This property also holds for $| \T'(0) \cdot ( X_T^{{^{\delta}}{\hat{u}}}-X_t^{{^{\delta}}{\hat{u}}}e^{r(T-t)} ) |$, i.e., $\sup_{\delta >0} \mathbb{E}| \T'(0) \cdot ( X_T^{{^{\delta}}{\hat{u}}}-X_t^{{^{\delta}}{\hat{u}}}e^{r(T-t)} )|^{1+\varepsilon}  < \infty.$  
	Due to $\T_+(x) \leq |\T(0)| + | \T'(0) \cdot x |$, it follows directly that $
	\sup_{\delta >0} \mathbb{E} |\T ( X_T^{{^{\delta}}{\hat{u}}}-X_t^{{^{\delta}}{\hat{u}}}e^{r(T-t)} )_+|^{1+\varepsilon} < \infty.$
	This implies the uniform integrability of $\lbrace (\T(X_T^{{^{\delta}}{\hat{u}}}-X_t^{{^{\delta}}{\hat{u}}}e^{r(T-t)} ))_+ \rbrace_{\delta > 0}$, because of the lemma of de la Vall{\'{e}}e Poussin with $\phi(x)=x^{1+\varepsilon}$. Thus, a general version of Fatou's lemma (see Shiryaev \cite[p.4]{shiryaev2007optimal}) leads to:
	\begin{align*}
		\limsup\nolimits_{\delta \to 0} \mathbb{E}_{t,x} [ \T(X_T^{{^{\delta}}{\hat{u}}}-x e^{r(T-t)}) ] &\leq \mathbb{E}_{t,x} [ \limsup\nolimits_{\delta \to 0} \T(X_T^{{^{\delta}}{\hat{u}}}-x e^{r(T-t)}) ] \overset{(\ref{4.3 equi 1})}{=} \mathbb{E}_{t,x} [ \T(X_T^{\tilde{u}}-x e^{r(T-t)}) ].
	\end{align*}
	We must still prove the upper semi-continuity for the second part of $J$. Therefore, we start with the term inside of $F$. As above, we note that $|\leftidx{^{\delta}}{\hat{u}}_s^\intercal \sigma_s v e^{r(T-s)}|^\alpha$ (resp. $|\leftidx{^{\delta}}{\hat{u}}_{s}^\intercal \sigma_{s} e^{r(T-{s})}|^2$) is deterministic, uniformly integrable with respect to $\diff t \times \tilde{\sigma}(\diff v)$, and $\leftidx{^{\delta}}{\hat{u}}$ converges $\diff t$-a.s. to $\tilde{u}$. Hence, we get by the Vitali convergence theorem:
	\begin{align*}
		\lim\nolimits_{\delta \to 0} \norm{\leftidx{^{\delta}}{\hat{u}}_s^\intercal \sigma_s v e^{r(T-s)}}_{\L^{\alpha}((t,T] \times \S^d, \diff s \times \tilde{\sigma}(\diff v))} &= \norm{{\tilde{u}}_s^\intercal \sigma_s v e^{r(T-s)}}_{\L^{\alpha}((t,T] \times \S^d, \diff s \times \tilde{\sigma}(\diff v))}, \\
		\lim\nolimits_{\delta \to 0} \norm{\leftidx{^{\delta}}{\hat{u}}_{s}^\intercal \sigma_{s} e^{r(T-{s})}}_{\L^2((t,T],\diff s;H_s)} &= \norm{{\tilde{u}}_{s}^\intercal \sigma_{s} e^{r(T-{s})}}_{\L^2((t,T],\diff s;H_s)}.
	\end{align*}
	These equations and (\ref{eq: convergence in the proof}) imply that $\lim\nolimits_{\delta \to 0} \norm{\leftidx{^{\alpha}}w_s (\leftidx{^{\delta}}{\hat{u}})}_{\L^{\alpha}((t,T], \diff s)} = \norm{\leftidx{^{\alpha}}w_s ({\tilde{u}})}_{\L^{\alpha}((t,T], \diff s)}$ and $\lim\nolimits_{\delta \to 0} m_s(\leftidx{^{\delta}}{\hat{u}}) = m_s({\tilde{u}})$.
	Then, it follows from this combined with Proposition \ref{risk measure formula} and $\leftidx{^{\delta}}{\hat{u}}$ and ${\tilde{u}}$ being deterministic that
	\begin{align*}
		\rho_{t,x} &(X_T^{{^{\delta}}{\hat{u}}} - x e^{r(T-t)}) = -\Int_t^T m_s (\leftidx{^{\delta}}{\hat{u}}) \diff s + \varrho^{\hat{L}_1} \norm{\leftidx{^{\alpha}}w_s (\leftidx{^{\delta}}{\hat{u}})}_{\L^{\alpha}((t,T], \diff s)}, \\
		\rho_{t,x} &(X_T^{\tilde{u}} - x e^{r(T-t)}) = -\Int_t^T m_s ({\tilde{u}}) \diff s + \varrho^{\hat{L}_1} \norm{\leftidx{^{\alpha}}w_s ({\tilde{u}})}_{\L^{\alpha}((t,T], \diff s)}
	\end{align*}
	with $m$, $w$ and $\hat{L}$ as in Definition \ref{def m w}.\\
	Since we already discussed the continuity property of these integrals, and they are inserted into the continuous functions $\sqrt[\alpha]{\cdot}$ and/or $F(\cdot)$, it follows that $\lim\nolimits_{\delta \to 0} F(\rho_{t,x} (X_T^{{^{\delta}}{\hat{u}}} - x e^{r(T-t)})) = F(\rho_{t,x} (X_T^{\tilde{u}} - x e^{r(T-t)}))$.
	Combining the two parts gives us the strong upper semi-continuity and the first claim.	
	
	For the second claim, the additional assumptions imply that $J(t,x,u) = \Int_t^T u_s^\intercal (\mu_s-r) e^{r(T-s)} \diff s - \lambda_t (\varrho^{\hat{L}_1})^{\alpha} \norm{\leftidx{^{\alpha}}w_s (u)}_{\L^{\alpha}((t,T], \diff s)}^{\alpha}$ with $w$ and $\hat{L}$ as in Definition \ref{def m w} and $u$ being deterministic. In particular, this equation holds if $u$ is $\leftidx{^{\delta}}{\hat{u}}$ or $\tilde{u}$. We observe that for $F(x)=\max\{0,x\}^\alpha$, the problem is time-consistent. Now, we define by $\leftidx{_{\beta}}{v}$ a $\beta$-optimal solution for all $\beta>0$, i.e., $J(t,x,\leftidx{_{\beta}}{v}) \geq \sup_{u \in \mathcal{U}} J(t,x,u) - \beta$. Then, we define by $\leftidx{_{\beta}^{\delta}}{v}$ a discretization of $\leftidx{_{\beta}}{v}$ with step size $\delta >0$, such that $\leftidx{_{\beta}^{\delta}}{v}$ converges strongly to $\leftidx{_{\beta}}{v}$ in $\L^{\alpha}((0,T),\diff \PP \times \diff t)$ for $\delta \to 0$ and $\lim_{\delta \rightarrow 0 } J^{\delta}_t(x,\leftidx{_{\beta}^{\delta}}{v}) = J(t,x,\leftidx{_{\beta}}{v})$. Such a discretization exists due to monotone class arguments. By construction, $\leftidx{_{\beta}^{\delta}}{v}$ is piecewise constant and bounded. Because $\leftidx{^{\delta}}{\hat{u}}$ is the optimal control function for each $\delta$, we get that $J^{\delta}_t (x,\leftidx{_{\beta}^{\delta}}{v}) \leq J^{\delta}_t(x,\leftidx{^{\delta}}{\hat{u}})$ for all $\beta>0$. Taking the limes superior of $\delta \rightarrow 0$ results in $J(t,x,\leftidx{_{\beta}}{v}) \leq J(t,x,\tilde{u})$ for all $\beta>0$. Since $\leftidx{_{\beta}}{v}$ is $\beta$-optimal and the inequality holds for all $\beta>0$, $\tilde{u}$ is an optimal control function and hence an equilibrium function due to the time-consistency property of the problem.
\end{proof}

For the following theorem, we need in higher dimensions that $L$ is isotropic. That means that for $\alpha<2$, the probability measure $\tilde{\sigma}$ is the uniform distribution on $\S^d$, and for $\alpha=2$, that $R_s$ is the identity.

\begin{theorem} \label{Arzela}
	Let (a) $d=1$ and $\sigma_t \in C^1$; or (b) $L$ be isotropic, $\sigma_t = \bar{\sigma}_t \cdot Id$ with $\bar{\sigma}_t$ being a deterministic $C^1$-function which is bounded away from $0$, and $\mathcal{V} \subset \R_{\geq 0}$. Moreover, let $\mu_{min} := \min_{t \in [0,T]} \mu_t >r$, $\lambda_{min}:=\min_{t \in [0,T]} \lambda_t >0$, $\mu_t,\lambda_t \in C^1$, $\T(x)=x$, $F(x)=\max \{0,x\}^{\beta}$ with $\beta \geq \alpha > 1$, and $\rho$ shift-invariant. If either $\beta=\alpha$ or under Assumption \ref{ass: bounded strategies}, there exists a sequence $\leftidx{^{\delta}}{\hat{u}}$ of optimal control functions with step size $\delta$ and a continuous $\hat{u} \in \mathcal{U}_{det}$ such that $\leftidx{^{\delta}}{\hat{u}}_t \xrightarrow{\delta \to 0} {\hat{u}}_t$ uniformly in $t$ and $\lim_{\delta \rightarrow 0 } J_t(x, \leftidx{^{\delta}}{\hat{u}} ) = J(t,x,{\hat{u}})$.
\end{theorem}

\begin{proof} %Arxiv
	We show this statement with the theorem of Arzel{\`a}-Ascoli. Specifically, we show that $\leftidx{^{\delta}}{\hat{u}}$ is equi-continuous and bounded on the time grid $t_i^{\delta}$. Henceforth, we therefore always assume that $t, s, t+\delta, s+\delta, \ldots$ are on the time grid. In this proof, we denote by $|\cdot|$ the classical $l_\alpha$-norm for vectors. First, we note that if $L$ is isotropic and $\sigma_t = \bar{\sigma}_t \cdot Id$, then it holds that $\Int_{\S^d} |u^\intercal \sigma_s v e^{r(T-s)}|^{\alpha} \tilde{\sigma} (\diff v) = C_{\alpha d} \bar{\sigma}_t |u|^{\alpha} e^{\alpha r(T-s)}$ with $C_{\alpha d} = \frac{\Gamma(\frac{\alpha+1}{2})\Gamma(\frac{d}{2})}{ \Gamma(\frac{\alpha+d}{2})\Gamma(\frac{1}{2})}$ if $\alpha<2$ by \c{C}{\i}nlar \cite[pp.332,337,339]{Cinlar}. Additionally, we define $C_{2 d}:=1=:C_{\alpha 1}$ for $\alpha=2$ resp. $d=1$. Under this theorem's assumptions using \eqref{eq: umschreibung xt-xn}, the optimization problem in discrete time \eqref{J def} with step size $\delta$ reduces to
	\begin{align} \label{eq: J def Arzela}
		\argmax_{u \in [-M,M]^d} \{ u^\intercal \mu_\delta + \Int_{t+\delta}^T \leftidx{^{\delta}}{\hat{u}}_s^\intercal (\mu_s-r)e^{r(T-s)} \diff s - \lambda_t (\varrho^{\hat{L}_1} \sqrt[\alpha]{ \sigma_\delta |u|^{\alpha} + C_{\alpha d} \Int_{t+\delta}^T |\leftidx{^{\delta}}{\hat{u}}_s^\intercal \sigma_s e^{r (T-s)}|^{\alpha} \diff s }\,)^{\beta} \},
	\end{align}
	where $\mu_\delta := \Int_t^{t+\delta} (\mu_s-r)e^{r(T-s)} \diff s$ and $\sigma_\delta := \Int_t^{t+\delta} \sigma_s^{\alpha} e^{\alpha r (T-s)} \diff s$ if $d=1$ and $\sigma_\delta := C_{\alpha d} \cdot \Int_t^{t+\delta} \bar{\sigma}_s^\alpha e^{\alpha r (T-s)} \diff s$ if $d \geq 2$. Proposition \ref{un non-negative} (if $d=1$ and Assumption \ref{ass: symmetric multidemensional} holds) or $\mathcal{V} \subset \R_{\geq 0}$ (otherwise) imply that $u$ is non-negative. Hence, we can take the $\argmax$ over $[0,M]^d$. Now, we define
	\begin{align*}
		f_{C_1,C_2}(u) := u^\intercal \mu_\delta + C_1 - \lambda_t (\varrho^{\hat{L}_1} \sqrt[\alpha]{ \sigma_\delta |u|^{\alpha} + C_2 }\, )^{\beta}.
	\end{align*}
	We note that the $\argmax$ of $f$ is independent of $C_1$, and we get our main problem in discrete time with step size $\delta$ if $C_2^{\delta}(t) = C_{\alpha d} \Int_{t+\delta}^T |\leftidx{^{\delta}}{\hat{u}}_s^\intercal \sigma_s e^{r (T-s)}|^{\alpha} \diff s$. To show the equicontinuity of $\leftidx{^{\delta}}{\hat{u}}$ (to apply the theorem of Arzel{\`a}-Ascoli), we note that for $t<s$ with $|s-t| < \tilde{\delta}$, it holds if $\beta > \alpha$ that $|C_2^{\delta} (s)-C_2^{\delta}(t)| = C_{\alpha d} \Int_{t+\delta}^{s+\delta} |\leftidx{^{\delta}}{\hat{u}}_s^\intercal \sigma_s e^{r (T-s)}|^{\alpha} \diff s \leq C_{\alpha d} \tilde{\delta} d^\alpha M^{\alpha} \sigma_{max}^{\alpha} e^{\alpha r T}$ independent of $\delta$ using Assumption \ref{ass: bounded strategies} with $\sigma_{max} := \max_{t \in [0,T]} \sigma_t$ if $d=1$ and $\sigma_{max} := \max_{t \in [0,T]} \bar{\sigma}_t$ else. As $\lambda_t$, $\mu_t$, and $\sigma_t$ are Lipschitz continuous (since they are in $C^1$), it remains to show the equicontinuity of $\leftidx{^{\delta}}{\hat{u}}$ as a function of $C_2$, $\delta^{-1}\sigma_\delta$, $\delta^{-1}\mu_\delta$, and $\lambda_t$. 
	
	First, we consider $\leftidx{^{\delta}}{\hat{u}}$ as a function of $C_2$ and for simplicity suppress in the notation the dependence on $\delta^{-1}\sigma_\delta$, $\delta^{-1}\mu_\delta$, and $\lambda_t$. 
	Thus, from now on, we consider $u$ simply as a function of $C_2$ and set a $\delta$ in the index to express the step size directly. In the following, we prove that the optimal control function $\leftidx{^{\delta}}{\hat{u}}$ is uniformly Lipschitz continuous in $C_2$ with Lipschitz constant independent of $\delta$, which directly implies the equicontinuity in $t$. Note that $u_\delta$ denotes the $\argmax$ of $f$ for given $C_1$ and $C_2$. Thus, it holds that $\leftidx{^{\delta}}{\hat{u}}_t = u_\delta$ with the correct $C_1$ and $C_2$. Now, let $\beta>\alpha$ and $k \in \{1,\ldots,d\}$ fixed. If $u_{\delta}^i \in (0,M)$ for all $i \in \{1,\ldots,k\}$ and $u_{\delta}^i=M$ for $i >k$, then it holds for $i \in \{1,\ldots,k\}$ that 
	\begin{align} \label{eq: f prime}
		\textstyle\frac{\diff}{\diff u_{\delta}^i} f_{C_1,C_2} (u_{\delta}(C_2)) = \mu_\delta^i  - \lambda_t (\varrho^{\hat{L}_1})^{\beta} \beta (\sigma_\delta |u_{\delta} (C_2)|^{\alpha}  + C_2 )^{\frac{\beta-\alpha}{\alpha}} \sigma_\delta (u_{\delta}^i (C_2))^{\alpha-1} = 0.
	\end{align}
	Define $B_{\delta}^i := (\frac{\mu_\delta^i}{\sigma_\delta \lambda_t (\varrho^{\hat{L}_1})^{\beta} \beta})^{\frac{\alpha}{\beta-\alpha}}$, and $q:= \frac{\alpha(\alpha-1)}{\beta-\alpha} > 0$. Then rearranging \eqref{eq: f prime} implies
	\begin{align} \label{eq: B equation}
		B^i_{\delta} = ((u_{\delta}^1 (C_2))^\alpha+\ldots+(u_{\delta}^k (C_2))^\alpha+(d-k)M^\alpha)\sigma_\delta (u_{\delta}^i (C_2))^q + C_2 (u_{\delta}^i (C_2))^q.
	\end{align} 
	Since the implicit function theorem combined with Lemma \ref{arzela lemma} in the appendix implies the existence of the derivative of $u$ as a function of $C_2$, we get by taking this derivative (where we suppress again the $C_2$): $0 = (\alpha (u_{\delta}^1)^{\alpha-1} (u_{\delta}^1)' + \ldots + \alpha (u_{\delta}^k)^{\alpha-1} (u_{\delta}^k)')\sigma_\delta (u_{\delta}^i)^{q} + ((u_{\delta}^1)^\alpha+\ldots+(u_{\delta}^k)^\alpha+(d-k)M^\alpha) \sigma_\delta q (u_{\delta}^i)^{q-1} (u_{\delta}^i)' + (u_{\delta}^i)^q + C_2 q (u_{\delta}^i)^{q-1}(u_{\delta}^i)'$. Then, multiplying $u_\delta^i ((u_{\delta}^1)^\alpha+\ldots+(u_{\delta}^k)^\alpha+(d-k)M^\alpha)$, inserting \eqref{eq: B equation} twice and canceling terms leads to:
	\begin{align} \label{eq: B2 equation}
		0 =& \alpha (\textstyle\sum_{j=1}^k (u_{\delta}^j)^{\alpha-1} (u_{\delta}^j)') u_{\delta}^i (B_\delta^i - C_2 (u_{\delta}^i)^q)+ D_\delta q (u_{\delta}^i)' B_\delta^i + D_\delta (u_{\delta}^i)^{q+1} \notag \\
		=& \alpha (u_{\delta}^1)^{\alpha-1} (u_{\delta}^1)' u_{\delta}^i (B_\delta^i - C_2 (u_{\delta}^i)^q) + \alpha (\textstyle\sum_{j=2}^k (u_{\delta}^j)^{\alpha-1} (u_{\delta}^j)') u_{\delta}^i (B_\delta^i - C_2 (u_{\delta}^i)^q)+ D_\delta q (u_{\delta}^i)' B_\delta^i \notag \\
		&+ D_\delta (u_{\delta}^i)^{q+1},
	\end{align}
	where $D_\delta := \sum_{j=1}^k (u_{\delta}^j)^\alpha+(d-k)M^\alpha$. Now, we show that $(u_{\delta}^1)'$ is bounded independent of $\delta$ and $k$. First, we assume that $k \geq 2$ (the proof for $k=1$ is actually an easier version of the following). Then, we multiply the equations for $i \in \{2,\ldots,k\}$ with $(u_{\delta}^i)^{\alpha-1} u_\delta^1 (B_\delta^1-C_2 (u_\delta^1)^q) (B_\delta^i)^{-1}$ and sum over all those $i$ to receive: 
	\begin{align} \label{eq: B3 equation}
		0 =& \alpha (u_\delta^1)^{\alpha} (u_\delta^1)' (B_\delta^1-C_2 (u_\delta^1)^q) (\textstyle\sum_{i=2}^k (u_\delta^i)^\alpha (B_\delta^i-C_2 (u_\delta^i)^q)(B_\delta^i)^{-1})  \notag \\
		&+ (\alpha \textstyle\sum_{j=2}^k (u_\delta^j)^{\alpha-1} (u_\delta^j)') u_\delta^1 (B_\delta^1-C_2 (u_\delta^1)^q) (\textstyle\sum_{i=2}^k (u_\delta^i)^\alpha (B_\delta^i-C_2 (u_\delta^i)^q)(B_\delta^i)^{-1}) \notag \\
		&+ D_\delta q u_\delta^1 (B_\delta^1-C_2 (u_\delta^1)^q) (\textstyle\sum_{i=2}^k (u_\delta^i)^{\alpha-1} (u_\delta^i)') + D_\delta u_\delta^1 (B_\delta^1-C_2 (u_\delta^1)^q) (\textstyle\sum_{i=2}^k (u_\delta^i)^{q+\alpha}(B_\delta^i)^{-1}) \notag \\
		=& (\textstyle\sum_{i=2}^k (u_\delta^i)^\alpha (B_\delta^i-C_2 (u_\delta^i)^q)(B_\delta^i)^{-1}) \\
		&\qquad \cdot [\alpha (u_\delta^1)^{\alpha} (u_\delta^1)' (B_\delta^1-C_2 (u_\delta^1)^q) + \underline{\alpha (\textstyle\sum_{j=2}^k (u_\delta^j)^{\alpha-1} (u_\delta^j)') u_\delta^1 (B_\delta^1-C_2 (u_\delta^1)^q)}] \notag \\
		&+ D_\delta q \alpha^{-1} [\underline{\alpha (\textstyle\sum_{j=2}^k (u_\delta^j)^{\alpha-1} (u_\delta^j)') u_\delta^1 (B_\delta^1-C_2 (u_\delta^1)^q)}] + D_\delta u_\delta^1 (B_\delta^1-C_2 (u_\delta^1)^q) (\textstyle\sum_{i=2}^k (u_\delta^i)^{q+\alpha}(B_\delta^i)^{-1}). \notag
	\end{align}
	Then, rearranging \eqref{eq: B2 equation} for $i=1$ gives us $\alpha (\textstyle\sum_{j=2}^k (u_{\delta}^j)^{\alpha-1} (u_{\delta}^j)') u_{\delta}^1 (B_\delta^1 - C_2 (u_{\delta}^1)^q) = -\alpha (u_{\delta}^1)^{\alpha} (u_{\delta}^1)' \cdot (B_\delta^1 - C_2 (u_{\delta}^1)^q) - D_\delta q (u_{\delta}^1)' B_\delta^1 - D_\delta (u_{\delta}^1)^{q+1}$. Inserting this twice into \eqref{eq: B3 equation} (underlined terms) leads to:
	\begin{align*}
		0 =& (\textstyle\sum_{i=2}^k (u_\delta^i)^\alpha (B_\delta^i-C_2 (u_\delta^i)^q)(B_\delta^i)^{-1}) [- D_\delta q (u_{\delta}^1)' B_\delta^1 - D_\delta (u_{\delta}^1)^{q+1}]   \\
		&+ D_\delta q \alpha^{-1}[- \alpha (u_{\delta}^1)^{\alpha} (u_{\delta}^1)' (B_\delta^1 - C_2 (u_{\delta}^1)^q) -  D_\delta q (u_{\delta}^1)' B_\delta^1 - D_\delta (u_{\delta}^1)^{q+1}] \\
		&+ D_\delta u_\delta^1 (B_\delta^1-C_2 (u_\delta^1)^q) (\textstyle\sum_{i=2}^k (u_\delta^i)^{q+\alpha}(B_\delta^i)^{-1})  
	\end{align*}
	Now, canceling $D_\delta$ in each term and solving for $(u_\delta^1)'$ gives us:
	\begin{align} \label{eq: B4 equation}
		(u_\delta^1)' =& \dfrac{u_\delta^1 (B_\delta^1-C_2 (u_\delta^1)^q)(\textstyle\sum_{i=2}^k (u_\delta^i)^{q+\alpha} (B_\delta^i)^{-1})- (u_\delta^1)^{q+1} (\textstyle\sum_{i=2}^k (u_\delta^i)^{\alpha} (B_\delta^i-C_2 (u_\delta^i)^q)(B_\delta^i)^{-1})}{q B_\delta^1 (\textstyle\sum_{i=2}^k (u_\delta^i)^\alpha (B_\delta^i-C_2 (u_\delta^i)^q)(B_\delta^i)^{-1}) +q (u_\delta^1)^\alpha(B_\delta^1-C_2 (u_\delta^1)^q)+ D_\delta q^2 \alpha^{-1} B_\delta^1} \notag \\
		&- \dfrac{D_\delta \alpha^{-1}q (u_\delta^1)^{q+1}}{q B_\delta^1 (\textstyle\sum_{i=2}^k (u_\delta^i)^\alpha (B_\delta^i-C_2 (u_\delta^i)^q)(B_\delta^i)^{-1}) + q (u_\delta^1)^\alpha(B_\delta^1-C_2 (u_\delta^1)^q)+ D_\delta q^2 \alpha^{-1} B_\delta^1}.
	\end{align}
	We note that $B_\delta^i - C_2 (u_\delta^i)^q \geq 0$ for all $i$ by \eqref{eq: B equation} since the $u^i$ are non-negative. Moreover, there exist $B_d,B_u \in (0,\infty)$ independent of $\delta$ and $i$ such that $B_d \leq B_\delta^i \leq B_u$ using the original and additional assumptions of the model parameters. Combining these properties also implies that $B_\delta^i - C_2 (u_\delta^i)^q \leq B_u$. Now, since $q,\alpha > 0$, $(u_\delta^1)'$ is uniformly bounded if the nominator of the first line in \eqref{eq: B4 equation} divided by $D_\delta$ is uniformly bounded since the first two terms in the denominator are non-negative and the last term divided by $D_\delta$ is bounded away from 0 due to $q^2 \alpha^{-1} B_\delta^1 \geq q^2 \alpha^{-1} B_d > 0$ independent of $\delta$. If $k<d$, $D_\delta \geq M^\alpha$ and the claim follows directly. If $k=d$, we observe that $D_\delta = \sum_{i=1}^d (u_\delta^i)^\alpha$ and $B_\delta^i - C_2 (u_\delta^i)^q = (\sum_{j=1}^d (u_\delta^j)^\alpha) \sigma_\delta (u_{\delta}^i)^q = D_\delta \sigma_\delta (u_{\delta}^i)^q$ by \eqref{eq: B equation}. Hence, the $D_\delta$ cancels out and the nominator is uniformly bounded due to $(B_\delta^i )^{-1}$, $B_\delta^i - C_2 (u_\delta^i)^q$, $u_\delta^i$, and $D_\delta$ being all bounded independent of $\delta$. By symmetry, we can show analogously that $(u_\delta^i)'$ is uniformly bounded for every $i \in \{2,\ldots,k\}$. \\
	From \eqref{eq: f prime}, it follows that $\frac{\diff}{\diff u_\delta^i}f(0) = \mu_\delta^i >0$. Thus, the $\argmax$ cannot be attained at $u_{\delta}^i=0$. Since the $u^i$ are uniform Lipschitz continuous in $C_2$ and constant if $u^i=M$, we get the uniform Lipschitz continuity of $u_{\delta}$. Hence, $u_{\delta}$ is equicontinuous for $\beta>\alpha$. If $\alpha=\beta$, we conclude from \eqref{eq: f prime} that ${\leftidx{^\delta}{\hat{u}}^i_t} = \sqrt[\alpha-1]{\frac{\mu_\delta^i}{\sigma_\delta \alpha \lambda_t (\varrho^{\hat{L}_1})^{\beta}}}$ (possibly being additional truncated at $M$ under Assumption \ref{ass: bounded strategies}). Now since $\lambda$, $\delta^{-1} \sigma_\delta$, and $\delta^{-1} \mu_\delta$ are in $\delta$ uniformly bounded away from $0$ and uniformly bounded, $u$ is equicontinuous.
	
	It remains to show the uniform boundedness in $t$. If $\beta\neq \alpha$, this property is given by Assumption \ref{ass: bounded strategies}. Now, we define $\sigma_{min} := \min_{t \in [0,T]} \sigma_t$ if $d=1$ resp. $\sigma_{min} := C_{\alpha d} \min_{t \in [0,T]} \bar{\sigma}_t$ if $d \geq 1$. If $\alpha=\beta$, we conclude as before that ${\leftidx{^\delta}{\hat{u}}_t} = \sqrt[\alpha-1]{\frac{\mu_\delta^i}{\sigma_\delta \alpha \lambda_t (\varrho^{\hat{L}_1})^{\beta}}} \leq \sqrt[\alpha-1]{\frac{\max_{t \in [0,T], i \in \{1,\ldots,d\}} (\mu^i_t-r) e^{rT}}{\sigma_{min}^\alpha \alpha \lambda_{min} (\varrho^{\hat{L}_1})^{\beta}}} < \infty$ independent of $\delta$. Hence, $u$ is uniformly bounded.
	
	The theorem of Arzel{\`a}-Ascoli implies the existence of a uniform convergent subsequence to a continuous limit. The convergence of the value functional $J$ for this subsequence follows directly since $J$ is continuous in the control function under the theorem's assumptions.
	
	Indeed for $\delta^{-1}\sigma_\delta$:
	It holds: For the equicontinuity of $\leftidx{^{\delta}}{\hat{u}}$, we note that for $t<s$ with $|s-t| < \tilde{\delta}$, it holds if $\beta > \alpha$ that $|\frac{C_2^\delta (s)}{\delta^{-1} \sigma_\delta(s)}-\frac{C_2^\delta (t)}{\delta^{-1} \sigma_\delta(t)}| = |\frac{C_2^\delta (s)\delta^{-1} \sigma_\delta(t)-C_2^\delta (t)\delta^{-1} \sigma_\delta(s)}{\delta^{-1} \sigma_\delta(s)\delta^{-1} \sigma_\delta(t)}| \leq \frac{|C_2^\delta (s)| \, |\delta^{-1} \sigma_\delta (t)-\delta^{-1} \sigma_\delta (s)|}{|\delta^{-1} \sigma_\delta(s)\delta^{-1} \sigma_\delta(t)|} + \frac{ |\delta^{-1} \sigma_\delta (s)| \, |C_2^\delta (s)-C_2^\delta (t)|}{|\delta^{-1} \sigma_\delta(s)\delta^{-1} \sigma_\delta(t)|} \leq \frac{C_{\alpha d}^2 T M^\alpha \sigma_{max}^\alpha e^{\alpha rT}}{(\sigma_{min}^\alpha C_{\alpha d})^2} |s-t| \frac{|\delta^{-1} \sigma_\delta (t)-\delta^{-1} \sigma_\delta (s)|}{|s-t|} + \frac{C_{\alpha d}^2 \sigma_{max}^\alpha e^{\alpha rT}}{(\sigma_{min}^\alpha C_{\alpha d})^2}  \Int_{t+\delta}^{s+\delta} |\leftidx{^{\delta}}{\hat{u}}_s^\intercal \sigma_s e^{r (T-\tau)}|^{\alpha} \diff \tau \leq K \tilde{\delta} TM^\alpha \frac{(\sigma_{max}^{\alpha})}{(\sigma_{min}^{\alpha})^2} e^{\alpha rT} + \tilde{\delta} d^\alpha M^{\alpha} \frac{(\sigma_{max}^{\alpha})^2}{(\sigma_{min}^{\alpha})^2} e^{2 \alpha r T}$ with $K$ such that $\frac{|\delta^{-1} \sigma_\delta (t)-\delta^{-1} \sigma_\delta (s)|}{|s-t|} \leq K$ (exists since $\delta^{-1} \sigma_\delta(t)$ has a bounded derivative independent of $\delta$ since for $d=1$ it holds (if $d \geq 2$: replace $\sigma$ by $\tilde{\sigma}$): $|\frac{\diff \delta^{-1} \sigma_\delta(t)}{\diff t}| = |C_{\alpha d} \delta^{-1} (\sigma_{t+\delta}^\alpha e^{\alpha r (T-(t+\delta))}- \sigma_{t}^\alpha e^{\alpha r (T-t)})| = |C_{\alpha d} \alpha \sigma_\xi^{\alpha-1} e^{\alpha r (T-\xi)} + C_{\alpha d} \sigma_\xi^\alpha (-\alpha r) e^{\alpha r (T-\xi)}| \leq C_{\alpha d} \alpha \sigma_{max}^{\alpha-1} e^{\alpha r T} + C_{\alpha d} \sigma_{max}^\alpha \alpha r e^{\alpha rT}$ with $\xi \in (t,t+\delta)$). The right-hand side is independent of $\delta$ using Assumption \ref{ass: bounded strategies} with $\sigma_{max} := \max_{t \in [0,T]} \sigma_t$ if $d=1$ resp. $\sigma_{max} := \max_{t \in [0,T]} \bar{\sigma}_t$ else and $\sigma_{min}$ defined analogously. Note that $\sigma_{min}>0$ by assumption. Now, we can replace $f$ by $(f-C_1) (\frac{1}{\delta^{-1} \sigma_\delta})^\frac{\beta}{\alpha}$ which has the same $\argmax$ as $f$ and then use that $\delta^{-1}\sigma_\delta$ is Lipschitz (with Lipschitz constant independent of $\delta$), bounded, and bounded away from 0 so that we can replace $C_2$ by $\frac{C_2}{\delta^{-1} \sigma_\delta}$ and $\mu_\delta$ by $\tilde{\mu}_\delta := {\mu_\delta}{(\delta^{-1} \sigma_\delta)}^{-{\beta}/{\alpha}}$.
	
	Indeed for $\delta^{-1}\mu_\delta^i$ and $\lambda_t$:
	Equal until:\\
	Define $B_{\delta}^i := (\frac{\mu_\delta^i}{\sigma_\delta \lambda_t (\varrho^{\hat{L}_1})^{\beta} \beta})^{\frac{\alpha}{\beta-\alpha}}$, and $q:= \frac{\alpha(\alpha-1)}{\beta-\alpha} > 0$. Since $B_\delta^i$ is Lipschitz continuous in $t$ since all terms in its definition are Lipschitz continuous in $t$, bounded, and bounded away from 0, we show that $u_\delta^i$ is Lipschitz continuous in $B_\delta^i$. Then rearranging \eqref{eq: f prime} implies
	\begin{align} \label{eq: B equation 2}
		B^i_{\delta} = ((u_{\delta}^1 (\cdot))^\alpha+\ldots+(u_{\delta}^k (\cdot))^\alpha+(d-k)M^\alpha)\sigma_\delta (u_{\delta}^i (\cdot))^q + C_2 (u_{\delta}^i (\cdot))^q.
	\end{align} 
	Since the implicit function theorem combined with Lemma \ref{arzela lemma} in the appendix implies the existence of the derivative of $u$ as a function of $B_\delta^i$, we get by taking this derivative (and with a slight abuse of notation again denoting it with $u'$): $1 = (\alpha (u_{\delta}^1)^{\alpha-1} (u_{\delta}^1)' + \ldots + \alpha (u_{\delta}^k)^{\alpha-1} (u_{\delta}^k)')\sigma_\delta (u_{\delta}^i)^{q} + ((u_{\delta}^1)^\alpha+\ldots+(u_{\delta}^k)^\alpha+(d-k)M^\alpha) \sigma_\delta q (u_{\delta}^i)^{q-1} (u_{\delta}^i)' +  C_2 q (u_{\delta}^i)^{q-1}(u_{\delta}^i)'$. Then, multiplying $u_\delta^i ((u_{\delta}^1)^\alpha+\ldots+(u_{\delta}^k)^\alpha+(d-k)M^\alpha)$, inserting \eqref{eq: B equation 2} twice and canceling terms leads to:
	\begin{align} \label{eq: B2 equation 2}
		0 =& \alpha (\textstyle\sum_{j=1}^k (u_{\delta}^j)^{\alpha-1} (u_{\delta}^j)') u_{\delta}^i (B_\delta^i - C_2 (u_{\delta}^i)^q)+ D_\delta q (u_{\delta}^i)' B_\delta^i - D_\delta u_{\delta}^i \notag \\
		=& \alpha (u_{\delta}^1)^{\alpha-1} (u_{\delta}^1)' u_{\delta}^i (B_\delta^i - C_2 (u_{\delta}^i)^q) + \alpha (\textstyle\sum_{j=2}^k (u_{\delta}^j)^{\alpha-1} (u_{\delta}^j)') u_{\delta}^i (B_\delta^i - C_2 (u_{\delta}^i)^q)+ D_\delta q (u_{\delta}^i)' B_\delta^i \notag \\ &- D_\delta u_{\delta}^i 
	\end{align}
	where $D_\delta := \sum_{j=1}^k (u_{\delta}^j)^\alpha+(d-k)M^\alpha$. Now, we show that $(u_{\delta}^1)'$ is bounded independent of $\delta$ and $k$. First, we assume that $k \geq 2$ (the proof for $k=1$ is actually an easier version of the following). Then, we multiply the equations for $i \in \{2,\ldots,k\}$ with $(u_{\delta}^i)^{\alpha-1} u_\delta^1 (B_\delta^1-C_2 (u_\delta^1)^q) (B_\delta^i)^{-1}$ and sum over all those $i$ to receive: 
	\begin{align} \label{eq: B3 equation 2}
		0 =& \alpha (u_\delta^1)^{\alpha} (u_\delta^1)' (B_\delta^1-C_2 (u_\delta^1)^q) (\textstyle\sum_{i=2}^k (u_\delta^i)^\alpha (B_\delta^i-C_2 (u_\delta^i)^q)(B_\delta^i)^{-1})  \notag \\
		&+ (\alpha \textstyle\sum_{j=2}^k (u_\delta^j)^{\alpha-1} (u_\delta^j)') u_\delta^1 (B_\delta^1-C_2 (u_\delta^1)^q) (\textstyle\sum_{i=2}^k (u_\delta^i)^\alpha (B_\delta^i-C_2 (u_\delta^i)^q)(B_\delta^i)^{-1}) \notag \\
		&+ D_\delta q u_\delta^1 (B_\delta^1-C_2 (u_\delta^1)^q) (\textstyle\sum_{i=2}^k (u_\delta^i)^{\alpha-1} (u_\delta^i)') - D_\delta u_\delta^1 (B_\delta^1-C_2 (u_\delta^1)^q) (\textstyle\sum_{i=2}^k (u_\delta^i)^{\alpha}(B_\delta^i)^{-1}) \notag \\
		=& (\textstyle\sum_{i=2}^k (u_\delta^i)^\alpha (B_\delta^i-C_2 (u_\delta^i)^q)(B_\delta^i)^{-1}) \\
		&\qquad \cdot [\alpha (u_\delta^1)^{\alpha} (u_\delta^1)' (B_\delta^1-C_2 (u_\delta^1)^q) + \underline{\alpha (\textstyle\sum_{j=2}^k (u_\delta^j)^{\alpha-1} (u_\delta^j)') u_\delta^1 (B_\delta^1-C_2 (u_\delta^1)^q)}] \notag \\
		&+ D_\delta q \alpha^{-1} [\underline{\alpha (\textstyle\sum_{j=2}^k (u_\delta^j)^{\alpha-1} (u_\delta^j)') u_\delta^1 (B_\delta^1-C_2 (u_\delta^1)^q)}] - D_\delta u_\delta^1 (B_\delta^1-C_2 (u_\delta^1)^q) (\textstyle\sum_{i=2}^k (u_\delta^i)^{\alpha}(B_\delta^i)^{-1}). \notag
	\end{align}
	Then, rearranging \eqref{eq: B2 equation 2} for $i=1$ gives us $\alpha (\textstyle\sum_{j=2}^k (u_{\delta}^j)^{\alpha-1} (u_{\delta}^j)') u_{\delta}^1 (B_\delta^1 - C_2 (u_{\delta}^1)^q) = -\alpha (u_{\delta}^1)^{\alpha} (u_{\delta}^1)' \cdot (B_\delta^1 - C_2 (u_{\delta}^1)^q) - D_\delta q (u_{\delta}^1)' B_\delta^1 + D_\delta (u_{\delta}^1)$. Inserting this twice into \eqref{eq: B3 equation 2} (underlined terms) leads to:
	\begin{align*}
		0 =& (\textstyle\sum_{i=2}^k (u_\delta^i)^\alpha (B_\delta^i-C_2 (u_\delta^i)^q)(B_\delta^i)^{-1}) [- D_\delta q (u_{\delta}^1)' B_\delta^1 + D_\delta u_{\delta}^1 ]   \\
		&+ D_\delta q \alpha^{-1}[- \alpha (u_{\delta}^1)^{\alpha} (u_{\delta}^1)' (B_\delta^1 - C_2 (u_{\delta}^1)^q) -  D_\delta q (u_{\delta}^1)' B_\delta^1 + D_\delta u_{\delta}^1 ] \\
		&- D_\delta u_\delta^1 (B_\delta^1-C_2 (u_\delta^1)^q) (\textstyle\sum_{i=2}^k (u_\delta^i)^{\alpha}(B_\delta^i)^{-1})  
	\end{align*}
	Now, canceling $D_\delta$ in each term and solving for $(u_\delta^1)'$ gives us:
	\begin{align} \label{eq: B4 equation 2}
		(u_\delta^1)' =& \dfrac{u_\delta^1 (B_\delta^1-C_2 (u_\delta^1)^q)(\textstyle\sum_{i=2}^k (u_\delta^i)^{\alpha} (B_\delta^i)^{-1})+ u_\delta^1  (\textstyle\sum_{i=2}^k (u_\delta^i)^{\alpha} (B_\delta^i-C_2 (u_\delta^i)^q)(B_\delta^i)^{-1})}{q B_\delta^1 (\textstyle\sum_{i=2}^k (u_\delta^i)^\alpha (B_\delta^i-C_2 (u_\delta^i)^q)(B_\delta^i)^{-1}) +q (u_\delta^1)^\alpha(B_\delta^1-C_2 (u_\delta^1)^q)+ D_\delta q^2 \alpha^{-1} B_\delta^1} \notag \\
		&+ \dfrac{D_\delta \alpha^{-1}q u_\delta^1 }{q B_\delta^1 (\textstyle\sum_{i=2}^k (u_\delta^i)^\alpha (B_\delta^i-C_2 (u_\delta^i)^q)(B_\delta^i)^{-1}) + q (u_\delta^1)^\alpha(B_\delta^1-C_2 (u_\delta^1)^q)+ D_\delta q^2 \alpha^{-1} B_\delta^1}.
	\end{align}
	We note that $B_\delta^i - C_2 (u_\delta^i)^q \geq 0$ for all $i$ by \eqref{eq: B equation 2} since the $u^i$ are non-negative. Moreover, there exist $B_d,B_u \in (0,\infty)$ independent of $\delta$ and $i$ such that $B_d \leq B_\delta^i \leq B_u$ using the original and additional assumptions of the model parameters. Combining these properties also implies that $B_\delta^i - C_2 (u_\delta^i)^q \leq B^u$. Now, since $q,\alpha > 0$, $(u_\delta^1)'$ is uniformly bounded if and only if the nominator of the first line in \eqref{eq: B4 equation 2} divided by $D_\delta$ is uniformly bounded. If $k<d$, $D_\delta \geq M^\alpha$ and the claim follows directly. If $k=d$, we observe that $D_\delta = \sum_{i=1}^d (u_\delta^i)^\alpha$ and $B_\delta^i - C_2 (u_\delta^i)^q = (\sum_{j=1}^d (u_\delta^j)^\alpha) \sigma_\delta (u_{\delta}^i)^q = D_\delta \sigma_\delta (u_{\delta}^i)^q$ by \eqref{eq: B equation 2}. Hence, the $D_\delta$ cancels out in the last term of the denominator (and the other terms are non-negative) when dividing the nominator by $D_\delta$ which implies that $(u_\delta^1)'$ is uniformly bounded since $q^2 \alpha^{-1} B_\delta^1 \geq q^2 \alpha^{-1} B_d > 0$ independent of $\delta$. Note that the nominator is uniformly bounded due to $(B_\delta^i )^{-1}$, $B_\delta^i - C_2 (u_\delta^i)^q$, $u_\delta^i$, and $D_\delta$ being all bounded independent of $\delta$. By symmetry, we can show analogously that $(u_\delta^i)'$ is uniformly bounded for every $i \in \{2,\ldots,k\}$.
\end{proof}

\subsection{Verification theorem} \label{sec: verification}

To prove the verification theorem, we need the following assumption:

\begin{assumption} \label{assumption expected value derivative}
	If $\mathcal{V} \subset \R^d_{\geq 0}$ (resp. $\mathcal{V} \subset \R^d_{\leq 0}$), it holds that $\E [|L_t \T'(a+b L_t)|] < \infty$ for all $a\in \R$, $b>0$ (resp. $b<0$), and for all $t \in (0,T]$. Otherwise, it holds that $\E [|L_t \T'(a+b L_t)|] < \infty$ for all $a,b \in \R$ and for all $t \in (0,T]$.
\end{assumption}

\begin{assumption} \label{ass: lambda C1}
	$\lambda_t \in C^{1}$.
\end{assumption}

\begin{proposition} \label{discussion continuous}
	Let $u$ be deterministic and continuous. Then, it holds that $(\lambda_{\cdot} F(\rho_{\cdot} (X_T^{u} - X_{\cdot}^{u} e^{r(T-{\cdot})})))  \in C^{1,2}$ and $\mathbb{E}_{\cdot} [\T(X_T^{u} - X_{\cdot}^{u} e^{r(T-{\cdot})})] \in C^{1,2}$.
\end{proposition}

\begin{lemma} \label{lemma: infinitesimal generator = derivative}
	Let $u$ be continuous and deterministic. If the infinitesimal generator from Definition \ref{inf gen def cont} is build with $v=u_t$, then it holds:
	\begin{align*}
		(\mathrm{A}^u \mathbb{E}_{\cdot} [\T(X_T^u - X_{\cdot}^u e^{r(T-\cdot)})]) (t,x) &= \tfrac{\diff}{\diff t} \mathbb{E}_{t,x} [\T(X_T^u-x e^{r(T-t)})] ,\\
		(\mathrm{A}^u (\lambda_{\cdot}F(\rho_{\cdot} (X_T^u - X_{\cdot}^u e^{r(T-\cdot)})))) (t,x) &= \tfrac{\diff}{\diff t} \lambda_t F(\rho_{t,x} (X_T^u - x e^{r(T-t)})).
	\end{align*}
\end{lemma}

We give these proofs in \ref{proofs}.

\begin{assumption} \label{ass: V C1}
	$V \in C^{1,2}$.
\end{assumption}

\begin{assumption} \label{ass: cont}
	The optimal equilibrium function $\hat{u}$ is continuous and deterministic.
\end{assumption}

If the assumptions of Theorem \ref{Arzela} are fulfilled, this assumption is redundant.

Following the steps of Bj{\"o}rk and Murgoci \cite{Bjork}, we prove that our extended HJB equation, which we derived by a limiting approach, is correct.

\begin{theorem} \label{verification}
	Assume $\tilde{V}$ is a solution of the extended HJB, and the control law $\hat{u}$ attains the supremum in Equation (\ref{hjb equation}). Then $\hat{u}$ is an equilibrium law, and $\tilde{V}$ is the corresponding value function.
\end{theorem}

\begin{proof}
	\underline{Step 1: Let us show that $\tilde{V}(t,x)=J(t,x,\hat{u})$:}
	
	It holds by assumption that
	\begin{align} \label{eq: verification theorem eq 1}
		(\mathrm{A}^{\hat{u}} \tilde{V})(t,x) - (\mathrm{A}^{\hat{u}} \mathbb{E}_{\cdot} [\T(X_T^{\hat{u}} - X_{\cdot}^{\hat{u}} e^{r(T-\cdot)})]) (t,x) + (\mathrm{A}^{\hat{u}} (\lambda_{\cdot}F(\rho_{\cdot} (X_T^{\hat{u}} - X_{\cdot}^{\hat{u}} e^{r(T-\cdot)})))) (t,x) = 0 .
	\end{align}
	Next, we replace $(t,x)$ by $(s,X_s^{\hat{u}})$, integrate over $\diff s$ in $(t,T)$, and take the conditional expectation. Now, from Dynkin's theorem, which we are allowed to use due to Lemma \ref{lemma: infinitesimal generator = derivative} and Assumption \ref{ass: V C1}, we conclude that $\mathbb{E}_{t,x} [\tilde{V}(T,X_T^{\hat{u}})]=\tilde{V}(t,x) + \mathbb{E}_{t,x} [ \Int_t^T (\mathrm{A}^{\hat{u}} \tilde{V}) (s,X_s^{\hat{u}}) \diff s ].$
	Due to Assumption \ref{ass: cont} (which can be replaced by the assumptions in Theorem \ref{Arzela}) and Proposition \ref{discussion continuous}, we use Dynkin again for the other two terms from \eqref{eq: verification theorem eq 1}. Solving this equation for $\tilde{V}$ leads to 
	\begin{align*}
		\tilde{V}(t,x) =& \mathbb{E}_{t,x} [\tilde{V} (T,X_T^{\hat{u}})] - \mathbb{E}_{t,x} [ \T(X_T^{\hat{u}} - X_{\cdot}^{\hat{u}} e^{r(T-\cdot)}) (T,X_T^{\hat{u}}) ] + \mathbb{E}_{t,x} [ \T(X_T^{\hat{u}} - X_{\cdot}^{\hat{u}} e^{r(T-\cdot)}) (t,x) ] \\
		&+ \mathbb{E}_{t,x} [\lambda_T F(\rho (X_T^{\hat{u}} - X_{\cdot}^{\hat{u}} e^{r(T-\cdot)})) (T,X_T^{\hat{u}}) ] - \lambda_t F(\rho (X_T^{\hat{u}} - X_{\cdot}^{\hat{u}} e^{r(T-\cdot)})) (t,x).
	\end{align*}
	Lastly, we insert the final value to infer:
	\begin{align*}
		\tilde{V}(t,x) =& \T(0) - \lambda_T F(0) -\T(0) + \mathbb{E}_{t,x} [ \T(X_T^{\hat{u}} - x e^{r(T-t)}) ] + \lambda_T F(0) - \lambda_t F(\rho (X_T^{\hat{u}} - x e^{r(T-t)})) \\
		=& \mathbb{E}_{t,x} [ \T(X_T^{\hat{u}} - x e^{r(T-t)}) ] - \lambda_t F(\rho (X_T^{\hat{u}} - x e^{r(T-t)})) = J(t,x,\hat{u}).
	\end{align*}
	
	\underline{Step 2: Let us show that $\hat{u}$ is an equilibrium function:}\\
	First, we construct a control law $u_h^{v,t}$ as in Definition \ref{inf gen def}. Now, \eqref{rec J}
	applied to the interval $t$ to $t+h$ gives us (suppressing the $h$ and the ${v,t}$ in $u_h$):
	\begin{align} \label{eq:veri 1}
		J (t,x,u) =& \mathbb{E}_{t,x} [\tilde{V} (t+h,X_{t+h}^u)] -\mathbb{E}_{t,x} [\T(X_T^u-X_{t+h}^u e^{r(T-t-h)})-\T(X_T^u-x e^{r(T-t)})] \notag \\
		&+ \mathbb{E}_{t,x} [\lambda_{t+h} F(\rho_{t+h,X_{t+h}^u} (X_T^u - X_{t+h}^u e^{r(T-t-h)})) - \lambda_t F(\rho_{t,x} (X_T^u - x e^{r(T-t)}))],
	\end{align}
	where we used $J (t+h,X_{t+h}^u,u)=\tilde{V} (t+h,X_{t+h}^u)$, which follows from the definition of Nash equilibria (see Definition \ref{equi control law cont}).
	Since $\tilde{V}$ is a solution of the extended HJB as in Definition \ref{hjb}, it holds that for all $u \in \mathcal{U}$
	\begin{align*}
		(\mathrm{A}^u \tilde{V} )(t,x) - (\mathrm{A}^u \T(X_T^u - X_{\cdot}^u e^{r(T-\cdot)})) (t,x) + (\mathrm{A}^u (\lambda_{\cdot} F(\rho_{\cdot} (X_T^u - X_{\cdot}^u e^{r(T-\cdot)})))) (t,x) \leq 0.
	\end{align*}
	A discretization of this equation leads to
	\begin{align} \label{eq:veri 2}
		\mathbb{E}_{t,x} [\tilde{V} (t+h,X_{t+h}^u)] - \tilde{V} (t,x) - \mathbb{E}_{t,x} [\T(X_T^u-X_{t+h}^u e^{r(T-t-h)})-\T(X_T^u-x e^{r(T-t)})] \notag \\
		+ \mathbb{E}_{t,x} [\lambda_{t+h} F(\rho_{t+h,X_{t+h}^u} (X_T^u - x e^{r(T-t-h)})) - \lambda_t F(\rho_{t,x} ( X_T^u - x e^{r(T-t)}))] &\leq \mathrm{o}(h).
	\end{align}
	Since $\tilde{V}(t,x)=J(t,x,\hat{u})$, it follows that $J(t,x,\hat{u}) \geq J(t,x,u) + \mathrm{o} (h)$ from (\ref{eq:veri 1}) and (\ref{eq:veri 2}), which implies that $\liminf_{h \rightarrow 0} \tfrac{J(t,x,\hat{u}) - J(t,x,u)}{h} \geq 0.$
	Definition \ref{equi control law cont} yields the proof.
\end{proof}

Now, we come back to the setting of Theorem \ref{Arzela} to show that this limit of the deterministic equilibria is also an optimal equilibrium in continuous time.

\begin{theorem}
	Let the assumptions of Theorem \ref{Arzela} hold. Then, $\hat{u}$ from Theorem \ref{Arzela} is an optimal equilibrium function.
\end{theorem}

\begin{proof}
	First, note that for convenience of the reader, we denote $\hat{u}$ from Theorem \ref{Arzela} by $\tilde{u}$.
	For this proof, we rewrite the HJB equations in discrete and continuous time using the special setting properties given by the assumptions. We remark that for continuous and deterministic optimal control functions, the infinitesimal generator equals the derivative with respect to $t$ using Lemma \ref{lemma: infinitesimal generator = derivative}. Now, similar to the calculations in the proof of Proposition \ref{discussion continuous}, we get for $\bar{u}$ with $\bar{u}_s = \hat{u}_s$ for $s>t$, where $\hat{u}$ denotes the optimal solution, and $\bar{u}_t = u$ for an arbitrary, but fixed $u \in \mathcal{V}$:
	\begin{align*}
		\tfrac{\diff}{\diff t}J (t,x,\bar{u}) 
		=& \,-u^\intercal (\mu_t-r)e^{r(T-t)} + \lambda_t (\varrho^{\hat{L}_1} \sqrt[\alpha]{C_{\alpha d}})^{\beta} \tfrac{\beta}{\alpha} (\Int_t^T |\hat{u}_s^\intercal \sigma_s e^{r(T-s)}|^{\alpha} \diff s)^{\frac{\beta-\alpha}{\alpha}} |u^\intercal \sigma_te^{r(T-t)}|^{\alpha} \\
		&- \lambda_t' (\varrho^{\hat{L}_1} \sqrt[\alpha]{C_{\alpha d}})^{\beta} (\Int_t^T |\hat{u}_s^\intercal \sigma_s e^{r(T-s)}|^{\alpha} \diff s)^{\frac{\beta}{\alpha}},
	\end{align*}
	where $C_{\alpha d}$ is defined as in the proof of Theorem \ref{Arzela}. Here, we used that $L$ is isotropic or $d=1$. Hence, the first part of the continuous HJB equation \eqref{hjb equation} becomes after noting that $V(t,x)=J(t,x,\hat{u})$ for $\hat{u}$ being optimal (since the $\lambda_t'$-part cancels):
	\begin{align} \label{eq: hjb convergence 1}
		\textstyle\sup\limits_{u \in \mathcal{V}} \{ &(u^\intercal-\hat{u}_t^\intercal) (\mu_t-r)e^{r(T-t)} \notag \\
		&- \lambda_t (\varrho^{\hat{L}_1}\sqrt[\alpha]{C_{\alpha d}})^{\beta} \tfrac{\beta}{\alpha} (\Int_t^T |\hat{u}_s^\intercal \sigma_s e^{r(T-s)}|^{\alpha} \diff s)^{\frac{\beta-\alpha}{\alpha}} (|u^\intercal \sigma_te^{r(T-t)}|^{\alpha}-|\hat{u}_t^\intercal\sigma_te^{r(T-t)}|^{\alpha}) \} = 0.
	\end{align}
	The other two parts are similar. Thus, we note that $\tilde{u}$ solves the HJB equation \eqref{hjb equation} if and only if the previous equation holds with $\tilde{u}$ instead of $\hat{u}$.
	Similarly, we can reduce the first part of the discrete HJB equation \eqref{hjb equation disc} to:
	\begin{align*} 
		\textstyle\sup\limits_{u \in \mathcal{V}} \{ (u^\intercal-\leftidx{^\delta}{\hat{u}}_t^\intercal) &\Int_t^{t+\delta}(\mu_s-r)e^{r(T-s)} \diff s \\
		&-\lambda_t (\varrho^{\hat{L}_1} \sqrt[\alpha]{C_{\alpha d}})^{\beta} ((|u|^{\alpha} \Int_t^{t+\delta}|\sigma_s^{11} e^{r(T-s)}|^{\alpha} \diff s + \Int_{t+\delta}^T |\leftidx{^\delta}{\hat{u}}_s^\intercal \sigma_s e^{r(T-s)}|^{\alpha} \diff s)^{\frac{\beta}{\alpha}} \notag \\ &- (|\leftidx{^\delta}{\hat{u}}_t|^{\alpha} \Int_t^{t+\delta}|\sigma_s^{11} e^{r(T-s)}|^{\alpha} \diff s + \Int_{t+\delta}^T |\leftidx{^\delta}{\hat{u}}_s^\intercal \sigma_s e^{r(T-s)}|^{\alpha} \diff s)^{\frac{\beta}{\alpha}}) \} = 0.
	\end{align*}
	Here, we also used that $|\hat{u}_s^\intercal \sigma_s e^{r(T-s)}|^{\alpha} = |\hat{u}_s|^\alpha |\sigma_s^{11} e^{r(T-s)}|^{\alpha}$ since either $d=1$ or $\sigma$ is the identity times possibly a constant. Again, the two other parts are similar.
	Now, we divide the equation by $\delta$ and show that for every arbitrary but fixed $u \in \mathcal{V}$ by taking the limit of $\delta \to 0$, we get the continuous HJB with $\tilde{u}$ instead of $\hat{u}$ on the left side. Note that for every fixed $u \in \mathcal{V}$ the discrete HJB equation holds with a ``$\leq 0$'' instead of a ``$=0$''.\\
	First, it holds that $\lim_{\delta \to 0} \delta^{-1}(u^\intercal-\leftidx{^\delta}{\hat{u}}_t^\intercal) \Int_t^{t+\delta}(\mu_s-r)e^{r(T-s)} \diff s = (u^\intercal-\tilde{u}_t^\intercal) (\mu_t-r)e^{r(T-t)}$ since $\leftidx{^\delta}{\hat{u}}$ converges to $\tilde{u}$ by Theorem \ref{Arzela} and $\mu$ is continuous. For the second term, the argumentation is trickier, and for the sake of simplicity we introduce $\sigma_\delta := \Int_t^{t+\delta}|\sigma_s^{11} e^{r(T-s)}|^{\alpha} \diff s$ and $C_{\delta} (x):= \Int_{t+\delta}^T |x_s^\intercal \sigma_s e^{r(T-s)}|^{\alpha} \diff s$. Now, we get
	\begin{align*}
		&\frac{(|u|^{\alpha} \sigma_\delta + C_{\delta} (\leftidx{^\delta}{\hat{u}}))^{\frac{\beta}{\alpha}} - (|\leftidx{^\delta}{\hat{u}}_t|^{\alpha}\sigma_\delta + C_{\delta} (\leftidx{^\delta}{\hat{u}}))^{\frac{\beta}{\alpha}}}{\delta} \\ &\qquad= \frac{(|u|^{\alpha} \sigma_\delta + C_{\delta} (\leftidx{^\delta}{\hat{u}}))^{\frac{\beta}{\alpha}} - (|\leftidx{^\delta}{\hat{u}}_t|^{\alpha}\sigma_\delta + C_{\delta} (\leftidx{^\delta}{\hat{u}}))^{\frac{\beta}{\alpha}}}{(|u|^{\alpha} \sigma_\delta + C_{\delta} (\leftidx{^\delta}{\hat{u}})) - (|\leftidx{^\delta}{\hat{u}}_t|^{\alpha}\sigma_\delta + C_{\delta} (\leftidx{^\delta}{\hat{u}}))} \cdot \frac{(|u|^{\alpha} \sigma_\delta + C_{\delta} (\leftidx{^\delta}{\hat{u}})) - (|\leftidx{^\delta}{\hat{u}}_t|^{\alpha}\sigma_\delta + C_{\delta} (\leftidx{^\delta}{\hat{u}}))}{\delta}.
	\end{align*}
	The second term converges to $(|u|^{\alpha}-|\tilde{u}|^{\alpha}) |\sigma_t^{11} e^{r(T-t)}|^{\alpha}= |u^\intercal\sigma_te^{r(T-t)}|^{\alpha}-|\hat{u}_t^\intercal\sigma_te^{r(T-t)}|^{\alpha}$ for $\delta \to 0$ due to the convergence of $\leftidx{^\delta}{\hat{u}}$ to $\tilde{u}$ and $\sigma$ being continuous if $d=1$ or a constant times the identity if $d \geq 2$. On the first term, we apply the mean value theorem of differentiation and thus get a $\xi_{\delta}$ between $|u|^{\alpha} \sigma_\delta + C_{\delta} (\leftidx{^\delta}{\hat{u}})$ and $|\leftidx{^\delta}{\hat{u}}_t|^{\alpha}\sigma_\delta + C_{\delta} (\leftidx{^\delta}{\hat{u}})$ such that the first term is equal to $\frac{\beta}{\alpha} \xi_{\delta}^{\frac{\beta-\alpha}{\alpha}}$. Since both boundary values for $\xi_{\delta}$ converge to $C_0 (\tilde{u})$ due to the convergence of $\leftidx{^\delta}{\hat{u}}$ to $\tilde{u}$, also $\xi_{\delta}$ converges to $C_0 (\tilde{u})$. Overall, we get the convergence of the inner part of the left side from the discrete to the continuous case for a fixed $u$ after noticing that the argumentation is similar for the other two terms. Hence, for a fixed $u$, we have \eqref{eq: hjb convergence 1} with a ``$\leq$'' for $\tilde{u}$. Now, we take the supremum over all $u \in \mathcal{V}$ which leads to the continuous HJB equation \eqref{eq: hjb convergence 1} with a ``$\leq 0$'' instead of a ``$=0$''. However, choosing $u = \tilde{u}_t$ gives immediately a ``$0$'' in the equation by the previously shown convergence results for the discrete optimal strategies. Hence, $\tilde{u}$ solves the continuous HJB equation \eqref{hjb equation}. Moreover, the constraints hold by Remark \ref{rem: hjb constraints}. Thus, $\tilde{u}$ is an optimal equilibrium function by Theorem \ref{verification}, and the theorem follows.
\end{proof}

\section{Formulas for the optimal control function in special cases} \label{sec: Black-Scholes}

In the following, we consider an exponential target function $\T(x) = c^{-1} (1- \exp ({-\gamma x}))$ with $c,\gamma>0$ in the first two subsections and the identity for $\T$ in the third one. Since the exponential target function is strictly concave, the optimal control function in discrete time is unique (see Theorem \ref{deterministic control function}). Moreover, we consider a multidimensional Black-Scholes model in the first, a multidimensional pure upward jump L{\'e}vy model in the second, and a one-dimensional general L{\'e}vy model in the third subsection. We provide proofs in \ref{proofs}. Note that under Assumption \ref{ass: bounded strategies}, the optimal strategies need to be truncated at $\pm M$. However, for the ease of exposition, we will assume in the sequel that under Assumption \ref{ass: bounded strategies} $M$ is large enough.

\subsection{Multidimensional Black-Scholes model with an exponential target function}

This section considers the Black-Scholes model, i.e., $\alpha=2$. Since Brownian Motions are symmetric, we are in the case of Assumption \ref{ass: symmetric multidemensional}.

\subsubsection{An explicit recursion}

We give a formula for calculating the optimal control function in discrete and continuous time. Recall that we restricted in discrete time the optimization to piecewise constant control functions.

\begin{proposition} \label{formula exp target disc}
	The optimal solution in discrete time is given by going backward in time from $n=T-1$ to $n=0$ solving the following equation for $\hat{u}_{n+1}^k$ for all $k \in \{ 1,\ldots, d\}$:
	\begin{align}
		0=&- c^{-1} \exp ({-\gamma \Int_n^{T} {\hat{u}}_s^\intercal (\mu_s-r)e^{r(T-s)} \diff s + \frac{\gamma^2}{2} \norm{{\hat{u}}_s^\intercal \sigma_s e^{r(T-s)}}^2_{\L^2((n,T],\diff s; H_s)}}) \notag \\
		&\hspace{0.5cm}\cdot ( -\gamma \Int_n^{n+1} (\mu_s^k-r)e^{r(T-s)} \diff s +\gamma^2 \Int_n^{n+1} \langle {\hat{u}}_s^\intercal \sigma_s e^{r(T-s)}, \sigma_{s}^{k \cdot} e^{r(T-s)}\rangle_{R_s} \diff s) \notag \\
		&- \lambda_n F' ( -\Int_n^T m_s ({\hat{u}}) \diff s + \varrho^{W_1^1} \norm{{\hat{u}}_s^\intercal \sigma_s e^{r(T-s)}}_{\L^{2}((n,T], \diff s; H_s)} )\notag \\
		&\hspace{0.5cm}\cdot ( - \1_{\{\text{\scriptsize$\rho$ cash-invariant}\}}\Int_n^{n+1} (\mu^k_s - r) e^{r(T-s)}\diff s + \varrho^{W_1^1}\textfrac{ \Int_n^{n+1} \langle {\hat{u}}_s^\intercal \sigma_{s} e^{r(T-s)},\sigma_{s}^{k \cdot} e^{r(T-s)} \rangle_{R_s} \diff s}{\norm{{\hat{u}}_s^\intercal \sigma_s e^{r(T-s)}}_{\L^{2}((n,T], \diff s; H_s)}} ), \label{eq: Prop. 6.1}
	\end{align}
	with $m$ as in Definition \ref{def m w}.
\end{proposition}

\begin{remark}
	Since $J$ is not differentiable in $0$, it is not ensured that \eqref{eq: Prop. 6.1} admits a solution. Whenever there is no solution, the optimal control function is then given by $0$ for this asset at this time point. Note that this remark also applies to all following results, where the solution is not given analytically.
\end{remark}

\begin{remark} \label{formula exp target cont}
	The continuous-time analog of \eqref{eq: Prop. 6.1} is given by going backward in time solving the following equation for $\hat{u}_t^k$ for all $k \in \{ 1,\ldots, d\}$:
	\begin{align*}
		0=&-c^{-1} \exp ({-\gamma \Int_t^{T} {\hat{u}}_s^\intercal (\mu_s-r)e^{r(T-s)} \diff s + \frac{\gamma^2}{2} \norm{{\hat{u}}_s^\intercal \sigma_s e^{r(T-s)}}^2_{\L^2((t,T],\diff s; H_s)}}) \\
		&\hspace{0.5cm}\cdot ( -\gamma (\mu_t^k-r) e^{r(T-t)} +\gamma^2 \langle {\hat{u}}_t^\intercal \sigma_t e^{r(T-t)}, \sigma_{t}^{k \cdot} e^{r(T-t)}\rangle_{R_t} ) \\
		&- \lambda_t F' ( -\Int_t^T m_s ({\hat{u}}) \diff s + \varrho^{W_1^1} \norm{{\hat{u}}_s^\intercal \sigma_s e^{r(T-s)}}_{\L^{2}((t,T], \diff s; H_s)} )\\
		&\hspace{0.5cm}\cdot ( - \1_{\{\text{\scriptsize$\rho$ cash-invariant}\}} (\mu^k_t - r) e^{r(T-t)} + \varrho^{W_1^1}\textfrac{ \langle {\hat{u}}_t^\intercal \sigma_t e^{r(T-t)}, \sigma_{t}^{k \cdot} e^{r(T-t)}\rangle_{R_t}}{\norm{{\hat{u}}_s^\intercal \sigma_s e^{r(T-s)}}_{\L^{2}((t,T], \diff s; H_s)}} ),
	\end{align*}
	with $m$ as in Definition \ref{def m w}. This can be seen by replacing $n$ and $n+1$ by $t$ and $t+\delta$, dividing by $\delta$, and taking the limit $\delta \to 0$.
\end{remark}

\subsubsection{No risk penalization}

The following paragraph considers the special case where we do not penalize risk, i.e., $F \equiv 0$. In this case, we can calculate the optimal control function analytically:

\begin{proposition} \label{optimal control no risk, disc}
	In discrete time, the optimal control function can be calculated by solving for each point $n \in \{0,\ldots, T-1\}$ the following $d$-dimensional system, where the $k$'th element is given by:
	\begin{align} \label{eq: Prop. 6.3}
		\hat{u}^k_{n+1} = \textfrac{\Int_n^{n+1} (\mu_s^k-r) e^{r(T-s)} \diff s - \gamma \Sum_{i=1,i\neq k}^d \hat{u}_{n+1}^i \Int_n^{n+1} \sigma_{s}^{i\cdot} R_s \sigma_{s}^{\cdot k} e^{2r(T-s)} \diff s}{\gamma \Int_n^{n+1} \sigma_{s}^{k\cdot} R_s \sigma_{s}^{\cdot k} e^{2r(T-s)} \diff s}.
	\end{align}
\end{proposition}

\begin{remark} \label{optimal control no risk, cont}
	The continuous-time analog of \eqref{eq: Prop. 6.3} for $\hat{u}$ can be calculated by solving for each point $t \in (0, T]$ the following $d$-dimensional system, where the $k$'th element is given by $\hat{u}^k_t = ({\mu_t^k-r - \gamma \Sum_{i=1,i\neq k}^d \hat{u}_t^i \sigma_{t}^{i\cdot} R_t \sigma_{t}^{\cdot k} e^{2r(T-t)} })({\gamma \sigma_{t}^{k\cdot} R_t \sigma_{t}^{\cdot k} e^{2r(T-t)} })^{-1}$. Again, this can be seen by replacing $n$ and $n+1$ by $t$ and $t+\delta$, dividing by $\delta$ and taking the limit $\delta \to 0$.
\end{remark}

\subsection{An explicit recursion for a pure upward L{\'e}vy model with an exponential target function}

In this subsection, we take for $L$ a pure upward jump L{\'e}vy process, i.e., $\tilde{\sigma} (v) > 0$ only if $v^i \geq 0$ for all $i \in \{1,\ldots,d\}$. We restrict to such L{\'e}vy processes for the exponential target function since the Laplace function $\E[e^{-tL}]$ only exists for such processes if $t \geq 0$.
\begin{proposition} \label{formula levy exp target disc}
	The optimal solution in discrete time is given by going backward in time from $n=T-1$ to $n=0$ solving the following equation for $\hat{u}_{n+1}^k$ for all $k \in \{ 1,\ldots, d\}$:
	\begin{align} \label{eq: Prop. 6.5}
		0=&-c^{-1} \exp \big({-\gamma \Int_n^{T} {\hat{u}}_s^\intercal (\mu_s-r)e^{r(T-s)} \diff s + \frac{c_{\alpha} \gamma^{\alpha}}{- \cos (\frac{\pi \alpha}{2})} \norm{{\hat{u}}_s^\intercal \sigma_s v e^{r(T-s)}}^{\alpha}_{\L^{\alpha}((n,T] \times \S^d,\diff s \times \tilde{\sigma} (\diff v))}} \big) \\
		&\hspace{0.5cm}\cdot ( -\gamma \Int_n^{n+1} (\mu_s^k-r)e^{r(T-s)} \diff s + \frac{\alpha c_{\alpha} \gamma^{\alpha}}{- \cos (\frac{\pi \alpha}{2})} \Int_n^{n+1} \Int_{\S^d} | {\hat{u}}_s^\intercal \sigma_s v e^{r(T-s)} |^{\alpha-1} \sigma_s^{k \cdot} v e^{r(T-s)} \tilde{\sigma} (\diff v)\diff s) \notag \\
		&- \lambda_n F' ( -\Int_n^T m_s ({\hat{u}}) \diff s + \varrho^{\tilde{L}_1} \norm{{\hat{u}}_s^\intercal \sigma_s v e^{r(T-s)}}_{\L^{\alpha}((n,T] \times \S^d,\diff s \times \tilde{\sigma} (\diff v))} ) \notag \\
		&\hspace{0.5cm}\cdot \big( - \1_{\{\text{\scriptsize$\rho$ cash-invariant}\}}\Int_n^{n+1} (\mu^k_s - r) e^{r(T-s)}\diff s \notag \\
		&\hspace{1cm}+ \varrho^{\tilde{L}_1}\textfrac{ \Int_n^{n+1} \Int_{\S^d} | {\hat{u}}_s^\intercal \sigma_s v e^{r(T-s)} |^{\alpha-1} \sigma_s^{k \cdot} v e^{r(T-s)} \tilde{\sigma} (\diff v) \diff s}{\norm{{\hat{u}}_s^\intercal \sigma_s v e^{r(T-s)}}^{1-\frac{1}{\alpha}}_{\L^{\alpha}((n,T] \times \S^d,\diff s \times \tilde{\sigma} (\diff v))}} \big), \notag
	\end{align}
	with $m$ as in Definition \ref{def m w} and $c_{\alpha}$ as in \eqref{eq: characteristic function}.
\end{proposition}

\begin{remark} \label{formula levy exp target cont}
	The continuous-time analog of \eqref{eq: Prop. 6.5} is given by going backward in time solving the following equation for $\hat{u}_t^k$ for all $k \in \{ 1,\ldots, d\}$:
	\begin{align*}
		0=&-c^{-1} \exp \big({-\gamma \Int_t^{T} {\hat{u}}_s^\intercal (\mu_s-r)e^{r(T-s)} \diff s + \frac{c_{\alpha} \gamma^{\alpha}}{- \cos (\frac{\pi \alpha}{2})} \norm{{\hat{u}}_s^\intercal \sigma_s v e^{r(T-s)}}^{\alpha}_{\L^{\alpha}((t,T] \times \S^d,\diff s \times \tilde{\sigma} (\diff v))}} \big) \\
		&\hspace{0.5cm}\cdot ( -\gamma (\mu_t^k-r) e^{r(T-t)} + \textstyle\textfrac{\alpha c_{\alpha} \gamma^{\alpha}}{- \cos (\frac{\pi \alpha}{2})} \Int_{\S^d} | {\hat{u}}_t^\intercal \sigma_t v e^{r(T-t)} |^{\alpha-1} \sigma_t^{k \cdot} v e^{r(T-t)} \tilde{\sigma} (\diff v) ) \\
		&- \lambda_t F' ( -\Int_t^T m_s ({\hat{u}}) \diff s + \varrho^{\tilde{L}_1} \norm{{\hat{u}}_s^\intercal \sigma_s v e^{r(T-s)}}_{\L^{\alpha}((t,T] \times \S^d,\diff s \times \tilde{\sigma} (\diff v))} )\\
		&\hspace{0.5cm}\cdot \big( - \1_{\{\text{\scriptsize$\rho$ cash-invariant}\}} (\mu^k_t - r) e^{r(T-t)} + \varrho^{\tilde{L}_1}\textfrac{ \Int_{\S^d} | {\hat{u}}_t^\intercal \sigma_t v e^{r(T-t)} |^{\alpha-1} \sigma_t^{k \cdot} v e^{r(T-t)} \tilde{\sigma} (\diff v)}{\norm{{\hat{u}}_s^\intercal \sigma_s v e^{r(T-s)}}^{1-\frac{1}{\alpha}}_{\L^{\alpha}((t,T] \times \S^d,\diff s \times \tilde{\sigma} (\diff v))}} \big),
	\end{align*}
	with $m$ as in Definition \ref{def m w} and $c_{\alpha}$ as in \eqref{eq: characteristic function}. Again, this can be seen by replacing $n$ and $n+1$ by $t$ and $t+\delta$, dividing by $\delta$ and taking the limit $\delta \to 0$.
\end{remark}

\subsection{Mean-fractional central moment in a one-dimensional L{\'e}vy model} \label{MV formula}

In this subsection, we use the identity as the target function, i.e., $\T(x)=x$, and fractional central moments as the risk measure, i.e., $\rho(L) = \sqrt[\eta]{\E |L-\E[L]|^\eta} = \sqrt[\eta]{\E |L|^\eta}$ (since $\E[L]=0$) and $F(x)=\max \{0,x\}^\eta$ with $1<\eta<\alpha$ if $1<\alpha<2$ and $1<\eta\leq \alpha$ if $\alpha=2$, which are shift-invariant. This includes the variance if $\eta=\alpha=2$. Moreover, we compare the mean-variance optimal solution from our setting (i.e., considering gains and losses) to the mean-variance optimal solution which optimizes terminal wealth (see, e.g., \cite{Basak}, \cite{Bjork}, \cite{Forsyth}, \cite{lim2002mean}, \cite{Lindensjo}, \cite{pedersen2016optimal}, \cite{Wang}). Following Bj{\"o}rk and Murgoci \cite[p.57]{Bjork}, we set $d=1$, $\mu_t \equiv \mu$, $\sigma_{t} \equiv \sigma$, and $\lambda_t \equiv \lambda$. In this subsection, we give the optimal solution in closed formulas. We get $F(\varrho^{\hat{L}_1})=F(\varrho^{W^1_1}) = 2^{\eta/2} \frac{\Gamma (\frac{\eta+1}{2})}{\sqrt{\pi}}$ if $\alpha=2$ (see Winkelbauer \cite{winkelbauer2012moments}), and $F(\varrho^{\hat{L}_1}) = \frac{2}{\pi} \Gamma(\eta+1) \sin (\frac{\pi \eta}{2} ) \Int_0^{\infty} \frac{1-e^{-c_{\alpha} |t|^{\alpha}}}{t^{\eta+1}} \diff t$ if $\alpha<2$ (see Lin and Hu \cite{lin2018absolute}), where $\Gamma$ denotes the Gamma function and $c_{\alpha}$ is as in \eqref{eq: characteristic function}. In particular, it holds that $F(\varrho^{\hat{L}_1})=F(\varrho^{W^1_1}) = 1$ if $\eta=\alpha=2$.

\begin{proposition} \label{terminal wealth disc}
	The unique discrete-time optimal solution $\hat{u}$ is given by going backward in time from $n=T-1$ to $n=0$ solving the following equation for $\hat{u}_{n+1}$
	\begin{align} \label{eq: Prop. 6.7}
		0 =& (\mu-r) \Int_{n}^{n+1} e^{r(T-s)} \diff s \notag \\
		&- \sign(\hat{u}_{n+1}) \eta \lambda F(\varrho^{\hat{L}_1}) \sigma^\eta ( \Int_n^T |\hat{u}_s|^\alpha e^{\alpha r (T-s)} \diff s )^{\eta/\alpha-1} |\hat{u}_{n+1}|^{\alpha-1} \Int_n^{n+1} e^{\alpha r (T-s)} \diff s.
	\end{align}
\end{proposition} 

\begin{remark} \label{terminal wealth cont}
	The continuous-time analog of \eqref{eq: Prop. 6.7} for $\hat{u}$ is given by $\hat{u}_{t} = (\textfrac{\mu-r}{\eta\lambda F(\varrho^{\hat{L}_1}) \sigma^\eta})^{\frac{1}{\alpha-1}} e^{-r(T-t)}$ $( \Int_t^T |\hat{u}_s|^\alpha e^{\alpha r (T-s)} \diff s)^{\frac{\alpha-\eta}{\alpha(\alpha-1)}}$ for all $t \in (0,T]$. Also this can be seen by replacing $n$ and $n+1$ by $t$ and $t+\delta$, dividing by $\delta$ and taking the limit $\delta \to 0$.
\end{remark} 
From Remark \ref{terminal wealth cont}, we conclude that the continuous-time optimal solution for mean-variance optimization in the Black-Scholes model, i.e., $\eta = \alpha=2$, is given by $\hat{u}_t = \frac{\mu-r}{2\lambda\sigma^2} e^{-r(T-t)}$ which is the same as in Bj{\"o}rk and Murgoci \cite[p.57]{Bjork} (with a different notation). Hence, optimizing gains and losses gives identical solutions as optimizing terminal wealth for mean-variance in the Black-Scholes model.

\section{Numerical results} \label{sec: numeric}

In the following section, we calculate numerical values of the optimal control function, applying the results from the previous sections. For the solutions shown, we use VaR (see Example \ref{ex: risk measures}(\ref{ex: VaR})) as the main risk measure and compare it to mean-variance optimization. In the case of VaR, it holds that $\varrho^{W^1_1} = -z_p = z_{1-p}$ for $\alpha=2$, where $z_{\beta}$ denotes the $\beta$-quantile of a standard normal distribution. If $\alpha<2$, we use a simulation based on the algorithm of Misiorek and Weron \cite{misiorek2012heavy} to calculate VaR.

So, we start to consider a two-dimensional geometric Brownian Motion with the risk-free interest rate $r=0.02$, expected return $\mu := \mu_t \equiv \begin{psmallmatrix} 0.08 \\ 0.06 \end{psmallmatrix}$, volatility matrix $\sigma := \sigma_t \equiv \begin{psmallmatrix} 0.20 & 0.10 \\ 0.10 & 0.15 \end{psmallmatrix}$, and correlation matrix $R := R_t \equiv \begin{psmallmatrix} 1 & 0.5 \\ 0.5 & 1 \end{psmallmatrix}$. We set $T=10$ with a step size of $\delta = 0.01$, i.e., there are $1000$ intermediate points. Moreover, we take $\lambda := \lambda_t \equiv 0.25$ for the risk aversion parameter and $p=0.01=1\%$ for the risk level. In the one-dimensional optimizations, we just take the first asset from this model.

In Figure \ref{fig: expTarget}, we compare the time-consistent optimal investments under Value-at-Risk (with different target and/or $F$ functions) to the optimal investment under mean-variance. In the two-dimensional Black-Scholes model in the left plot, we illustrate this for the cases that $F \equiv 0$, $F(x)=x$ or $F(x)=\max \{x,0\}^2$ when $\T$ is the exponential target function with parameters $c = \gamma = 1$, and $F(x)=x$ or $F(x)=\max \{x,0\}^2$ when $\T$ is the identity target function $\T(x)=x$. In the one-dimensional L{\'e}vy model in the middle plot, we replace the Brownian Motion with a pure upward jump $\alpha$-stable L{\'e}vy process with $c=0.05$. We give the result for the exponential target function and the Value at Risk with $F(x)=\max \{x,0\}^2$. In the one-dimensional L{\'e}vy model of the right plot, we use a pure upward jump, a symmetric, and a pure downward jump L{\'e}vy model with $c=0.5$, the identity target function and the Value at Risk with $F(x)=\max \{x,0\}^2$. We display the optimal solutions as a fraction of the mean-variance optimal solution (with $\alpha=2$). The fraction for the second asset in the 2-dimensional Black-Scholes model is not shown since the curves differ only slightly.
\begin{figure}[!htb]
	\centering
	\begin{minipage}{0.3\textwidth}
		\includegraphics[trim= 0mm 6mm 10mm 19mm, clip,width=\textwidth]{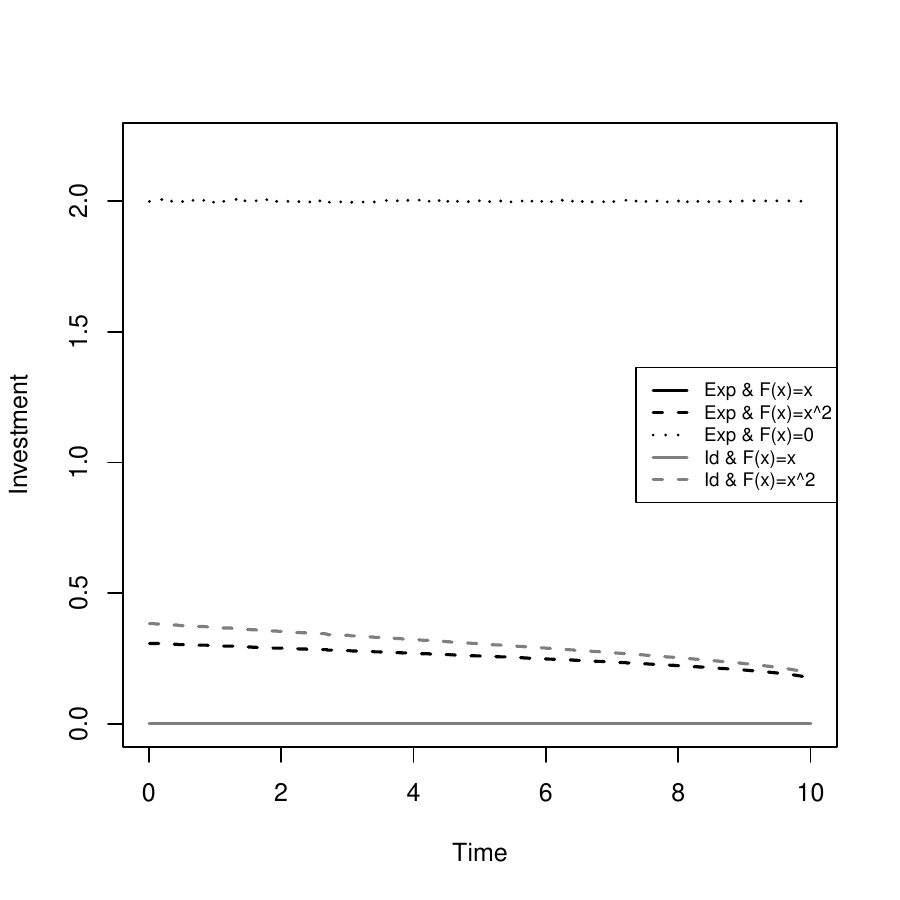}
	\end{minipage}
	\quad
	\begin{minipage}{0.3\textwidth}
		\includegraphics[trim= 0mm 6mm 10mm 19mm, clip,width=\textwidth]{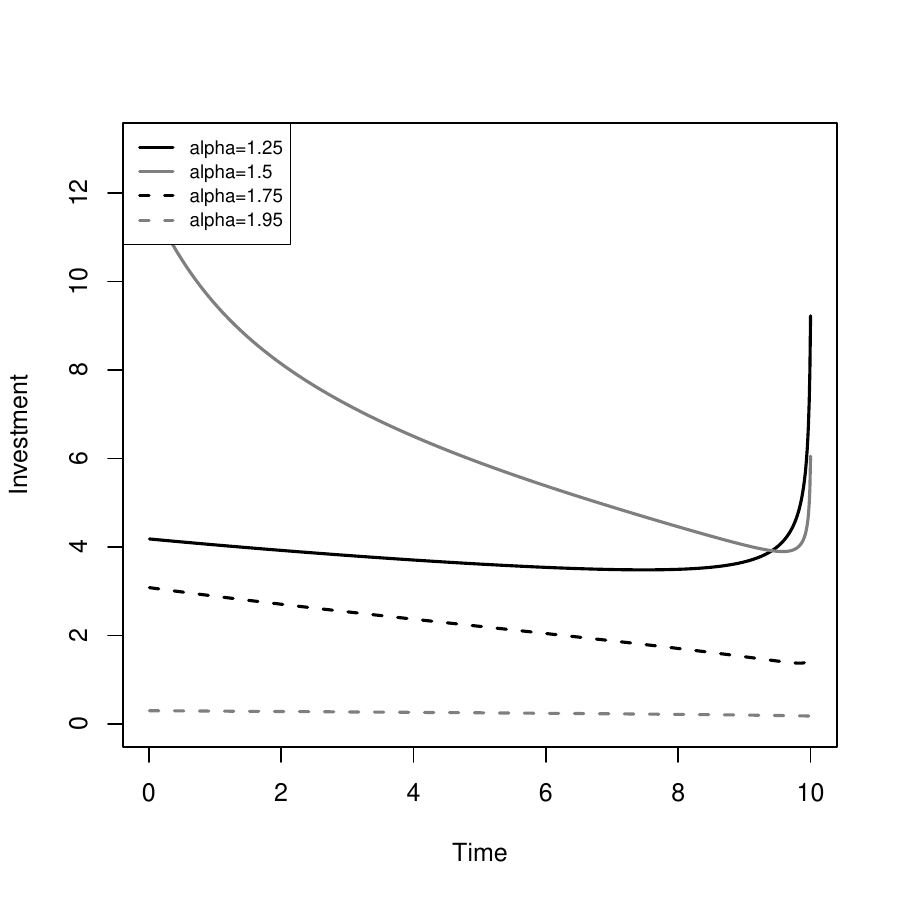}
	\end{minipage}
	\quad
	\begin{minipage}{0.3\textwidth}
		\includegraphics[trim= 0mm 6mm 10mm 19mm, clip,width=\textwidth]{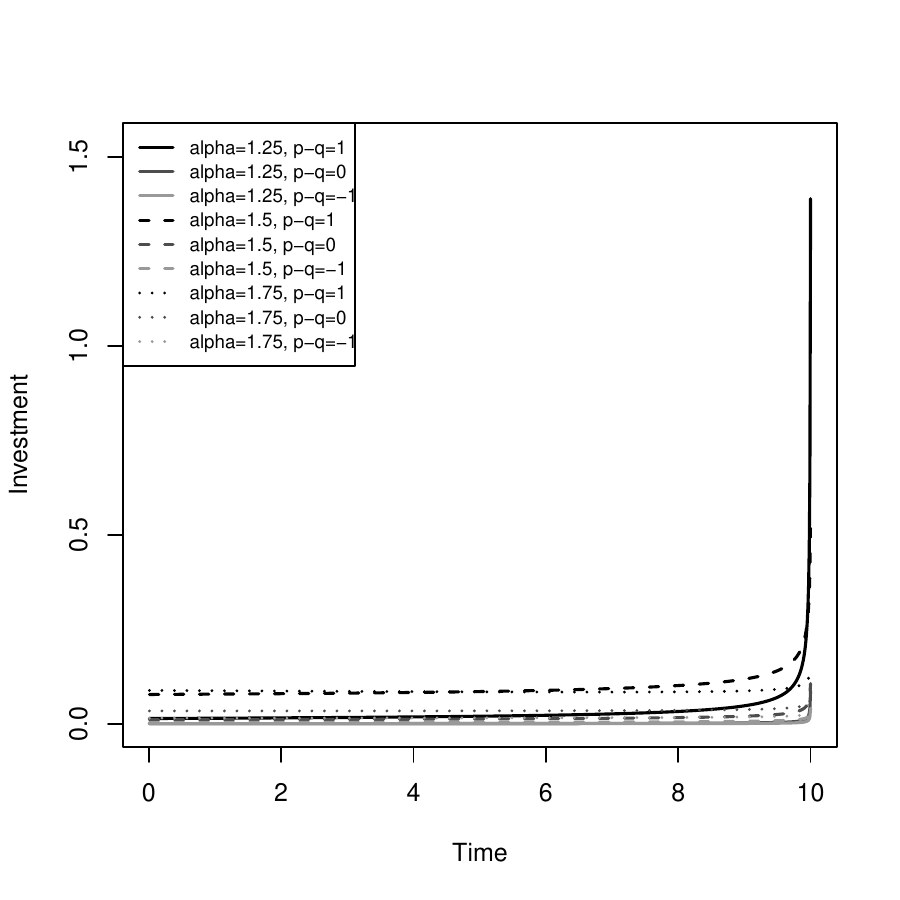}
	\end{minipage}
	\caption{Comparison of optimal control functions using the exponential or identity target function with different $F$ functions to the mean-variance optimization for the 2-dimensional Black-Scholes model (left) and the one-dimensional L{\'e}vy model (middle and right).}
	\label{fig: expTarget}
\end{figure}	

The left plot of Figure \ref{fig: expTarget} shows that the investment with VaR optimization non-increases over time compared to the mean-variance optimal investment in all cases. In particular, under a shorter time horizon, the mean-variance investor tends to take riskier positions than the optimal VaR investor. Moreover, we notice that the optimal investment into the risky asset is zero if $F(x)=x$. Furthermore, we see that the exponential utility investment without risk penalization has an investment structure similar to mean-variance optimization. This stems from the similar structure of these problems, which can be seen by considering the Taylor series of the exponential utility function.
Hence, we conclude that VaR optimization leads to more conservative investments for shorter maturities and comparably larger investments in the risky asset than mean-variance. On top of that, we notice that using the exponential target function compared to the identity leads to a more conservative investment. This is reasonable since when taking $\gamma=c$ and their limit to $0$, the exponential target function reduces to the identity. If these values are positive, the investor is risk-averse, whereas the investor is risk-neutral in the identity case.\\
The other plots of Figure \ref{fig: expTarget} show that the style of the investment strategy is different when using L{\'e}vy processes instead of the Brownian Motion. Contrary to Brownian Motions, L{\'e}vy processes show the feature that compared to the Black-Scholes mean-variance optimization, the riskiness increases for very short maturities compared to longer maturities. Furthermore, the middle plot displays no monotonicity in the investment behavior concerning $\alpha$. This stems from the observation that Value at Risk of $\hat{L}_1$ is actually not monotone in $\alpha$. In fact, Value at Risk value describes a $U$-shape for this parametrization. Moreover, we observe that for smaller $\alpha$'s, the investment strategy gets comparably riskier with shorter maturities. This effect decreases and inverses for higher $\alpha$'s as observed for $\alpha=2$ in the left plot. From the right plot, we can observe that a model with pure upward jumps tends to entail riskier strategies than one with pure downward jumps. This stems from the fact that the Value at Risk of $\hat{L}_1$ is smaller for pure upward jumps, and hence, the penalization of the risk is smaller. In addition, when comparing the middle and the right plots, we observe that a higher $c$ in the L{\'e}vy model induces less risky strategies.

In Figure \ref{fig: Compare}, we compare the Nash equilibrium strategy of a sophisticated investor to a pre-committed investor who optimizes the problem only once in the beginning and sticks to this strategy until maturity. We consider only the first risky asset from the parametrization at the start of the section and set $T=1$. We analyze mean-variance optimization in the left plot and VaR optimization using the exponential target function with $\gamma=c=0.25$ (and $F(x)=\max\{0,x\}^2$) in the right plot. Both plots give the investments as a fraction of the optimal pre-commitment investment. In the mean-variance case, the pre-commitment strategy is taken from Zhou and Li \cite{zhou2000continuous}. In the VaR case, we calculate the pre-commitment solution for the exponential target function in the same way as Basak and Shapiro \cite{Shapiro} did for power utility. Note that the pre-commitment strategies are not deterministic. Hence, we take the expected value of the optimal strategy instead. Since this cannot be done explicitly for VaR, we use, for this case, the average of a Monte-Carlo-Simulation with $N=100000$.
\begin{figure}[!htb]
	\centering
	\begin{minipage}{0.3\textwidth}
		\includegraphics[trim= 0mm 6mm 10mm 19mm, clip,width=\textwidth]{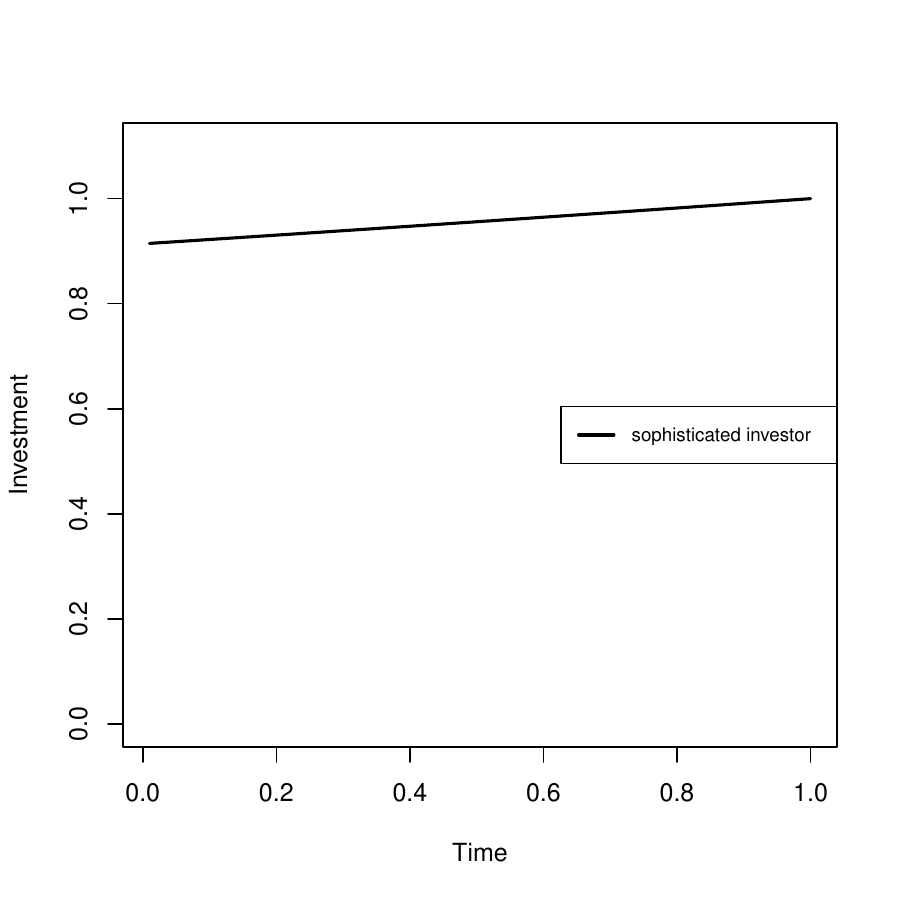}
	\end{minipage}
	\qquad
	\begin{minipage}{0.3\textwidth}
		\includegraphics[trim= 0mm 6mm 10mm 19mm, clip,width=\textwidth]{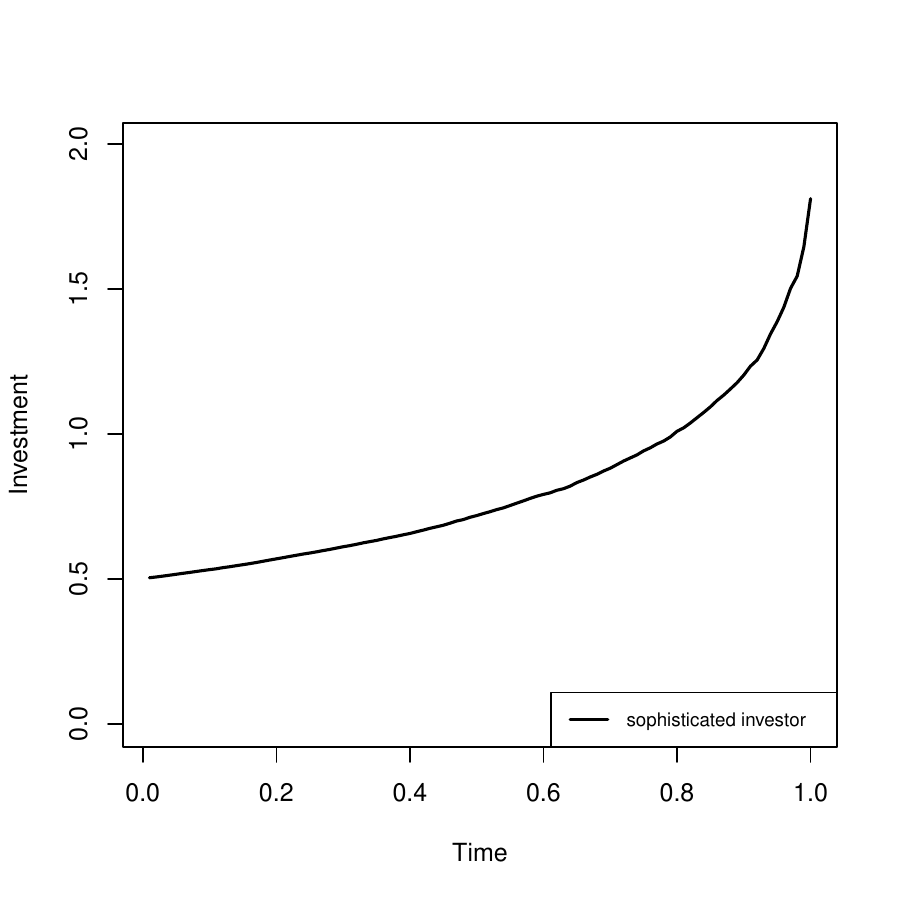}
	\end{minipage}
	\caption{Comparison of Nash equilibrium to pre-commitment optimal strategies for mean-variance (left) and mean-VaR$^2$ (right) optimization.}
	\label{fig: Compare}
\end{figure}	

Figure \ref{fig: Compare} illustrates that time-consistent investment strategies tend to be more conservative than pre-commitment strategies over longer time horizons. For shorter maturities, however, this pattern holds only under the mean-variance criterion, where the difference between the two approaches remains relatively small. In contrast, under the mean-VaR$^2$ framework, the time-consistent sophisticated strategy is noticeably riskier for short-term horizons than the pre-commitment strategy.

\section{Conclusion}

In this paper, we derived an HJB equation for a multidimensional asset allocation problem focusing on gains \& losses with rather general risk measures in an $\alpha$-stable L{\'e}vy model. We showed that under appropriate conditions an optimal strategy in the sense of a Nash subgame equilibrium exists and is deterministic. We also proved a convergence result from discrete to continuous time of the optimal strategies. In continuous time, we needed additional assumptions on the model. A numerical analysis shows that Value at Risk leads to less risky investments compared to mean-variance for short maturities and that a time-consistent investor invests less risky than a pre-commitment investor. %(Future research directions include possibly generalizing this approach to other L{\'e}vy processes and risky assets.)

\section*{Funding}

This research did not receive any specific grant from funding agencies in the public, commercial, or not-for-profit sectors.

\setcounter{section}{0}
\renewcommand{\thesection}{\Alph{section}}
\appendix
\section{L{\'e}vy processes and their stochastic integral} \label{app: Levy}

In this section, we describe more properties of L{\'e}vy processes as a continuation of Section \ref{sec: levy} and construct their stochastic integral. The description is based on \c{C}{\i}nlar \cite[pp.314-315,322-323]{Cinlar}, and the stochastic integral construction on Protter \cite[pp.51-65]{protter2005stochastic}.

\subsection{L{\'e}vy processes} \label{app-section: Levy Erklärung}

Let $L$ be a L{\'e}vy process, i.e., it is adapted to $\mathcal{F}$, has c{\`a}dl{\`a}g paths for almost all $\omega$ with $L_0 (\omega)=0$, and has independent increments with $L_{t+u}-L_t \overset{d}{=} L_u$ for all $t,u \in \R_{\geq 0}$. The last property is equivalent to having independent and stationary increments. Let $L$ be $\alpha$-stable with $\alpha \in (0,2]\backslash\{1\}$. Then, the characteristic function $\varphi$ of $L$ is given by
\begin{align} \label{eq: characteristic function}
	\varphi_{L_t} (r) = \E [\exp ({i\scprod{r,L_t}})] = \exp(-t c_{\alpha} \Int_{\S^d} |\scprod{r,v}|^{\alpha} [1-i \tan (\frac{\pi \alpha}{2}) \sign (\scprod{r,v})] \tilde{\sigma} (\diff v))
\end{align}
where $r \in \R^d$, $c_{\alpha} := c \frac{\Gamma(1-\alpha)}{\alpha} \cos ( \frac{\pi \alpha}{2}  )$, and $\Gamma(\cdot)$ denotes the Gamma function. If $\alpha \in (1,2]$, it holds that $\E[L_t]= 0$ and that $L_t$ is a martingale if $L_t$ is additionally symmetric.
In the one-dimensional case, the characteristic function reduces to
\begin{align} \label{eq: characteristic function 1dim}
	\varphi_{L_t} (r) = \E [\exp({ir L_t})] = \exp(-t c_{\alpha} |r|^{\alpha} [1-i (p-q)\tan \textstyle (\frac{\pi \alpha}{2}) \sign (r) ] ),
\end{align}
where $p$ is the (conditional) probability of an upward jump and $q$ is the (conditional) probability of a downward jump. In particular, $p+q=1$.

\subsection{Stochastic integral for L{\'e}vy processes} \label{def: stochastic integral}

Every L{\'e}vy process is a semimartingale; hence, we use the definition of stochastic integrals for semimartingales here. \emph{A sequence of processes $(X^n)_{n \geq 1}$ converges in $ucp$ to the process $X$, if $\sup_{0 \leq s \leq t} | X_s^n-X_s | \xrightarrow{\PP} 0$ for all $t>0$} (\cite[Cf. p.57]{protter2005stochastic}). We define $\cL$ as the set of all adapted processes with c{\`a}gl{\`a}d (left-continuous with right limits) paths and $\cD$ as the set of all adapted processes with c{\`a}dl{\`a}g (right-continuous with left limits) paths. Moreover, we write $\cL_{ucp}$ for $\cL$ endowed with the $ucp$-topology and $\cL^{det}_{sup}$ for all functions from $\cL_{ucp}$ which are deterministic. For $\cD$, we use analogous notations. \emph{We call the process $H$ simple predictable if we can write $H$ as $H_t = H_0 \1_{\{0\}} (t) + \Sum_{j=1}^n H_j \1_{(T_{j-1},T_{j}]} (t)$ where $0=T_0 \leq \ldots \T_n < \infty$ is a sequence of stopping times and for all $j \in \{0,\ldots,n\}$ holds: $H_j \in \mathcal{F}_{T_j}$} (\cite[Cf. p.51]{protter2005stochastic}). Let $\bS$ be the set of all simple predictable processes.

Let $L$ be a L{\'e}vy process and $H \in \bS$. Then the stochastic integral of $H$ with respect to $L$ is defined as $\Int_0^t H_s \diff L_s = H_0 L_0 + \Sum_{j=1}^n H_j (L_{T_j}-L_{T_{j-1}})$ with $H$ constructed as above and $T_j \leq t$ for all $j$. This is a mapping from $\bS$ to $\cD$. We define the general stochastic integral as the extension of this from $\cL_{ucp}$ to $\cD_{ucp}$. If we want to exclude the $0$, we write $\Int_{0+}^t H_s \diff L_s$. For this paper, an important property of the stochastic integral over $\cL$ is that it holds for $Y \in \cL_{sup}^{det}$ and $t>0$ that
\begin{align} \label{eq: formula integral construction}
	\Int_{0+}^t Y_s^{\pi_n} \diff L_s \xrightarrow{\PP} \Int_{0+}^t Y_s \diff L_s
\end{align}
where $\pi_n: 0=t_0 \leq \ldots \leq t_n = t$ is the equidistant partition of $[0,T]$ in $n$ intervals. Note: $\Int_{0+}^t Y_s^{\pi_n} \diff L_s = \Sum_{j=1}^n Y_{t_j} (L_{t_j}-L_{t_{j-1}})$. 

One can generalize this definition to predictable processes such that the stochastic integral is linear in the integrand, see Applebaum \cite[p.237]{applebaum2009levy} and the references therein. 

\section{Lemmas and Proofs} \label{proofs}

\renewcommand{\thetheorem}{\Alph{section}.\arabic{theorem}}

\begin{myproof}[Proof of Lemma \ref{alpha additive}] 
	First, we define $Z_{t-s} = \scprod{a,L_t - L_s}$. Then we get from (\ref{eq: characteristic function}) since $r \in \R$:
	\begin{align*}
		\varphi_{Z_{t-s}} (r) &= \E [ \exp ({irZ_{t-s}}) ] = \E [ \exp ({ir\scprod{a,L_t - L_s}}) ] = \E [ \exp ({i\scprod{ra,L_t - L_s}}) ] \\
		&= \exp \{ -(t-s) c_{\alpha} \Int_{\S^d} |\scprod{ra,v} |^{\alpha} [1-i \tan (\frac{\pi \alpha}{2}) \sign (\scprod{ra,v})] \tilde{\sigma} (\diff v) \} \\
		&= \exp \{ -(t-s) c_{\alpha} |r|^{\alpha} [ \Int_{\S^d} | \scprod{a,v} |^{\alpha} \tilde{\sigma} (\diff v) - i \tan (\frac{\pi \alpha}{2}) \Int_{\S^d} | \scprod{a,v} |^{\alpha} \sign (r) \sign (\scprod{a,v}) \tilde{\sigma} (\diff v)] \}  \\
		&= \exp \big\{ -(t-s) c_{\alpha} |r|^{\alpha} \Int_{\S^d} | \scprod{a,v} |^{\alpha} \tilde{\sigma} (\diff v) [ 1 - i \tan (\frac{\pi \alpha}{2}) \sign (r) \dfrac{\Int_{\S^d} | \scprod{a,v} |^{\alpha} \sign (\scprod{a,v}) \tilde{\sigma} (\diff v)}{\Int_{\S^d} | \scprod{a,v} |^{\alpha} \tilde{\sigma} (\diff v)}] \big\}  \\
		&= \varphi_{\tilde{L}_{(t-s)  \Int_{\S^d} | \scprod{a,v} |^{\alpha} \tilde{\sigma} (\diff v)}} (r),
	\end{align*}
	when comparing this formula with \eqref{eq: characteristic function 1dim}, since $p-q := {\Int_{\S^d} | \scprod{a,v} |^{\alpha} \sign (\scprod{a,v}) \tilde{\sigma} (\diff v)}({\Int_{\S^d} | \scprod{a,v} |^{\alpha} \tilde{\sigma} (\diff v)})^{-1}$ is independent of $a$ (see Remark \ref{remark: p-q ind a}) under our assumptions.\\
	Finally, we obtain $\tilde{L}_{(t-s)\Int_{\S^d} | \scprod{a,v} |^{\alpha} \tilde{\sigma} (\diff v)} \overset{d}{=} \sqrt[\alpha]{\Int_{\S^d} |\scprod{a,v}|^{\alpha} \tilde{\sigma}(\diff v)} \tilde{L}_{(t-s)} \overset{d}{=} \sqrt[\alpha]{\Int_{\S^d} |\scprod{a,v}|^{\alpha} \tilde{\sigma}(\diff v)} (\tilde{L}_t-\tilde{L}_s )$ due to the $\alpha$-stability of $\tilde{L}$.
\end{myproof}

\begin{myproof}[Proof of Lemma \ref{alpha dist}] 
	Here, we must distinguish the cases $\alpha=2$ and $\alpha<2$. If $\alpha=2$, then the L{\'e}vy process is a Brownian Motion, and hence, this lemma follows from the well-known properties of a Brownian Motion and Definition \ref{remark: alpha=2} as
	\begin{align*}
		\Int_{t+}^T f(s) \diff W_s^1 \overset{d}{=} \normal (0,\Int_{t+}^T f(s)^2 \diff s ) \overset{d}{=} \sqrt{\Int_{t+}^T f(s)^2 \diff s} W_1^1 \overset{d}{=} \sqrt{\Int_{t+}^T f(s)^2 \diff s} \tilde{L}_1.
	\end{align*}
	In the case $\alpha<2$, we have $\hat{L}=\tilde{L}$. Define $Z := {\Int_{t+}^T f(s) \diff \tilde{L}_s}({\sqrt[\alpha]{\Int_t^T |f(s)|^{\alpha} \diff s}})^{-1}$, and let us show that $Z \overset{d}{=} \tilde{L}_1$. We use the time equidistant discretizations $t = s_0 < \ldots < s_n = T$, $\Delta \tilde{L}_{s_k} = \tilde{L}_{s_k}-\tilde{L}_{s_{k-1}}$ and $\Delta s_k = s_k-s_{k-1}$ for $k \in \{1,\ldots,n\}$ to define $Z_n := ({\Sum_{k=1}^n f(s_k) \Delta \tilde{L}_{s_k}})({\sqrt[\alpha]{\Sum_{k=1}^n |f(s_k)|^{\alpha} \Delta s_k}})^{-1}$. We prove the result in three steps:
	
	In Step 1 we show $\varphi_{Z_n} = \varphi_{\tilde{L}_1}$:
	Using that the $\Delta \tilde{L}_{s_k}$'s are independent for different $k$'s, and (\ref{eq: characteristic function 1dim}), we get :
	\begin{align*}
		\varphi_{Z_n}(r) =& \E [\exp({irZ_n})] 
		=  \E \big[\exp \big( i\textfrac{r}{\sqrt[\alpha]{\Sum_{k=1}^n |f(s_k)|^{\alpha} \Delta s_k}} \Sum_{k=1}^n f(s_k) \Delta \tilde{L}_{s_k}\big)\big] \\
		=& \textstyle\prod_{k=1}^n \E \big[\exp \big(i\textfrac{r f(s_k)}{\sqrt[\alpha]{\Sum_{l=1}^n |f(s_l)|^{\alpha} \Delta s_l}} \Delta \tilde{L}_{s_k}\big)\big] \\
		=& \textstyle\prod_{k=1}^n \exp \big(- \Delta s_k c_{{\alpha}} | \textfrac{r f(s_k)}{\sqrt[\alpha]{\Sum_{l=1}^n |f(s_l)|^{\alpha} \Delta s_l}} |^{\alpha} [ 1 - i (p-q) \tan ( \frac{\pi \alpha}{2}) \sign(\textfrac{r f(s_k)}{\sqrt[\alpha]{\Sum_{l=1}^n |f(s_l)|^{\alpha} \Delta s_l}}) ]  \big) \\
		=& \exp (- c_{\alpha} \textfrac{1}{\Sum_{l=1}^n |f(s_l)|^{\alpha} \Delta s_l} \Sum_{k=1}^n |f(s_k)|^{\alpha} \Delta s_k |r |^{\alpha} [ 1 - i (p-q) \tan ( \frac{\pi \alpha}{2}) \sign(r) ] ) \\
		=& \exp (-c_{\alpha} |r |^{\alpha} [ 1 - i (p-q) \textstyle \tan ( \frac{\pi \alpha}{2}) \sign(r) ] ) = \E [ \exp (ir\tilde{L}_1) ] = \varphi_{\tilde{L}_1} (r),
	\end{align*}
	where we used that either $f \geq 0$ in the case of Assumptions \ref{ass: asymmetric} and \ref{ass: one dimensional} or $p-q=0$ in the case of Assumption \ref{ass: symmetric multidemensional}.
	
	In Step 2 we prove $Z_n \xrightarrow{\PP} Z$:
	For this purpose, we define $\phi_n = \Sum_{k=1}^n | f(s_{k}) |^{\alpha} \1_{(s_{k-1},s_{k}]}$ and $\chi_n = \Sum_{k=1}^n f(s_{k}) \1_{(s_{k-1},s_{k}]}$. We see that both functions are deterministic, $\phi_n \xrightarrow{n \to \infty} |f|^{\alpha}$, and $\chi_n \xrightarrow{n \to \infty} f$ pointwise on $(t,T]$ since $f$ is continuous and deterministic. Now, the convergence of Riemann sums and \eqref{eq: formula integral construction} give us with $\tilde{\lambda}$ being the Lebesgue measure: 
	\begin{align*}
		\lim\nolimits_{n \to \infty} \Sum_{k=1}^n |f(s_k)|^{\alpha} \Delta s_k &= \lim\nolimits_{n \to \infty} \Int_{(t,T]} \phi_n \diff \tilde{\lambda} = \Int_t^T |f(s)|^{\alpha} \diff s, \\
		\Sum_{k=1}^n f(s_k) \Delta \tilde{L}_{s_k} &= \Int_{(t,T]} \chi_n \diff \tilde{L}_s \xrightarrow{\PP} \Int_{t+}^T f(s) \diff \tilde{L}_s.
	\end{align*}
	These two limits give us the claim of Step 2.
	
	In Step 3 we finally show the claim $Z \overset{d}{=} \tilde{L}_1$:
	We know $\lim_{n \to \infty} Z_n \overset{\PP}{=} Z$ from Step 2, which implies the convergence of the characteristic functions, i.e., $\lim_{n \to \infty} \varphi_{Z_{n}} = \varphi_Z$. Due to Step 1, we then get $\varphi_{\tilde{L}_1} = \lim_{n \to \infty} \varphi_{\tilde{L}_1} = \lim_{n \to \infty} \varphi_{Z_{n}} = \varphi_Z$ which is equivalent to the claim.
\end{myproof}

\begin{myproof}[Proof of Proposition \ref{Jn concave}] 
	We prove these propositions with properties for compositions of concave and convex functions, see Z{\u{a}}linescu \cite[p.43]{Zalinescu}.
	
	First of all, we show that the first term of $J_n$ is concave in $u$:\\
	We know from (\ref{eq: umschreibung xt-xn}):
	\begin{align*}
		X_T^{u} - X_{n}^{u} e^{r(T-n)} &= \Sum_{i=1}^d \Int_{n}^T u^i_s (\mu^i_s-r) e^{r(T-s)} \diff s + \Sum_{i,j=1}^d  \Int_{n+}^T u^i_{s} \sigma^{ij}_s e^{r(T-s)} \diff L^j_s.
	\end{align*}
	Due to the linearity of the integrals, $X_T^{u} - X_{n}^{u} e^{r(T-n)}$ is a linear mapping. Since $\T$ is a concave function and the expected value is linear, it follows that $\mathbb{E}_{n,x} [ \T(X_T^{u} - x e^{r(T-n)}) ]$ is concave in $u$.
	
	The next step is to show that $\rho_{n,x} (X_T^{u} - x e^{r(T-n)})$ is convex in $u$. 
	We know from Proposition \ref{risk measure formula} that $\rho_{n,x} (X_T^u - x e^{r(T-n)}) = -\Int_n^T m_s (u) \diff s + \left\{ \begin{matrix*}[l]
		\varrho^{\hat{L}_1} \norm{u_s^\intercal \sigma_s v e^{r(T-s)}}_{\L^{\alpha}((n,T] \times \S^d, \diff s \times  \tilde{\sigma}(\diff v))} & \text{if }\alpha<2. \\
		\varrho^{W_1^1} \norm{u_{s}^\intercal \sigma_{s} e^{r(T-{s})}}_{\L^2((n,T],\diff s;H_s)} & \text{if }\alpha=2.
	\end{matrix*} \right.  $
	The first term is linear since $m_s$ is linear. Furthermore, since $\alpha>1$ by assumption, the second term consists in each case of an $\L^{\alpha}$-norm, which is convex due to the positive homogeneity and the triangular inequality of norms. Hence, the whole second term is convex in $u$ as well. Since $F$ is non-decreasing and convex, $F(\rho_{n,x} (X_T^{\hat{u}} - x e^{r(T-n)}))$ is convex in $u$.	Then we know that $-\lambda_{n} F(\rho_{n,x} (X_T^{u} - x e^{r(T-n)}))$ is concave in $u$ and so the functional $J_{n}$ is concave in $u$ as the sum of concave functions.
	
	If $\T$ or $-F$ is strictly concave, then the respective term is also strictly concave since $\sigma_t$ is positive definite. Hence, $J_n$ is strictly concave as well as the sum of a strictly concave and a concave function.
\end{myproof}

\begin{myproof}[Proof of Proposition \ref{u unif bounded}] 
	In this proof, we denote by $\leftidx{^\delta}{\hat{u}}$ the optimal strategy with respect to the step size $\delta$. However, before showing Proposition \ref{u unif bounded}, we prove the following two lemmas:
	\begin{lemma} \label{Jn non-negative}
		If the control function $u$ is deterministic and $\lambda_n \equiv \lambda$, then the functional $J_n$ is always non-negative in the optimum, i.e., $J_n (\hat{u}) \geq 0$ for all $n \in \{0,\ldots,T-1\}$.
	\end{lemma}
	\begin{proof}
		We show this proposition by induction, where the induction start is similar to the induction step. Thus, let $n \in \{0,\ldots,T-1\}$. Now, \eqref{eq: umschreibung xt-xn} and the assumption that $\lambda$ is constant imply that $J_n (x,u = (\hat{u}_{n+1},\hat{u}_{n+2},\ldots,\hat{u}_T)^\intercal) \geq J_n (x,u = (0,\hat{u}_{n+2},\ldots,\hat{u}_T)^\intercal) = J_{n+1} (x,u = (\hat{u}_{n+2},\ldots,\hat{u}_T)^\intercal) \geq 0$ by the induction step where the equality follows by Proposition \ref{risk measure formula}.
	\end{proof}
	
	\begin{lemma} \label{u inf w inf}
		We have $\norm{\leftidx{^{\alpha}}{w}_s (\leftidx{^\delta}{\hat{u}})}^{\alpha}_{\L^{\alpha} ((0,T],\diff s)} \geq \tilde{\lambda}_{min}^{\alpha} \min\{1,\mathbb{A}_d\} \cos^{\alpha} (\tilde{\varepsilon}) \hat{\varepsilon}  \norm{\leftidx{^\delta}{\hat{u}}_s}^{\alpha}_{\L^{\alpha} ((0,T],\diff s)} $ under Assumption \ref{ass: shift-inv and cash-inv and condition}.
	\end{lemma}

	\begin{proof}
		First, we can assume without loss of generality that $\leftidx{^\delta}{\hat{u}}_s \neq 0$ for all $s$. Otherwise, since $\leftidx{^\delta}{\hat{u}}_s$ is piecewise constant, the right hand side of the equation is already ``$0$'' in this interval and can be ignored in the following proof.
		
		Now, we show the lemma in the case that $\tilde{\sigma}$ has a density which is bigger than $\hat{\varepsilon}$ on the $(d-1)$-dimensional unit ball around a vector with respect to the radial distance with radius $12 \tilde{\varepsilon}$. We denote this vector by $w \in \S^d$, the unit ball on the sphere around $w$ by $O_w$, and the radial distance between two vectors by $d$. Moreover, we define the distance $\tilde{d}$ as the maximum norm on the sphere, i.e., using polar coordinates, one can describe a point in $\S^d$ by $d-1$ angles. For the maximal norm $\tilde{d}(x,y)$ on the sphere, the standard (componentwise) maximal norm is then applied to the $d-1$ angles between the two vectors $x$ and $y$ in their polar coordinates (see the set $A$ in Figure \ref{fig: sphere}).	
		
		\begin{floatingfigure}[r]{5cm}
			\mbox{\includegraphics[trim= 10mm 10mm -10mm 10mm, clip,width=4cm]{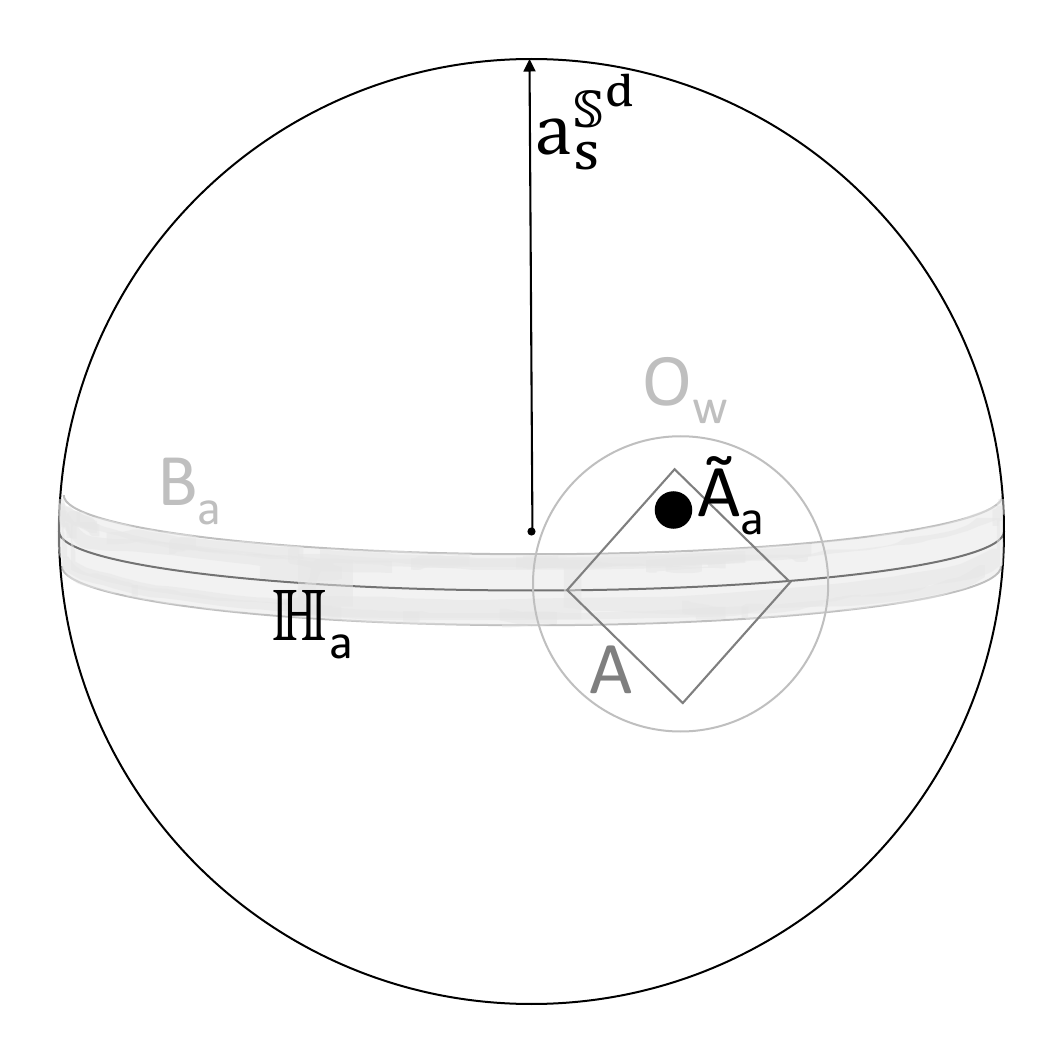}}
			\caption{Graphical illustration of the proof for $d=3$ if $w \in \mathbb{H}_a$. If $w \not\in \mathbb{H}_a$, the situation is simpler.}
			\label{fig: sphere}
		\end{floatingfigure} \pagebreak
		
		Then, it holds that $A:= \{\hat{w} \in \S^d | \tilde{d}(\hat{w},w) \leq 6 \tilde{\varepsilon} \} \subset O_w$ by construction. We can describe $A$ as some $(d-1)$-dimensional hyper-`cube' with curved sides. Define $a_s^{\S^d} (\leftidx{^\delta}{\hat{u}}) := \frac{\leftidx{^\delta}{\hat{u}}_s^\intercal \sigma_{s} e^{r(T-s)}}{\norm{\leftidx{^\delta}{\hat{u}}_s^\intercal \sigma_{s} e^{r(T-s)}}_2} \in \S^d$. Then, the set of vectors on $\S^d$ which are orthogonal to $a_s^{\S^d}$ form a $(d-2)$-dimensional curved hyperplane, denoted by $\mathbb{H}_a$. By orthogonality, for all $\tilde{w} \in \mathbb{H}_a$ the angle between $a_s^{\S^d} (\leftidx{^\delta}{\hat{u}})$ and $\tilde{w}$ is $\frac{\pi}{2}$. Now, define $B_a := \{ \hat{w} \in \S^d | dist(\hat{w},\mathbb{H}_a) \leq \tilde{\varepsilon} \} \subset \S^d$, where $dist$ denotes the radial distance between a vector and a set on $\S^d$, i.e., $dist(x,C)=\inf_{y \in C} d(x,y)$ for $x \in \S^d$ and $C \subset \S^d$. Hence, $B_a$ is a hyperspherical segment. By construction, there exists a point $\hat{v}_a \in A \backslash B_a$ such that $\tilde{A}_a := \{\hat{w} \in \S^d | d(\hat{w},\hat{v}_a) \leq \tilde{\varepsilon} \} \subset A \backslash B_a$. In particular, it holds for all $\hat{w} \in \tilde{A}_a$ that the angle between $\hat{w}$ and $a_s^{\S^d}$ is at least $\tilde{\varepsilon}$ away from $\frac{\pi}{2}$ and the density of $\tilde{\sigma}$ is bigger or equal than $\hat{\varepsilon}$. Moreover, note that the surface area of $\tilde{A}_a$ is equal to $\mathbb{A}_d$ with $\mathbb{A}_d$ as in Assumption \ref{ass: shift-inv and cash-inv and condition}.
		Hence, we get since $|a \cdot v| = |\cos(\phi)| |a|$ for $v \in \S^d$ and $\phi$ denoting the angle between $a$ and $v$:
		\begin{align*}
			\norm{\leftidx{^{\alpha}}{w}_s (\leftidx{^\delta}{\hat{u}})}_{\L^{\alpha} ((0,T],\diff s)}^{\alpha} &= \Int_0^T \Int_{\S^d} | \leftidx{^\delta}{\hat{u}}_s^\intercal \sigma_{s} e^{r(T-s)} v |^{\alpha} \tilde{\sigma}(\diff v) \diff s \\
			&\geq |\cos|^{\alpha}(\tilde{\varepsilon}) \Int_0^T | \leftidx{^\delta}{\hat{u}}_s^\intercal \sigma_{s} e^{r(T-s)} |^{\alpha} \Int_{\tilde{A}_a} \tilde{\sigma}(\diff v) \diff s \\
			&\geq \mathbb{A}_d \cos^{\alpha}(\tilde{\varepsilon}) \hat{\varepsilon} \Int_0^T | \leftidx{^\delta}{\hat{u}}_s^\intercal \sigma_{s} e^{r(T-s)} |^{\alpha}  \diff s 
			\geq \mathbb{A}_d \cos^{\alpha}(\tilde{\varepsilon}) \hat{\varepsilon} \tilde{\lambda}^{\alpha}_{min} \Int_0^T | \leftidx{^\delta}{\hat{u}}_s^\intercal  |^{\alpha}  \diff s,
		\end{align*}
		i.e., the claim follows. For the other case, the argumentation is similar, just without $\mathbb{A}_d$. \vspace{12pt}
	\end{proof} \linebreak 
	Now, we are ready to prove Proposition \ref{u unif bounded}: \\
	In the case of Assumption \ref{ass: bounded strategies}, the claims follow directly by possibly increasing $M$. Otherwise, we show this proposition by contradiction, i.e., we assume that $\norm{\leftidx{^\delta}{\hat{u}}}_{\L^{\alpha} ((0,T],\diff s)} \xrightarrow{\delta \to 0} \infty$ by possibly switching to a subsequence. Now, since $\lim_{x \to \infty} \T'(x) =0$, for all $\gamma_\T >0$ exists a $C_{\gamma}>0$ such that $\T(x) \leq C_{\gamma} + \gamma_\T x$ and since $F$ is convex, there exists a $C_F \leq 0$ and a $\gamma_F>0$ such that $F(x)\geq C_F + \gamma_F x$. Also, note that $\norm{\cdot}_{\L^1 ((0,T],\diff s)} \leq C_{\alpha} \norm{\cdot}_{\L^{\alpha} ((0,T],\diff s)}$ with $C_{\alpha} := T^{1-1/\alpha}$ since $\alpha>1$ by assumption. Then, we get:
	\begin{align*}
		J_0 (x,\leftidx{^\delta}{\hat{u}}) =& \mathbb{E}_{0,x} [\T(\Int_{0}^T \leftidx{^\delta}{\hat{u}}^\intercal_s (\mu_s-r) e^{r(T-s)} \diff s +  \Int_{0+}^T \leftidx{^\delta}{\hat{u}}^\intercal_{s} \sigma_s e^{r(T-s)} \diff L_s)]  \\
		&- \lambda F(-\Int_0^T m_s (\leftidx{^\delta}{\hat{u}}) \diff s + \varrho^{\hat{L}_1} \norm{\leftidx{^{\alpha}}w_s (\leftidx{^\delta}{\hat{u}})}_{\L^{\alpha}((0,T],\diff s)}) \\
		\leq& C_{\gamma} + \gamma_\T \tilde{\mu}_{max} C_{\alpha} \norm{\leftidx{^\delta}{\hat{u}}_s}_{\L^{\alpha} ((0,T],\diff s)} - \lambda C_F + \lambda \gamma_F \Int_0^T |m_s (\leftidx{^\delta}{\hat{u}})| \diff s - \lambda \gamma_F \varrho^{\hat{L}_1} \norm{\leftidx{^{\alpha}}w_s (\leftidx{^\delta}{\hat{u}})}_{\L^{\alpha}((0,T],\diff s)},
	\end{align*} 
	where we recall that $\tilde{\mu}_{max} = \max\{|\mu^i_s-r|e^{r(T-s)} | s \in [0,T], i \in \{1,\ldots,d\}\}$. We used Proposition \ref{risk measure formula} in the last inequality.
	Now, we consider Assumption \ref{ass: shift-inv and cash-inv and condition}. If $\rho$ is shift-invariant, we get with Lemma \ref{u inf w inf}:
	\begin{align*}
		J_0 (x,\leftidx{^\delta}{\hat{u}}) \leq& C_{\gamma} + \gamma_\T \tilde{\mu}_{max} C_{\alpha} \norm{\leftidx{^\delta}{\hat{u}}_s}_{\L^{\alpha} ((0,T],\diff s)} - \lambda C_F \\
		&- \lambda \gamma_F \varrho^{\hat{L}_1} \tilde{\lambda}_{min} \min\{1,\mathbb{A}_d\}^{1/\alpha} \cos (\tilde{\varepsilon}) \hat{\varepsilon}^{1/\alpha}  \norm{\leftidx{^\delta}{\hat{u}}_s}_{\L^{\alpha} ((0,T],\diff s)}.
	\end{align*}
	Choosing $\gamma_\T < \frac{\lambda \gamma_F \varrho^{\hat{L}_1} \tilde{\lambda}_{min} \min\{1,\mathbb{A}_d\}^{1/\alpha} \cos (\tilde{\varepsilon}) \hat{\varepsilon}^{1/\alpha}}{\tilde{\mu}_{max} C_{\alpha}}$ leads to $J_0 (x,u) \xrightarrow{\delta \to 0} - \infty$ which is a contradiction to Lemma \ref{Jn non-negative} and thus we get the claim. If $\rho$ is cash-invariant, we get with Lemma \ref{u inf w inf} that
	\begin{align*}
		J_0 (x,u) \leq& C_{\gamma}- \lambda C_F+ \norm{\leftidx{^\delta}{\hat{u}}_s}_{\L^{\alpha} ((0,T],\diff s)} (\gamma_\T \tilde{\mu}_{max} C_{\alpha}+ \lambda \gamma_F \tilde{\mu}_{max} \\
		&- \lambda \gamma_F \varrho^{\hat{L}_1} \tilde{\lambda}_{min} \min\{1,\mathbb{A}_d\}^{1/\alpha} \cos (\tilde{\varepsilon}) \hat{\varepsilon}^{1/\alpha}). 
	\end{align*}
	Now, due to the assumption that $\varrho^{\hat{L}_1} \tilde{\lambda}_{min} \min\{1,\mathbb{A}_d\}^{1/\alpha} \cos (\tilde{\varepsilon}) \hat{\varepsilon}^{1/\alpha} > \tilde{\mu}_{max}$, it is possible to choose $0<\gamma_\T< \frac{\lambda \gamma_F (\varrho^{\hat{L}_1} \tilde{\lambda}_{min} \min\{1,\mathbb{A}_d\}^{1/\alpha} \cos (\tilde{\varepsilon}) \hat{\varepsilon}^{1/\alpha} -\tilde{\mu}_{max})}{\tilde{\mu}_{max} C_{\alpha}}$ and the contradiction can be derived similarly. The case of Assumption \ref{ass: one dimensional 2} is a simpler variant of this argument where \eqref{eq: one dimensional simplification} replaces Lemma \ref{u inf w inf}. Also, Assumption \ref{ass: alpha=2} leads to a simpler variant after noticing that $\norm{\leftidx{^{\alpha}}{w}_s (\leftidx{^\delta}{\hat{u}})}_{\L^{2} ((0,T],\diff s)}^{2} \geq \tilde{\lambda}_{min}^2 \norm{\leftidx{^\delta}{\hat{u}_s}}_{\L^{2} ((0,T],\diff s)}^{2}$ by the definition of $\tilde{\lambda}_{min}^2$.
\end{myproof}

\begin{myproof}[Proof of Proposition \ref{Jn cont}]
	Plugging Equation \eqref{eq: umschreibung xt-xn} and Proposition \ref{risk measure formula} into Equation \eqref{J def} leads to:
	\begin{align*}
		J_n (x,u) =& \mathbb{E}_{n,x} [\T(\Int_{n}^T u^\intercal_s (\mu_s-r) e^{r(T-s)} \diff s +  \Int_{n+}^T u^\intercal_{s} \sigma_s e^{r(T-s)} \diff L_s)]  \\
		&- \lambda_n F(-\Int_n^T m_s (u) \diff s + \varrho^{\hat{L}_1} \norm{\leftidx{^{\alpha}}w_s (u)}_{\L^{\alpha}((n,T],\diff s)}).
	\end{align*} 
	That $m$ and $w$ are continuous in $u$ follows directly from their definitions. Hence, the second part of $J_n$ is continuous in $u$ since $F \in C^1$.
	In the first part, it follows similar to the proof of Proposition \ref{risk measure formula} that $\Int_{n+}^T u^\intercal_{s} \sigma_s e^{r(T-s)} \diff L_s \overset{d}{=} \norm{\leftidx{^{\alpha}}w_s (u)}_{\L^{\alpha}((n,T],\diff s)} \hat{L}_1$.
	Moreover, if $u^k \xrightarrow{k \to \infty} u$, $u$ and $u^k$ are uniformly bounded converging vectors and therefore uniformly bounded. 
	Hence, Assumption \ref{assumption expected value} and Lebesgue's convergence theorem (since $\T$ is monotone) also give us the continuity of $J_n$ in $u$ for the first part.
\end{myproof}

\begin{myproof}[Proof of Proposition \ref{un non-negative}] 
	Denote the optimal deterministic control function by $\hat{u}$. We show the desired property by backward induction. Therefore, we show at each time point $n \in \{0,\ldots,T-1\}$ that the derivative of $J_n$ with respect to ${\hat{u}}_{n+1}$ at the point $0+$ is non-negative. Combining this with the concavity of $J_n$ (Proposition \ref{Jn concave}) and the independence of $x$ (Theorem \ref{deterministic control function}), i.e., we can identify $J_n (x,{\hat{u}}) = J_n({\hat{u}})$ (see Equation (\ref{J def})), giving the claim. Indeed, for a maximum of a differentiable one-dimensional concave function, the derivative is then non-negative for all smaller values and non-positive for bigger values.
	
	The induction start is similar to the induction step. Hence, we only show the induction step for an arbitrary $n \in \{0,\ldots,T-1\}$:\\
	Using Equation (\ref{eq: umschreibung xt-xn}), where the existence of the expected value is guaranteed due to the second assumption of this theorem, the monotonicity of $\T$ and $\T'$, and Proposition \ref{risk measure formula}, it holds that
	\begin{align*} 
		\textfrac{\partial}{\partial {\hat{u}}_{n+1}^1} J_n ({\hat{u}}) =& \E_{n,x} [ \T' ( \Int_{n}^T {\hat{u}}^1_s (\mu^1_s-r) e^{r(T-s)} \diff s + \Int_{n+}^T {\hat{u}}^1_{s} \sigma^{11}_s e^{r(T-s)} \diff L^1_s )   \notag \\
		&\hspace{1cm} \cdot   ( \Int_{n}^{n+1} (\mu^1_s-r) e^{r(T-s)} \diff s + \Int_{n+}^{n+1} \sigma^{11}_s e^{r(T-s)} \diff L^1_s ) ] \notag \\
		&- \lambda_n F' ( -\Int_n^T {\hat{u}}_s^1 (\mu^1_s - r) e^{r(T-s)}\diff s + \varrho^{\hat{L}_1} \norm{\hat{u}_s^1 \sigma_{s}^{11} e^{r(T-s)}}_{\L^{\alpha}((n,T], \diff s)} ) \notag \\
		&\hspace{0.5cm}\cdot ( -\Int_n^{n+1} (\mu^1_s - r) e^{r(T-s)}\diff s + \varrho^{\hat{L}_1} \textfrac{|\hat{u}_{n+1}^1 |^{\alpha-1} \sign(\hat{u}_{n+1}^1) \Int_n^{n+1} | \sigma_{s}^{11} e^{r(T-s)}|^{\alpha} \diff s }{\norm{\hat{u}_s^1 \sigma_{s}^{11} e^{r(T-s)}}^{\alpha-1}_{\L^{\alpha}((n,T], \diff s}}  ).
	\end{align*}
	We show that the last fraction is bounded. Hence, the second part of the term converges to $0$ for $u$ going to $0+$. 
	We get with the H{\"o}lder inequality using $p = \frac{\alpha}{\alpha-1}$ and $q=\alpha$:
	\begin{align*}
		\Int_n^{n+1} |\hat{u}_{n+1}^1 |^{\alpha-1} \sign(\hat{u}_{n+1}^1) &| \sigma_{s}^{11} e^{r(T-s)}|^{\alpha} \diff s	\\
		&\leq \norm{|\hat{u}_{n+1}^1 \sigma_{s}^{11} e^{r(T-s)} |^{\alpha-1} | \sigma_{s}^{11} e^{r(T-s)}| }_{\L^{1}((n,n+1], \diff s)} \\
		&\leq \norm{|\hat{u}_{n+1}^1 \sigma_{s}^{11} e^{r(T-s)}|^{\alpha-1}}_{\L^{p}((n,n+1], \diff s)} \norm{\sigma_{s}^{11} e^{r(T-s)}}_{\L^{q}((n,n+1], \diff s)} \\
		&\leq \norm{\hat{u}_{s}^1 \sigma_{s}^{11} e^{r(T-s)}}_{\L^{\alpha}((n,T], \diff s)}^{\alpha-1} \norm{\sigma_{s}^{11} e^{r(T-s)}}_{\L^{\alpha}((n,n+1], \diff s)}.
	\end{align*}
	Thus, the fraction is bounded by $\norm{\sigma_{s}^{11} e^{r(T-s)}}_{\L^{\alpha}((n,n+1], \diff s)}$ ($\leq \max_{\{s \in [0,T]\}} \sigma_s^{11} e^{rT} T$). Since $L^1$ is symmetric by assumption, the L{\'e}vy measure $\nu$ is also symmetric, and $L^1$ is therefore a martingale. Thus, we get the following when inserting $0+$ into the derivative:
	\begin{align*}
		\textfrac{\partial}{\partial {\hat{u}}_{n+1}^1} J_n (0+) &= \E_{n,x} [ \T' ( 0 ) ( \Int_{n}^{n+1} (\mu^1_s-r) e^{r(T-s)} \diff s + \Int_{n+}^{n+1} \sigma^{11}_s e^{r(T-s)} \diff L^1_s ) ] - 0 \\
		&= \T' (0) \Int_{n}^{n+1} (\mu^1_s-r) e^{r(T-s)} \diff s \geq 0
	\end{align*}
	due to the assumption $\mu_t \geq r$ for all $t$, and $\T$ being an increasing function.
\end{myproof}

\begin{myproof}[Proof of Theorem \ref{hjb discrete}]
	
	We start with the main formula of the HJB equation. We notice that the value functional $J_n$ (see Equation (\ref{J def})) satisfies the following recursion formula:
	\begin{align} \label{rec J}
		J_n (x,u) = \ &\mathbb{E}_{n,x} [J_{n+1} (X_{n+1}^u,u)] -\mathbb{E}_{n,x} [\mathbb{E}_{n+1,X_{n+1}^u} [\T(X_T^u-X_{n+1}^u e^{r(T-n-1)})]-\T(X_T^u-x e^{r(T-n)})] \notag \\
		&+ \mathbb{E}_{n,x} [\lambda_{n+1} F(\rho_{n+1,X_{n+1}^u} (X_T^u - X_{n+1}^u e^{r(T-n-1)})) - \lambda_{n} F(\rho_{n,x} (X_T^u - x e^{r(T-n)}))]. 
	\end{align}
	This characterization and the definition of $V$ (see Definition \ref{def: Nash equilibrium disc}) imply that
	\begin{align*}
		\sup\nolimits_{u \in \mathcal{V}} &\{ \mathbb{E}_{n,x} [J_{n+1} (X_{n+1}^u,\bar{u})] - V_n (x)  \\
		&-\mathbb{E}_{n,x} [\mathbb{E}_{n+1,X_{n+1}^u} [\T(X_T^{\bar{u}}-X_{n+1}^u e^{r(T-n-1)})]-\T(X_T^{\bar{u}}-x e^{r(T-n)})] \\
		&+   \mathbb{E}_{n,x} [\lambda_{n+1} F(\rho_{n+1,X_{n+1}^u} (X_T^{\bar{u}} - X_{n+1}^u e^{r(T-n-1)})) - \lambda_{n}F(\rho_{n,x} (X_T^{\bar{u}} - x e^{r(T-n)}))] \} = 0.
	\end{align*}
	By the definition of Nash equilibria (see Definition \ref{def: Nash equilibrium disc}),	we get $\mathbb{E}_{n,x} [J_{n+1} (X_{n+1}^u,\bar{u})] = V_{n+1} (X_{n+1}^u)$. With that and the infinitesimal operator $\mathrm{A}^u$, this recursion simplifies to the main part of our HJB equation:
	\begin{align*}
		&\sup\nolimits_{u \in \mathcal{V}} \{ (\mathrm{A}^u V)_n (x) - (\mathrm{A}^u \mathbb{E}_{\cdot} [\T(X_T^{\bar{u}} - X_{\cdot}^u e^{r(T-\cdot)})])_n (x)  +   (\mathrm{A}^u (\lambda_{\cdot}F(\rho_{\cdot} (X_T^{\bar{u}} - X_{\cdot}^u e^{r(T-\cdot)}))))_n (x)
		\} = 0.
	\end{align*}
	For the final value, we get:
	\begin{align*}
		V_T (x) &= J_T (X_T^{\hat{u}},\hat{u}) = \mathbb{E}_{T,X_T^u} [\T(X_T^{\hat{u}}- X_T^{\hat{u}} e^{r(T-T)})]  - \lambda_T F(\rho_{T,X_T^u} (X_T^{\hat{u}} - X_T^{\hat{u}} e^{r(T-T)})) \notag \\ 
		&= \T(X_T^{\hat{u}} - X_T^{\hat{u}}) - \lambda_T F(\rho_{T,X_T^u} (0)) = \T(0) - \lambda_T F(0).
	\end{align*}
	
	It remains to prove the formulas for the constraints where the first one is simply the definition of the infinitesimal generator (see Definition \ref{inf gen def}). We show the formula for the second constraint by backward induction:	
	
	\underline{$n=T$:}
	$\lambda_T F(\rho_{T,x} (X_T^{\hat{u}} - x e^{r(T-T)})) = \lambda_T F(\rho_{T,x} (0)) = \lambda_T F(0)$.
	
	\underline{$n=T-1$:}
	This is a more straightforward case than for $T-2$ and will be omitted.
	
	\underline{$n=T-2$:}
	Due to the optimal control function being deterministic (Theorem \ref{deterministic control function}), we can use Proposition \ref{risk measure formula} to calculate:
	\begin{align*}
		&(\mathrm{A}^{\hat{u}} (\lambda_{\cdot} F(\rho_{\cdot} (X_T^{\hat{u}} - X_{\cdot}^{\hat{u}} e^{r(T-\cdot)}))))_{T-2} (x) \\
		&\hspace{0.6cm}= \lambda_{T-1} \mathbb{E}_{T-2,x} [F( -\Int_{T-1}^T m_s (\hat{u}) \diff s + \varrho^{\hat{L}_1} \norm{\leftidx{^{\alpha}}w_s (\hat{u})}_{\L^{\alpha}((T-1,T], \diff s)} ) ] \\
		&\hspace{1cm} - \lambda_{T-2} F( -\Int_{T-2}^T m_s (\hat{u}) \diff s + \varrho^{\hat{L}_1} \norm{\leftidx{^{\alpha}}w_s (\hat{u})}_{\L^{\alpha}((T-2,T], \diff s)} ) \\
		&\hspace{0.6cm}= \lambda_{T-1} [ F (-\Int_{T-1}^T m_s (\hat{u}) \diff s + \varrho^{\hat{L}_1} \norm{\leftidx{^{\alpha}}w_s (\hat{u})}_{\L^{\alpha}((T-1,T], \diff s)} )   \\
		&\hspace{2.2cm}-  F (-\Int_{T-2}^T m_s (\hat{u}) \diff s + \varrho^{\hat{L}_1} \norm{\leftidx{^{\alpha}}w_s (\hat{u})}_{\L^{\alpha}((T-2,T], \diff s)} ) ] \\
		&\hspace{1cm}- (\lambda_{T-2}-\lambda_{T-1}) F(-\Int_{T-2}^T m_s (\hat{u}) \diff s + \varrho^{\hat{L}_1} \norm{\leftidx{^{\alpha}}w_s (\hat{u})}_{\L^{\alpha}((T-2,T], \diff s)} ).
	\end{align*}
	
	\underline{for all other $n$:} 
	Similarly.
\end{myproof}

\begin{myproof}[Proof of Lemma \ref{au exp cont}]
	Using the definition for the period from $t$ to $t+h$, it holds that
	\begin{align*}
		(\mathrm{A}^{\hat{u}} \mathbb{E}_{\cdot} &[\T(X_T^{\hat{u}} - X_{\cdot}^{\hat{u}} e^{r(T-\cdot)})]) (t,x) \\
		&= \lim\nolimits_{h \rightarrow 0} h^{-1} (\mathrm{A}_h^{\hat{u}} \mathbb{E}_{\cdot} [\T(X_T^{\hat{u}} - X_{\cdot}^{\hat{u}} e^{r(T-\cdot)})]) (t,x) \\
		&= \lim\nolimits_{h \rightarrow 0} h^{-1} \mathbb{E}_{t,x}[{\mathbb{E}_{t+h,X_{t+h}^{\hat{u}}} [ \T(X_T^{\hat{u}} - X_{t+h}^{\hat{u}} e^{r(T-t-h)})] -  (\T(X_T^{\hat{u}} - x e^{r(T-t)}))}] \\
		&= \lim\nolimits_{h \rightarrow 0} h^{-1} ({\mathbb{E} [ \T(X_T^{\hat{u}} - X_{t+h}^{\hat{u}} e^{r(T-t-h)})] - \mathbb{E} [ \T(X_T^{\hat{u}} - X_t^{\hat{u}} e^{r(T-t)})]}) \\
		&=   \textfrac{\partial}{\partial h} \mathbb{E} [ \T(X_T^{\hat{u}} - X_{t+h}^{\hat{u}} e^{r(T-t-h)})] |_{h=0},
	\end{align*}
	where we used in the third equation that the term inside the expected value is deterministic. The existence of the derivative follows immediately from Proposition \ref{discussion continuous}.
\end{myproof}

\begin{myproof}[Proof of Lemma \ref{risk measure formula in continuous time}]
	It holds with Proposition \ref{risk measure formula} and the second constraint of Theorem \ref{hjb discrete} used for the interval from $t$ to $t+h$ that
	\begin{align*}
		&(\mathrm{A}^{\hat{u}} (\lambda_{\cdot} F(\rho_{\cdot} (X_T^{\hat{u}} - X_{\cdot}^{\hat{u}} e^{r(T-\cdot)})))) (t,x) \\
		&\hspace{0.6cm}= \lim\nolimits_{h \rightarrow 0} h^{-1} (\mathrm{A}_h^{\hat{u}} (\lambda_{\cdot}F(\rho_{\cdot} (X_T^{\hat{u}} - X_{\cdot}^{\hat{u}} e^{r(T-\cdot)})))) (t,x))\\
		&\hspace{0.6cm}= \lim\nolimits_{h \rightarrow 0} \lambda_{t+h} h^{-1} [ F (-\Int_{t+h}^T m_s (\hat{u}) \diff s + \varrho^{\hat{L}_1} \norm{\leftidx{^{\alpha}}w_s (\hat{u})}_{\L^{\alpha}((t+h,T], \diff s)} )   \\
		&\hspace{3.1cm}-  F (-\Int_{t}^T m_s (\hat{u}) \diff s + \varrho^{\hat{L}_1} \norm{\leftidx{^{\alpha}}w_s (\hat{u})}_{\L^{\alpha}((t,T], \diff s)} ) ] \\
		&\hspace{1cm}- \lim\nolimits_{h \rightarrow 0} h^{-1} (\lambda_t-\lambda_{t+h}) F(-\Int_{t}^T m_s (\hat{u}) \diff s + \varrho^{\hat{L}_1} \norm{\leftidx{^{\alpha}}w_s (\hat{u})}_{\L^{\alpha}((t,T], \diff s)} ) \\
		&\hspace{0.6cm}= \lambda_t F'(-\Int_{t}^T m_s (\hat{u}) \diff s + \varrho^{\hat{L}_1} \norm{\leftidx{^{\alpha}}w_s (\hat{u})}_{\L^{\alpha}((t,T], \diff s)} ) [ m_t (\hat{u}) - \varrho^{\hat{L}_1} \leftidx{^{\alpha}}w_t (\hat{u}) ( \Int_t^T | \leftidx{^{\alpha}}w_s (\hat{u}) |^{\alpha}  \diff s )^{\frac{1}{\alpha}-1} ] \\
		&\hspace{1cm}+ \lambda'_t F(-\Int_{t}^T m_s (\hat{u}) \diff s + \varrho^{\hat{L}_1} \norm{\leftidx{^{\alpha}}w_s (\hat{u})}_{\L^{\alpha}((t,T], \diff s)} ).
	\end{align*}
	The chain rule may be used since $F$ is in $C^1$ and the optimal control function and the other functions are continuous.
\end{myproof}

\begin{lemma} \label{arzela lemma}
	Let $k \in \N$, $u \in (0,M)^k$, $\alpha,q>0$, $K_1,K_3 \geq 0$, $K_2>0$, and $K_4 \in \R$. Now, consider $f=(f_1,\ldots,f_k)$ defined by: $f_i (u) := (u_1^\alpha + \ldots + u_k^\alpha + K_1)K_2u_i^q +K_3u_i^q+K_4$. Then the Jacobi matrix $J_f$ is invertible.
\end{lemma}

\begin{proof}
	It holds for $i \neq j$ that $\textfrac{\partial f_i}{\partial u_j} = \alpha u_j^{\alpha-1} K_2 u_i^q$, and
	\begin{align*}
		\textfrac{\partial f_i}{\partial u_i} &= \alpha u_i^{\alpha-1} K_2 u_i^q + q(u_1^\alpha + \ldots + u_k^\alpha + K_1)K_2 u_i^{q-1} + K_3 q u_i^{q-1}.
	\end{align*}
	Define $D=(d_{1},\ldots,d_{k})$ as the invertible diagonal matrix with positive diagonal values $d_{i} = q(u_1^\alpha + \ldots + u_k^\alpha + K_1)K_2 u_i^{q-1} + K_3 q u_i^{q-1}>0$ for all $i \in \{1,\ldots,k\}$ (where for simplicity we only specify the non-zero values) and the matrix $M=(m_{ij})_{ij}$ with $m_{ij} = \alpha u_j^{\alpha-1} K_2 u_i^q$ for all $i,j \in \{1,\ldots,k\}$. Then, $J_f = M + D$. Note that $M= \alpha K_2 v w^\intercal$ where $v$ and $w$ are two column vectors with positive entries $v_i = u_i^{\alpha-1}$ and $w_i = u_i^q$ for $i \in \{1,\ldots,k\}$ and $v w^\intercal$ is the outer product. Now, the Matrix determinant lemma resp. the Sherman-Morrison formula implies that $J_f$ is also invertible.
\end{proof}

\begin{myproof}[Proof of Proposition \ref{discussion continuous}]
	For the first claim, we note that $\rho_{t,x} (X_T^{u} - x e^{r(T-t)})$ does not depend on $x$ since we know from Proposition \ref{risk measure formula} that $\rho_{t,x} (X_T^{u} - x e^{r(T-t)}) = -\Int_t^T m_s ({u}) \diff s + \varrho^{\hat{L}_1} \norm{\leftidx{^{\alpha}}w_s ({u})}_{\L^{\alpha}((t,T], \diff s)}$.
	Moreover, $\rho_{t,x} (X_T^{u} - x e^{r(T-t)})$ is continuously differentiable in $t$ if ${u}$ is continuous, so it follows that $\textfrac{\partial}{\partial t} \rho_{t,x} (X_T^{u} - x e^{r(T-t)}) = m_t ({u}) - \varrho^{L^1_1} {\leftidx{^{\alpha}}w_t ({u})}({\alpha ( \Int_t^T \leftidx{^{\alpha}}w_s ({u}) \diff s )^{1-\frac{1}{\alpha}} })^{-1}$.
	Hence, we get:
	\begin{align*}
		\textfrac{\partial}{\partial t} F(\rho_{t,x} (X_T^{u} - x e^{r(T-t)})) =& (m_t ({u}) - \varrho^{L^1_1} {\leftidx{^{\alpha}}w_t ({u})}({\alpha ( \Int_t^T \leftidx{^{\alpha}}w_s ({u}) \diff s )^{1-\frac{1}{\alpha}} })^{-1} ) \\
		&\cdot F' (-\Int_t^T m_s ({u}) \diff s + \varrho^{\hat{L}_1} \norm{\leftidx{^{\alpha}}w_s ({u})}_{\L^{\alpha}((t,T], \diff s)} ).
	\end{align*}
	By definition, we know that $F$ is in $C^1$, and $m$ and $w$ are continuous. In addition, we know that ${u}$ is continuous by assumption, and $\lambda_t \in C^{1}$ by Assumption \ref{ass: lambda C1}. Hence, we get the first claim.\\
	For proving the second claim, we have from (\ref{eq: umschreibung xt-xn}) and Lemmas \ref{alpha additive} and \ref{alpha dist} used as in the proof of Proposition \ref{risk measure formula}:
	\begin{align*}
		X_T^{u} - x e^{r(T-t)} &= \Sum_{i=1}^d \Int_{t}^T {u}^i_s (\mu^i_s-r) e^{r(T-s)} \diff s + \Sum_{i,j=1}^d  \Int_{t+}^T {u}^i_{s} \sigma^{ij}_s e^{r(T-s)} \diff L^j_s \overset{d}{=} a_t + b_t \hat{L}_1,
	\end{align*}
	where $a_t := \Int_t^T u_s^\intercal (\mu_s-r) e^{r(T-s)} \diff s$, $b_t := \norm{\leftidx{^{\alpha}}w_s (u)}_{\L^{\alpha}((t,T],\diff s)}$, and $m$, $w$, and $\hat{L}$ are given in Definition \ref{def m w}.\\
	First, we notice again that $X_T^{u} - x e^{r(T-t)}$ does not depend on $x$, and hence, of course,  $X_T^{u} - x e^{r(T-t)} \in C^{\infty}$ and $\mathbb{E}_{t,x} [\T(X_T^{u} - x e^{r(T-t)})] \in C^{\infty}$ as functions of $x$. Furthermore, we get with $u$ being deterministic by assumption, the cdf $F_{\hat{L}_1}$ of $\hat{L}_1$ and Assumptions \ref{assumption expected value} and \ref{assumption expected value derivative}:
	\begin{align*}
		\textfrac{\partial}{\partial t} \E_{t,x} [ \T (X_T^{u} - x e^{r(T-t)})] &= \textfrac{\partial}{\partial t} \E [ \T (a_t + b_t \hat{L}_1)] \\
		&= \textfrac{\partial}{\partial t} \Int_{-\infty}^{\infty} \T (a_t + b_t y ) \diff F_{\hat{L}_1} (y) \\
		&= a_t' \cdot \Int_{-\infty}^{\infty} \T' (a_t + b_t y ) \diff F_{\hat{L}_1} (y) + b_t' \cdot \Int_{-\infty}^{\infty} y \T' (a_t + b_t y ) \diff F_{\hat{L}_1} (y) \\ 
		&= a_t' \E [ \T' (a_t + b_t \hat{L}_1)] + b_t' \E [ \hat{L}_1 \T' (a_t + b_t \hat{L}_1)].
	\end{align*}
	Since $u$ is continuous by assumption and $b_t \geq 0$ with $b_t \equiv 0$ if and only if $u_t \equiv 0$ by definition, we can conclude that $a_t, b_t \in C^1$. Moreover, $\T \in C^1$ and $\T$ and $\T'$ are monotone by the model setup. Hence, Assumptions \ref{assumption expected value} and \ref{assumption expected value derivative} allow us to apply Lebesgue's dominated convergence theorem, which then yields that $\E [ \T' (a_t + b_t \hat{L}_1)]$ resp. $\E [ \hat{L}_1 \T' (a_t + b_t \hat{L}_1)]$ is continuous in $t$ which implies the second claim.
\end{myproof}

\begin{myproof}[Proof of Lemma \ref{lemma: infinitesimal generator = derivative}]
	We provide the proof of the first equation only, as the proof of the second follows analogously. Using the notation $w := u_h^{u_t,t}$, we obtain:
	\begin{align*}
		(\mathrm{A}^u \mathbb{E}_{\cdot} &[\T(X_T^u - X_{\cdot}^u e^{r(T-\cdot)})]) (t,x) = \lim_{h \to 0} \tfrac{1}{h} (\mathrm{A}_h^w \mathbb{E}_{\cdot} [\T(X_T^u - X_{\cdot}^u e^{r(T-\cdot)})]) (t,x) \\
		&= \lim_{h \to 0} \tfrac{1}{h} \E_{t,x} [\E_{t+h,X_{t+h}^w}[\T(X_T^w - X_{t+h}^{u_t} e^{r(T-(t+h))}) ] - \E_{t,x}[\T(X_T^w - x e^{r(T-t)}) ] ] \\
		&=\lim_{h \to 0} \tfrac{1}{h} \E[\T(a_{t+h}+b_{t+h}\hat{L}_1) \\
		&\hspace{50pt}-\T(\Int_t^{t+h} u_t^\intercal (\mu_s-r) e^{r(T-s)} \diff s + a_{t+h} + \sqrt[\alpha]{\Int_t^{t+h} \leftidx{^{\alpha}}w_s (u_t)^\alpha \diff s + \Int_{t+h}^T \leftidx{^{\alpha}}w_s (u)^\alpha \diff s} \hat{L}_1) ],
	\end{align*}
	where $a_t$, $b_t$, $\hat{L}_1$, and $w$ are defined as in the proof of Proposition \ref{discussion continuous}. We now expand the expression inside the expectation by $\pm \T(a_{t}+b_{t}\hat{L}_1)$, and analyze the resulting terms individually. Using an argument similar to the one employed to interchange limits in the proof of Proposition \ref{discussion continuous}, we obtain for the first term:
	\begin{align*}
		\lim_{h \to 0} \tfrac{1}{h} \E [\T(a_{t+h}+b_{t+h}\hat{L}_1)-\T(a_{t}+b_{t}\hat{L}_1)] = \tfrac{\partial \E[\T(a_{t}+b_{t}\hat{L}_1)]}{\diff t}.
	\end{align*}
	For the second term, we again apply a similar limit-interchange argument to get:
	\begin{align*}
		\E [\lim_{h \to 0} &\tfrac{\T(a_{t+h}+b_{t+h}\hat{L}_1) -\T(\int_t^{t+h} u_t^\intercal (\mu_s-r) e^{r(T-s)} \diff s + a_{t+h} + \sqrt[\alpha]{\int_t^{t+h} \leftidx{^{\alpha}}w_s (u_t)^\alpha \diff s + \int_{t+h}^T \leftidx{^{\alpha}}w_s (u)^\alpha \diff s} \hat{L}_1)}{a_{t}+b_{t}\hat{L}_1 -(\int_t^{t+h} u_t^\intercal (\mu_s-r) e^{r(T-s)} \diff s + a_{t+h} + \sqrt[\alpha]{\int_t^{t+h} \leftidx{^{\alpha}}w_s (u_t)^\alpha \diff s + \int_{t+h}^T \leftidx{^{\alpha}}w_s (u)^\alpha \diff s} \hat{L}_1)} \\
		&\hspace{30pt}\cdot \tfrac{a_{t}+b_{t}\hat{L}_1 -(\int_t^{t+h} u_t^\intercal (\mu_s-r) e^{r(T-s)} \diff s + a_{t+h} + \sqrt[\alpha]{\int_t^{t+h} \leftidx{^{\alpha}}w_s (u_t)^\alpha \diff s + \int_{t+h}^T \leftidx{^{\alpha}}w_s (u)^\alpha \diff s} \hat{L}_1)}{h}].
	\end{align*}
	By Assumption \ref{assumption expected value derivative}, the first fraction converges to $\T'(a_t+b_t\hat{L}_1)<\infty$. We now decompose the second fraction as follows:
	\begin{align*}
		\tfrac{a_{t}- a_{t+h} -\int_t^{t+h} u_t^\intercal (\mu_s-r) e^{r(T-s)} \diff s}{h} + \tfrac{\sqrt[\alpha]{b_t^\alpha}- \sqrt[\alpha]{\int_t^{t+h} \leftidx{^{\alpha}}w_s (u_t)^\alpha \diff s + \int_{t+h}^T \leftidx{^{\alpha}}w_s (u)^\alpha \diff s}}{h} \hat{L}_1.
	\end{align*}
	The first summand converges to zero by the definition of $a_t$ and the mean value theorem for integrals, since $u$ is continuous. For the second summand, observe that for $\alpha>1$ and $\min\{a,b\}>0$, the inequality $|\sqrt[\alpha]{c_1}-\sqrt[\alpha]{c_2}| \leq \tfrac{1}{\alpha} \min\{c_1,c_2\}^{1/\alpha-1}|c_2-c_1|$ follows directly from a simple application of the mean value theorem. If $c_{min}:=\min\{b_t^\alpha,\int_t^{t+h} \leftidx{^{\alpha}}w_s (u_t)^\alpha \diff s + \int_{t+h}^T \leftidx{^{\alpha}}w_s (u)^\alpha \diff s\} = 0$, then the second summand is clearly equal to zero by the definition of $w$ (see Definition \ref{def m w}). Hence, we may assume that $c_{min}>0$, which yields:
	\begin{align*}
		|\tfrac{\sqrt[\alpha]{b_t^\alpha}- \sqrt[\alpha]{\int_t^{t+h} \leftidx{^{\alpha}}w_s (u_t)^\alpha \diff s + \int_{t+h}^T \leftidx{^{\alpha}}w_s (u)^\alpha \diff s}}{h}| &\leq \tfrac{1}{\alpha} c_{min}^{1/\alpha-1}|b_t^\alpha-(\Int_t^{t+h} \leftidx{^{\alpha}}w_s (u_t)^\alpha \diff s + \Int_{t+h}^T \leftidx{^{\alpha}}w_s (u)^\alpha \diff s)| \\
		&\leq \tfrac{1}{\alpha} c_{min}^{1/\alpha-1} \Int_t^{t+h} |\leftidx{^{\alpha}}w_s (u)^\alpha - \leftidx{^{\alpha}}w_s (u_t)^\alpha| \diff s.
	\end{align*}
	Since $u$ is continuous, this summand also converges to zero, and the claim follows.
\end{myproof}

\begin{myproof}[Proof of Proposition \ref{formula exp target disc}]
	Similar steps as in the proof of Proposition \ref{risk measure formula}, (\ref{eq: umschreibung xt-xn}) lead to
	\begin{align*}
		X_T^{\hat{u}} - X_n^{\hat{u}} e^{r(T-n)} \sim \normal ( \Int_{n}^T {\hat{u}}^\intercal_s (\mu_s-r) e^{r(T-s)} \diff s, \norm{{\hat{u}}_s^\intercal \sigma_s e^{r(T-s)}}^2_{\L^{2}((n,T], \diff s; H_s)} ).
	\end{align*}
	Thus: $e^{-\gamma(X_T^{\hat{u}} - X_n^{\hat{u}} e^{r(T-n)})} \sim \mathcal{LN} ( -\gamma \Int_{n}^T {\hat{u}}^\intercal_s (\mu_s-r) e^{r(T-s)} \diff s, \gamma^2 \norm{{\hat{u}}_s^\intercal \sigma_s  e^{r(T-s)}}^2_{\L^{2}((n,T], \diff s; H_s)} ).$
	Hence, with the formula for the expected value of log-normal distributions, we get:
	\begin{align*}
		\mathbb{E}_{n,x} [\T (X_T^{\hat{u}} - x e^{r(T-n)})] = c^{-1}(1- \exp({-\gamma \Int_{n}^T {\hat{u}}^\intercal_s (\mu_s-r) e^{r(T-s)} \diff s + \frac{\gamma^2}{2} \norm{{\hat{u}}_s^\intercal \sigma_s e^{r(T-s)}}^2_{\L^{2}((n,T], \diff s; H_s)}}) ).
	\end{align*}
	With this and Proposition \ref{risk measure formula}, the functional $J_n$ from Equation (\ref{J def}) becomes:
	\begin{align*}
		J_n (\hat{u}) =& c^{-1}(1-\exp( {-\gamma \Int_n^T {\hat{u}}^\intercal_s (\mu_s - r) e^{r(T-s)}ds + \frac{\gamma^2}{2} \norm{{\hat{u}}_s^\intercal \sigma_s e^{r(T-s)}}^2_{\L^{2}((n,T], \diff s; H_s)}}) ) \\
		&- \lambda_n F ( -\Int_n^T m_s ({\hat{u}}) \diff s + \varrho^{W_1^1} \norm{{\hat{u}}_s^\intercal \sigma_s e^{r(T-{s})}}_{\L^2((n,T],\diff s;H_s)} ).
	\end{align*}
	Taking the partial derivative for $\hat{u}_{n+1}^k$ and setting it equal to $0$ gives the claim.
\end{myproof}

\begin{myproof}[Proof of Proposition \ref{optimal control no risk, disc}]
	We prove this proposition by backward induction. To avoid repetition, we only show the induction step for arbitrary $n$:\\
	We get the optimum by setting the derivative equal to $0$. From Proposition \ref{formula exp target disc}, we get the partial derivative:
	\begin{align} \label{eq: differentiation without risk}
		&0 \overset{!}{=}\textfrac{\partial}{\partial \hat{u}_{n+1}^k} J_n (\hat{u}) = -c^{-1} \exp ({-\gamma \Int_n^{T} {\hat{u}}_s^\intercal (\mu_s-r)e^{r(T-s)} \diff s + \frac{\gamma^2}{2} \norm{{\hat{u}}_s^\intercal \sigma_s e^{r(T-s)}}^2_{\L^2((n,T],\diff s; H_s)}})  \notag\\
		&\hspace{2.9cm}\cdot ( -\gamma \Int_n^{n+1} (\mu_s^k-r) e^{r(T-s)} \diff s +\gamma^2 \Int_n^{n+1} \langle {\hat{u}}_s^\intercal \sigma_s e^{r(T-s)}, \sigma_{s}^{k \cdot} e^{r(T-s)} \rangle_{R_s} \diff s) \notag \\
		\Leftrightarrow\ & \gamma \Int_n^{n+1} (\mu_s^k-r) e^{r(T-s)} \diff s = \gamma^2 \Sum_{i=1}^d {\hat{u}}^i_{n+1} \Int_n^{n+1} \langle \sigma_s^{i \cdot} e^{r(T-s)}, \sigma_{s}^{k \cdot} e^{r(T-s)}\rangle_{R_s} \diff s,
	\end{align}
	because of the assumption that $\hat{u}$ is piecewise constant in each component. Solving \eqref{eq: differentiation without risk} for $\hat{u}_{n+1}^k$ gives us the claim.
\end{myproof}

\begin{myproof}[Proof of Proposition \ref{formula levy exp target disc}]
	As above:
	\begin{align*}
		X_T^{\hat{u}} - X_n^{\hat{u}} e^{r(T-n)} \overset{d}{=} \Int_{n}^T {\hat{u}}^\intercal_s (\mu_s-r) e^{r(T-s)} \diff s + \norm{{\hat{u}}_s^\intercal \sigma_s v e^{r(T-s)}}^{\alpha}_{\L^{\alpha}((n,T] \times \S^d, \diff s \times \tilde{\sigma} (\diff v))} \tilde{L}_1.
	\end{align*}
	Hence, it holds when using the Lagrange function of $\tilde{L}_1$ in the second step (see Janson \cite[p.9]{janson2011stable}):
	\begin{align*}
		\E [ e^{-\gamma (X_T^{\hat{u}} - X_n^{\hat{u}} e^{r(T-n)}) }] &= e^{-\gamma \Int_{n}^T {\hat{u}}^\intercal_s (\mu_s-r) e^{r(T-s)} \diff s} \E [ e^{-\gamma \norm{{\hat{u}}_s^\intercal \sigma_s v e^{r(T-s)}}^{\alpha}_{\L^{\alpha}((n,T] \times \S^d, \diff s \times \tilde{\sigma} (\diff v))} \tilde{L}_1}] \\
		&= e^{-\gamma \Int_{n}^T {\hat{u}}^\intercal_s (\mu_s-r) e^{r(T-s)} \diff s} e^{\frac{c_{\alpha} \gamma^{\alpha}}{-\cos (\frac{\pi \alpha}{2})} \norm{{\hat{u}}_s^\intercal \sigma_s v e^{r(T-s)}}^{\alpha}_{\L^{\alpha}((n,T] \times \S^d, \diff s \times \tilde{\sigma} (\diff v))}}.
	\end{align*}
	Thus, plugging this result into \eqref{J def} for $\T$ being the exponential target function, using Proposition \ref{risk measure formula}, taking the partial derivative for $\hat{u}_{n+1}^k$ and setting it equal to $0$ gives the claim.
\end{myproof}

\begin{myproof}[Proof of Proposition \ref{terminal wealth disc}]
	Let $s$, $\T$, $F$, and $\rho$ as in the description before the proposition. Then $J_n$ given by (\ref{J def}) leads to $
	J_{n} (u) = (\mu-r) \Int_n^T u_s e^{r(T-s)} \diff s - \lambda F( \varrho^{\hat{L}_1} ) \sigma^{\eta} (\Int_n^T |u_s|^\alpha e^{\alpha r(T-s)} \diff s)^{\eta/\alpha}$ where we used \eqref{eq: umschreibung xt-xn} and Proposition \ref{risk measure formula}. Moreover, we dropped the dependency on $x$ in the argumentation due to $J_n$ being independent of $x$. Since $u$ is piecewise constant, we can take the derivative with respect to $u_{n+1}$. Setting this derivative equal to $0$ to find the maximum gives us the claim.
\end{myproof}

\footnotesize
\bibliography{bibliography.bib}
\footnotesize
\bibliographystyle{abbrv}

\end{document}